\documentclass[11pt,letterpaper]{article}
\usepackage[margin=1in]{geometry}
\usepackage{amsmath,amssymb,amsthm}
\usepackage{newtxtext,newtxmath}
\usepackage{bm}
\usepackage{endnotes}
\usepackage{setspace}
\usepackage{microtype}
\usepackage{nicefrac}
\usepackage{graphicx,psfrag,epsf,color}
\usepackage{epsfig}
\usepackage{tikz}
\usetikzlibrary{cd}
\usepackage{multirow}
\usepackage{enumitem}
\usepackage{comment}
\usepackage[ruled]{algorithm2e}
\usepackage{algpseudocode}
\usepackage{mathtools}
\usepackage{url}
\usepackage{array}
\usepackage{booktabs}
\usepackage{float}
\usepackage{subcaption}
\usepackage[font=small,labelfont=bf]{caption}

\usepackage{thmtools}
\DeclareMathOperator*{\argmax}{arg\,max}

\newcommand{\wh}[1]{\widehat{#1}}
\newcommand{\wt}[1]{\widetilde{#1}}
\newcommand{\ind}{1\!\!\mathrm{l}}

\newcommand{\Var}{{\mbox{Var}}}
\newcommand{\Cov}{{\mbox{Cov}}}

\def\EE{\mathbb{E}}

\newcommand\independent{\protect\mathpalette{\protect\independenT}{\perp}}
\newcommand\dependent{\protect\mathpalette{\not\!\perp\!\!\!\!\perp}}
\def\independenT#1#2{\mathrel{\rlap{$#1#2$}\mkern2mu{#1#2}}}

\def\given{\, | \,}

\def\calD{\mathcal {D}}
\def\calF{\mathcal {F}}
\def\calH{\mathcal {H}}
\def\calP{\mathcal {P}}
\def\calU{\mathcal {U}}
\def\calV{\mathcal {V}}
\def\calX{\mathcal {X}}

\newcommand{\norm}[1]{\|#1\|}

\newcommand{\nmF}[1]{\|#1\|_{\mathcal{F}}}
\newcommand{\nmH}[1]{\|#1\|_{\mathcal{H}}}

\newcommand{\RR}{{\mathbb{R}}}

\newcommand{\bI}{{\mathbf{I}}}

\newcommand{\mD}{{\mathcal{D}}}
\newcommand{\mF}{{\mathcal{F}}}

\newcommand{\mH}{{\mathcal{H}}}

\newcommand{\mN}{{\mathcal{N}}}

\newcommand{\mV}{{\mathcal{V}}}
\newcommand{\mX}{{\mathcal{X}}}

\usepackage{natbib}
 \bibpunct[, ]{(}{)}{,}{a}{}{,}

\theoremstyle{plain}
\newtheorem{theorem}{Theorem}
\newtheorem{lemma}{Lemma}
\newtheorem{proposition}{Proposition}
\newtheorem{corollary}{Corollary}

\newtheorem{assumption}{Assumption}
\theoremstyle{definition}
\newtheorem{remark}{Remark}
\newtheorem{example}{Example}

\usepackage[bookmarksnumbered=true]{hyperref}
\hypersetup{
  colorlinks=true,
  linkcolor=blue,
  urlcolor=blue,
  citecolor=blue,
  pdftitle={Personalized Pricing with Invalid Instrumental Variables: Identification, Estimation, and Policy Learning},
  pdfauthor={Rui Miao, Zhengling Qi, Cong Shi, Lin Lin},
  pdfkeywords={endogeneity, instrumental variable, model identification, personalized pricing, policy learning, unmeasured confounder}
}
\usepackage[capitalize,noabbrev]{cleveref}

\def\QED{\hfill \quad{\bf Q.E.D.}\medskip}
\def\EMAIL#1{\href{mailto:#1}{#1}}%

\newcommand{\clip}[3]{\operatorname{clip}_{[#1,#2]}\left(#3\right)}  %

\title{Personalized Pricing with Invalid Instrumental Variables: Identification, Estimation, and Policy Learning}
\author{
  Rui Miao\thanks{Department of Mathematical Sciences, The University of Texas at Dallas. \EMAIL{rui.miao@utdallas.edu}}
  \and Zhengling Qi\thanks{Department of Decision Sciences, GW School of Business, The George Washington University. \EMAIL{qizhengling@email.gwu.edu}}
  \and Cong Shi\thanks{Department of Management, Miami Herbert Business School, University of Miami. \EMAIL{congshi@bus.miami.edu}}
  \and Lin Lin\thanks{Department of Biostatistics and Bioinformatics, Duke University. \EMAIL{lynn.lin@duke.edu}}
}
\date{}

\begin{document}
\maketitle

\begin{abstract}
Pricing based on individual customer characteristics is widely used to maximize sellers' revenues. We study how a firm can learn such a personalized pricing policy from historical transactions when unobserved factors affect both the prices it charged and the revenue it earned. Standard instrumental variable (IV) methods either focus on a discrete treatment space or require that the IV affect revenue only through price. This exclusion restriction is difficult to justify in pricing, where the most natural candidate IVs are price-setting levers that may also affect revenue directly. We propose PRINT (Personalized pRicing using Invalid iNsTrumental variables), a policy learning method for continuous prices that uses such candidate IVs. Relying on structural models of revenue and price, we show that the optimal pricing policy can be identified with invalid IV that shifts higher moments of historical prices, when the unobserved factors that drive customers’ price sensitivity are separated from those used by the historical pricing rule. Based on this identification, which leads to conditional moment restrictions with generalized residual functions, we construct an adversarial min-max estimator that accommodates flexible machine learning models and establish regret bounds for the learned pricing policy. Simulation studies and an application to a U.S. online auto lender illustrate the method's performance. Our results give a principled way to learn personalized prices from observational data using imperfect price-setting levers, and clarify when doing so is justified. The reproducible code is available at \url{https://github.com/NAIVlab/PRINT}.
\end{abstract}

\noindent\textbf{Keywords:} endogeneity, instrumental variable, model identification, personalized pricing, policy learning, unmeasured confounder

\clearpage

\section{Introduction}

The growth of transaction-level customer and product data creates new opportunities for personalized pricing. Historical records can be used to estimate how revenue varies with price and observed characteristics and to learn a pricing rule that maps those characteristics to a feasible price.
Formally, this is an offline policy learning problem in which the treatment, price, is continuous. The offline data often consist of customer/product information, the offered price, and the resulting revenue. The firm's goal is to use these data to learn a pricing policy that maps $X$, the observed covariates of the product and the customer, to a price that maximizes expected revenue.
In practice, converting historical transactions into a deployable pricing rule helps keep discounts consistent, margins steady across regions, and price responses monitored as conditions shift, without running new price experiments, which are often costly or impractical.

Because historical prices were set by the firm, often using information the analyst does not observe, omitted variables may affect both price and revenue and thereby create endogeneity. Without additional structure, this confounding prevents identification of the revenue-maximizing policy from observational data, and standard policy-learning methods can be biased. Indeed, observational price variation can fail to reproduce experimentally measured price elasticities \citep{bray2024observational}.

In the literature on causal inference and econometrics, an instrumental variable
(IV) is widely used to account for unmeasured
confounding in identifying the (local) causal effect of treatment. This task is closely
related to personalized pricing because evaluating a pricing strategy requires
estimating the revenue that would result under an intervention on price. Therefore, an IV model is
a potential approach to overcome endogeneity in learning an optimal
pricing strategy. A valid IV is a pretreatment variable independent of all
unobserved covariates given the observed ones, and only affects the outcome through the treatment (the so-called exclusion restriction). Meanwhile, the IV must be associated with the treatment conditionally on the covariates, observed and unobserved (the relevance condition).
Classic examples include the quarter of birth as an IV for schooling in estimating returns to education \citep{angrist1991does} and, closer to our setting, prices in other geographic markets and rival product characteristics as IVs for own price in demand estimation \citep{hausman1996valuation, nevo2001measuring, berry1995automobile}.

However, the standard IV approach has limited applicability in pricing because the most natural candidate IVs are themselves price-setting levers, such as rate tables or APR-related levers, that may affect revenue directly. To allow such direct effects of the instrument on revenue, in exchange, we replace the exclusion restriction with a different identifying restriction: unobserved heterogeneity that modifies marginal price response must be separated from the conditional price moments generated by the historical price rule. This partial-information orthogonality condition (Assumption \ref{ass: identification}\ref{ass:Orthogonality condition}) is plausible when historical prices are driven by business rules, operational frictions, or levers that do not fully reveal the latent demand primitives determining individual price sensitivity. 

Under this target class of data-generating processes, we identify the two policy-relevant functions entering the optimal personalized price and develop PRINT, a conditional moment policy learning method for continuous prices. PRINT estimates these functions with an adversarial min-max estimator that accommodates flexible machine learning models. We further establish regret guarantees and evaluate the method in simulations and an auto-loan application. Practically, this provides a disciplined way to use candidate imperfect price-setting instruments, such as rate tables or APR-related levers. %

\subsection{Motivating Example}
A motivating example is the pricing problem for auto loans. Consumer lending is a prominent industry in which personalized prices (i.e., lending rates) are both socially acceptable and in current practice, albeit at varying degrees of granularity \citep{ban2021personalized}. The norm of price negotiation, high variation in customer willingness to pay, low cost of bargaining, and other considerations provide tremendous opportunities to offer personalized lending rates for customers in order to maximize profit \citep{phillips2015effectiveness}.

In this example, the offline data consist of customer information (e.g., loan term, loan amount, annual percentage rate (APR), monthly London interbank offered rate (LIBOR), FICO score, and state of residence), the offered loan price constructed from the payment schedule, and the resulting revenue determined by whether the loan was accepted. However, when recording the information of past deals in the offline data, some local information, such as the operational costs of the lender, the competitors' rates on individual deals, and some unknown customer demographics, may not be available, which potentially hinders the decision maker from designing an optimal pricing strategy based on the available covariate information. Thus, it is essential to devise an offline learning algorithm for personalized pricing under endogeneity, which learns the pricing strategy based on available covariates.

In this setting, the APR or related pricing levers may be strongly associated with the constructed loan price. We also note that interest rate variation has long served as identifying variation in consumer-credit studies, whether as a randomized offer rate \citep{karlan2008credit} or as an instrumented endogenous loan price \citep{attanasio2008credit}. However, APR-type variation is not a classical valid instrument if it also affects contract acceptance and revenue directly. We therefore treat APR-type variation as a candidate invalid price-setting IV. The key question is not whether APR satisfies exclusion, but whether the unobserved components that modify marginal price response are separated from the conditional moments of the historical pricing rule after conditioning on observed borrower and loan covariates. This maintained condition is substantive and should be assessed through institutional knowledge of how the lender sets prices.

\subsection{Major Contributions}

\begin{enumerate}[label=(\alph*)]

\item \textbf{Policy learning with invalid instruments.}
We develop PRINT for offline personalized pricing with a continuous treatment when candidate instruments may directly affect revenue. Rather than imposing the exclusion restriction, we identify price-response functions under a partial-information orthogonality condition (Assumption \ref{ass: identification}\ref{ass:Orthogonality condition}) that separates unobserved marginal price-response heterogeneity from the conditional moments of the historical price rule.

\item  \textbf{Identification under continuous treatments.}
We generalize the causal inference literature on treatment effect
estimation under unmeasured confounding. In particular, the existing literature
is predominantly focused on discrete treatment settings, where various
approaches using a valid IV are developed. Much less attention has been paid to
dealing with continuous treatment. While there is also a stream of recent
literature studying causal effects under invalid IVs
\citep{guo2018confidence,tchetgen2021genius,
  guo2023causal, bruneel2024identification, lin2024instrumental,
  kang2025identification, xie2024identification, ye2024genius}, 
most of this work still concentrates on discrete treatments. Our identification exploits instrument-induced heteroskedasticity in prices, in the spirit of \citet{rigobon2003identification} and \citet{lewbel2012using}, but allows the price effect to vary across customers. To our knowledge, our identification result for heterogeneous effects of a continuous treatment with a candidate invalid IV fills a gap in the causal inference literature,
particularly concerning causal identification with heterogeneity (i.e.,
effects varying across individuals) which could be of independent interest.

\item \textbf{Estimation and regret guarantees.}
We establish an asymptotic regret guarantee for our policy learning
algorithm in finding an optimal pricing decision, based on a newly developed
adversarial min-max estimator for solving conditional moment restrictions with
generalized non-linear residual functions. This adversarial min-max estimator
is motivated by solving a zero-sum game that can incorporate flexible machine
learning models.
\end{enumerate}

\subsection{Related Work} \label{sec: related work}

Since the seminal work of \citet{manski2004statistical}, there has been a surge
of interest in studying offline policy learning in economics, statistics, and
computer science communities such as
\citet{qian2011performance,dudik2011doubly,zhao2012estimating,chen2016personalized,kitagawa2018should,cai2021jump,biggs2022convex,qi2022offline,bertsimas2023power,fang2025improving,ai2026data},
among many others. In revenue management, \citet{ban2021personalized} study personalized dynamic pricing with high-dimensional features and heterogeneous price elasticities, and \citet{biggs2022convex} develops convex surrogate loss functions for contextual pricing with transaction data. However, most existing works rely on the unconfoundedness
assumption, which is hardly satisfied in practice. \citet{bertsimas2023power}
discuss the limitations of using observational data in pricing. To mitigate the effect of possible endogeneity, practitioners often collect and adjust for as many covariates as possible. While this may be the best approach, it is often very costly and even infeasible, as we have no control over offline data collection. To address this limitation, more recently, various policy learning methods under unmeasured confounding have been proposed, such as using a binary and valid IV for a point or partial identification of the optimal policy \citep{cui2021semiparametric,qiu2021optimal,han2019optimal,pu2020estimating,stensrud2024optimal}, using a sensitivity model for policy improvement \citep{kallus2018confounding}, and leveraging the proximal causal inference \citep{qi2022proximal,miao2022off,wang2026blessing,shen2023optimal}. To our knowledge, none of the aforementioned works studies continuous price policy learning with candidate instruments that may directly affect revenue.

Price endogeneity is a long-standing concern in empirical marketing and economics. Common remedies include cost-based or cross-market instruments \citep{berry1995automobile,hausman1996valuation,nevo2001measuring}, control functions \citep{petrin2010control}, and instrument-free corrections based on copulas \citep{park2012handling} or heteroskedasticity \citep{rigobon2003identification,lewbel2012using}. \citet{rossi2014even} critically examines the validity and strength of instruments used in such applications. These methods mainly target homogeneous or parametrically specified demand parameters, rather than a personalized pricing policy over a continuous price, which is the focus of this paper.

Our work is also closely related to IV models, which have been extensively
studied in the literature on causal inference and econometrics. See
\citet{angrist1995identification,ai2003efficient,newey2003instrumental,hall2005nonparametric,chen2011rate,chen2018optimal,darolles2011nonparametric,blundell2007semi,wang2018bounded}
for earlier references. A typical assumption in the aforementioned literature is
the existence of a valid IV that satisfies (i) independence from all unobserved
covariates $U$, (ii) the exclusion restriction that prohibits the direct effect
of IVs on the outcome, and (iii) correlation with the endogenous variable.
Tremendous efforts have been made to develop statistical and econometric methods
to account for possible violation of these assumptions. See
\citet{staiger1994instrumental,stock2000gmm,stock2002survey,chao2005consistent,newey2009generalized}
and many others for relaxing (iii), and
\citet{lewbel2012using,kang2016instrumental,guo2018confidence,tchetgen2021genius}
for relaxing (i) and (ii). In particular,
\citet{kang2016instrumental,guo2018confidence,kolesar2015identification,windmeijer2019use}
considered the multiple IV setting and restricted their methods to some specific parametric
models. In contrast, \citet{lewbel2012using,tchetgen2021genius}, which are
closely related to our proposal, either focus on discrete treatments or restrict continuous treatments to constant causal effects, which rules out personalized pricing.

\section{Problem Formulation and Challenges}
\label{sec: causal framework}
We introduce the problem of personalized pricing under the causal inference framework. We also illustrate the challenges of finding an optimal pricing strategy in observational studies due to endogeneity.

\subsection{Personalized Pricing without Unmeasured Confounding}
Let $P$ be the price of a product that takes values in a \textit{known} and \textit{continuous} action space $\calP = [p_1, p_2]$ with $0 \leq p_1 < p_2$. Define the potential revenue under the intervention $P=p$ as $Y(p)$ for $p \in \calP$. In the counterfactual world, $Y(p)$ is a random variable representing the revenue the company would have earned had it set the price at $p$ for its product. Denote by $X$ the observed $q$-dimensional covariate associated with the product and the customer, which belongs to the covariate space~$\calX \subseteq \mathbb{R}^q$. A personalized pricing policy $\pi$ is determined by the covariate $X$ and is a measurable function mapping from the covariate space~$\calX$ into the action space $\calP$. For example, a lender's policy may quote different prices to applicants with different credit scores, loan terms, or vehicle types.
We use a jointly measurable potential-revenue process $Y(\cdot)$ with an integrable envelope. The precise path-space regularity is in Appendix~\ref{app: process bridge}. The conditional mean given $X$ is continuous in price, but realized revenue paths are allowed to have jumps. This accommodates purchase or acceptance decisions: a customer who accepts at one price may decline a slightly higher quote, so realized revenue drops to zero, while expected revenue still varies smoothly with price. When an instrument $G$ is introduced, $Y(p)$ denotes a price-only intervention that leaves the pre-price $G$ and the reference law of $(U,G)\mid X$ unchanged.

The potential revenue under a policy $\pi$ is then defined as $Y(\pi) \triangleq Y(\pi(X))$. For a pricing strategy $\pi$, we evaluate its performance using the expected revenue 
\begin{align*}
	\calV(\pi) = \EE\left[Y(\pi)\right].
\end{align*}
The goal of personalized pricing is to find an optimal policy $\pi^\ast$, such that
\begin{align}\label{def: optimal policy 2}
    \pi^\ast \in \argmax_{\pi \in \Pi} \calV(\pi),
\end{align}
where $\Pi$ is a class of  policies depending on $X$.
However, since for each instance of the random tuple $(X, P, Y)$, we only observe one $Y$ corresponding to the price $P$, but not other values of $Y$ and $P$, it is impossible to learn the joint distribution of $(X, P, \{Y(p)\}_{p \in \calP})$ in general. For instance, for an applicant who accepted a loan priced at \$3,000, we never observe whether the applicant would have accepted the same loan at \$3,500. Therefore, it is crucial to derive
identification conditions  for learning $\calV(\pi)$ from the observed data. We first impose the following standard causal assumptions.

\begin{assumption}[Standard Causal Assumptions] \label{ass: standard} The following conditions hold.
\begin{enumerate}[label=(\alph*)]
\item (Consistency)\label{ass: consistency}
$Y = Y(P)$ almost surely.
\item (Positivity)\label{ass: positivity}
Conditionally on $X=x$, the price $P$ admits a density $f(\cdot \given x)$ with respect to Lebesgue measure on $\calP$, with $f(p\given x)>0$ for almost every $p$ and $x$.
\item (No unmeasured confounding)\label{ass: unconfoundedness}
$Y(\cdot) \independent P \given X$, where conditional independence concerns the whole potential-revenue process specified in Appendix~\ref{app: process bridge}.
\end{enumerate}
\end{assumption}

Assumption \ref{ass: standard}\ref{ass: consistency} ensures that the observed $Y$ matches the potential revenue under the intervention~$P$. Assumption~\ref{ass: standard}\ref{ass: positivity} guarantees full conditional support across the action space. Assumption~\ref{ass: standard}\ref{ass: unconfoundedness} requires the entire potential-outcome path to be conditionally independent of the historical price given $X$. Intuitively, positivity requires that customers with similar recorded characteristics were historically quoted a range of prices, and no unmeasured confounding requires that those prices did not depend on anything else that also drives revenue, as in a randomized price test or a pricing rule based only on recorded characteristics. Under Assumption~\ref{ass: standard} and the standing process regularity in Appendix~\ref{app: process bridge}, for each $\pi\in\Pi$,
\[
\calV(\pi) = \EE[Q(X, \pi(X))],
\]
where $Q(x,p)=\EE[Y(p)\given X=x]$ is the expected revenue from charging price $p$ given covariates $x$, and the expectation is taken over $X$.
Assumption~\ref{ass: standard} enables us to identify $Q(x,p)$ from the observed conditional mean $\EE[Y\given X=x,P=p]$ for almost every feasible price $p$ and $x$. Appendix~\ref{app: process bridge} justifies evaluating $Q$ at the policy price $\pi(X)$ to obtain the value $\calV(\pi)$.
Hence $\calV(\pi)$ is nonparametrically identified by the observed data, and $\pi^\ast$ satisfies
    $\pi^\ast \in \argmax_{\pi \in \Pi} \EE[Q(X, \pi(X))]$.
When $\Pi$ contains all measurable policies, an equivalent pointwise characterization is
    $\pi^\ast(X) \in \argmax_{p \in \calP} Q(X, p)$,
almost surely, where the maximum is attained by continuity of $Q(x,\cdot)$ on the compact $\calP$. In words, absent confounding, the firm can trace out an expected-revenue curve for each customer profile by comparing similar customers who were charged different prices, and then charge each profile the price at the peak of its curve.

\subsection{Challenges due to Endogeneity}

In practice, the unconfoundedness condition in Assumption~\ref{ass:
  standard}\ref{ass: unconfoundedness} cannot be ensured without further
restriction on the data-generating procedure, such as a randomized
experiment. Failure of Assumption~\ref{ass: standard}\ref{ass: unconfoundedness} can make the policy value $\calV(\pi)$ unidentified, producing biased value estimates and a suboptimal policy. 
\begin{example}[Bias from unmeasured confounding]\label{ex:confounding}
Consider $Y=(4+X+U)P-P^2+\varepsilon$, where $X \sim \mathrm{Uniform}\{0,1\}$ is observed, $U \sim \mathrm{Uniform}\{-1,1\}$ is unobserved and independent of $X$, $P \in [1, 3]$, and $\varepsilon$ is a mean-zero noise independent of $(X,U,P)$, so the demand
$D = 4+X+U-P$. In the historical data,
the firm charges $P = 3$ when $U = 1$ and $P = 1$ when
$U = -1$, regardless of $X$, so $P$ and $U$ are positively
correlated but $U$ is unobserved by the analyst. The observed prices and associated expected revenues $(P, \EE[Y\given X,U,P])$ are
$(3, 6)$ and $(1, 2)$ when $X = 0$, and $(3, 9)$ and $(1, 3)$ when $X = 1$,
leading a naive method---one that recommends the observed price with the higher conditional mean revenue---to the policy $\widetilde \pi(X) \equiv 3$,
whose true expected value is
$\calV(\widetilde \pi) = \EE[(4 + X + U) \times 3 - 9] = 4.5$. However, the best
personalized policy $\pi^*(X) = 2\cdot\ind(X=0) + 2.5\cdot\ind(X=1)$ achieves
$\calV(\pi^*)= 5.125$, so the unmeasured confounder $U$ biases value
identification and yields a suboptimal policy. Indeed, the mechanism is common in practice. When a firm charges more precisely when demand is strong, high prices coincide with high revenue, and the data overstate how much revenue a high price earns.

\end{example}

In the causal inference literature, IV models are widely used to account for unmeasured confounding when identifying the average treatment effect \citep[e.g.,][]{angrist1996identification}. It is often assumed that there exists an IV, denoted by $G$, such that
\begin{align*}
    P = \Lambda(X, U, G) + \widetilde \varepsilon, \qquad \EE[\widetilde \varepsilon \given X, U, G] = 0,
\end{align*}
so that $\Lambda$ is the conditional mean of the price given $(X,U,G)$, foreshadowing the price model \eqref{Model: action} below. A valid IV is required to satisfy the following conditions.
\begin{assumption}\label{ass: valid IV}
    \begin{enumerate}[label=(\alph*)]
        \item \label{ass: IV relevance} (IV relevance) $G \dependent P \given (U, X)$.
        \item \label{ass: IV independence} (IV independence) $G \independent U \given X$.
        \item \label{ass: IV Exclusion restriction} (IV exclusion restriction) $G \independent Y \given (P, X, U)$.
    \end{enumerate}
\end{assumption}
Assumption~\ref{ass: valid IV}\ref{ass: IV Exclusion restriction} states the exclusion restriction as a conditional independence in the observed-data law. Since our framework defines potential outcomes with respect to the price only (i.e., $Y(p)$ rather than $Y(p,g)$), we adopt this distributional formulation as a surrogate for the counterfactual exclusion restriction of \citet{angrist1996identification}, which is precisely the form whose failure is exhibited in Example~\ref{ex:running} below.
In pricing terms, a valid instrument resembles a cost shock that moves the price (relevance), is unrelated to customers' unobserved willingness to pay (independence), and affects revenue only by changing the price (exclusion). A classical IV approach would use variation in a variable $G$ that shifts price but has no direct effect on revenue $Y$. In many pricing settings this exclusion restriction is difficult to satisfy because candidate price-setting levers can directly affect demand, contract acceptance, or operational revenue.

\begin{example}[Partial-information auto-loan pricing]\label{ex:running}
An auto lender observes borrower and loan covariates $X$ and posts a loan price $P$. The offered APR, denoted $G$, strongly shifts $P$ but may also affect demand directly, so it is a candidate invalid IV. Let $U=(U^a,U^p)$, where the historical pricing rule uses the operational component $U^p$ but not the demand component $U^a$. Consider
\begin{align*}
P &= \gamma_0 + \gamma_1 G + \gamma_2^{\top}X + \gamma_3 U^p + \epsilon_p,\\
D &= \eta_0 + \eta_1 P + \eta_2 G + \eta_3^{\top}X + \eta_4 U^a + \epsilon_d, \qquad \eta_1 < 0,\\
Y &= D\cdot P + \epsilon_y,
\end{align*}
where the noise terms are mean zero under the conditions detailed in Appendix Section~\ref{app: running example details}. When $\gamma_1\neq0$ and $\eta_2\neq0$, the APR shifts the historical price and also directly shifts demand, violating exclusion. Yet if $U^a\independent U^p\mid X$ and $G\independent U\mid X$, the historical price uses a different latent component from the one determining marginal demand. In words, the lender's historical prices react to operational conditions that the analyst does not observe, such as funding costs or underwriting capacity, but not to borrowers' unobserved price sensitivity. This partial-information separation motivates Assumption~\ref{ass: identification}. Appendix Section~\ref{app: assumption verification} gives the formal model mapping, relevance conditions, and verification.
\end{example}

This example motivates an invalid IV approach for settings in which the exclusion restriction is not credible. Furthermore, when $P$ is continuous, existing literature
\citep[e.g.,][]{ai2003efficient,newey2003instrumental} often assumes an
additively separable structural model with a valid instrument. Identification in these frameworks hinges on the exclusion restriction. While subsequent research has sought to relax structural assumptions, for instance, by accommodating non-separable models via restrictions on unobserved heterogeneity $U$ \citep{chernozhukov2007instrumental} or by generalizing the revenue model conditions \citep{qiu2021optimal, han2019optimal, pu2020estimating}, these advances  maintain the exclusion restriction in Assumption \ref{ass: valid IV}. Our target setting differs because candidate price-setting instruments may directly affect revenue. This distinction is relevant to the auto-loan application in Section~\ref{sec: numerical}, where APR is strongly related to the constructed price and may also affect contract acceptance and revenue directly.

Given these challenges, in the following section, we consider invalid IVs and develop a novel identification for an optimal pricing strategy from the observational data.

\section{Assumptions and Identification}
\label{sec: policy learning with invalid iv}
We present an identification result using a candidate invalid IV, denoted $G\in\RR$, for the policy $\pi^\ast$ in \eqref{def: optimal policy 2}. The result is based on price and revenue models together with restrictions on how the instrument and latent variables affect price and revenue.
Our setting and assumptions are not meant to cover every form of endogenous pricing. In particular, we exclude fully informed pricing rules in which the same latent demand primitive determines both marginal revenue and the historical price. The target class is partial-information offline pricing, where historical prices may be confounded with baseline revenue, and the instrument may directly affect revenue, but unobserved heterogeneity in marginal price response is separated from the conditional moments of the price-assignment rule.

\subsection{Model Assumptions}
We first introduce our revenue and price models in the presence of unmeasured confounders $U \in \calU$, which can be multi-dimensional. We assume the following structural conditional-mean revenue model and historical price model for our data-generating process of a random tuple $(X, U, G, P, Y)$:
{
\begin{align}
		\EE\left[Y(p) \given G, X, U \right] &= [\beta_{p, 1}(U, X) + \beta_{p,g}(X,G)] p + \beta_{p, 2}(U, X) p^2 +\beta_g(U, X, G)+ \beta_{u, x}(U, X), \label{Model: counterfactual}\\
		\EE\left[P \given G, X, U \right] &= \alpha_g(X, G)+ \alpha_{u, x}(U, X)\label{Model: action}.
\end{align}
}
The revenue model \eqref{Model: counterfactual} is quadratic in price, as arises from linear demand \citep{bertsimas2006dynamic,winkelmann2008econometric,bastani2022meta}. Its policy-relevant objects are the conditional averages $\EE[\beta_{p,1}(U,X)+\beta_{p,g}(X,G)\given X]$ and $\EE[\beta_{p,2}(U,X)\given X]$, which determine the optimal price. By contrast, the generics effect of $(X, U, G)$ on revenue~$Y$, captured by $\beta_g(U,X,G)$, and $\beta_{u,x}(U,X)$ are price-independent nuisance components and need not be identified for policy optimization. The model permits $G$ to affect revenue directly and to interact with the linear price term through $\beta_{p,g}(X,G)P$, relaxing the exclusion and additivity restrictions common in standard IV models in causal inference and econometrics literature \citep[e.g.,][]{wright1928tariff, goldberger1972structural, bowden1990instrumental,small2007sensitivity,kang2016instrumental,hernan2006instruments}. For tractability, it does not allow $G$ to modify the quadratic coefficient. While our model assumes a quadratic effect of price on revenue, implying a linear demand function, this functional form should be interpreted as a policy-relevant approximation rather than a universal demand model. 

The historical price model \eqref{Model: action} decomposes the conditional mean of the historical price into an instrument-dependent component $\alpha_g(X,G)$ and a latent-variable component $\alpha_{u,x}(U,X)$. Thus, the way candidate instrument $G$ shifts price does not vary with the unobserved state $U$. Together with the IV independence condition $G\independent U\mid X$ introduced in Assumption~\ref{ass: identification} below, this separability isolates instrument-driven price variation. For flexibility, both components can accommodate nonparametric historical pricing rules, and it is unnecessary to identify them as they are not affecting the optimal policy. However, the model excludes historical rules whose response to $G$ depends on latent information.

Due to unmeasured confounders $U$, in our offline data, for each instance, we only observe a sample of the random tuple $(X, G, P, Y)$. Under this model setup, our goal is to find the optimal personalized pricing strategy that maximizes the expected revenue.
The policy depends on $X$ only and $G$ is only used for identification, while its natural conditional distribution is averaged over in the price-only policy value.
Then we have the following proposition.

\begin{proposition}\label{prop: optimal policy}
Under the counterfactual revenue model \eqref{Model: counterfactual}, and the process regularity of Appendix~\ref{app: process bridge}, with $\beta_{p,2}(U,X)\leq -c$ almost surely, if $\Pi$ contains all measurable policies from $\calX$ to $\calP$, then the optimal policy $\pi^\ast$ is
\begin{equation*}
\begin{aligned}
	\pi^\ast(X)
	& = \argmax_{p \in \calP} \left\{ \,  \EE\left[\beta_{p, 1}(U, X) + \beta_{p,g}(X,G) \given X\right]p +    \EE\left[\beta_{p, 2}(U, X) \given X\right] p^2 \,  \right\}.
\end{aligned}
\end{equation*}
In particular, writing $\clip{a}{b}{x}:=\min\{\max\{x,a\},b\}$ for given $p_1$ and $p_2$,
\begin{equation}\label{eqn: model-based optimal policy formulation}
	\pi^\ast(X) = \clip{p_1}{p_2}{- \EE\left[\beta_{p, 1}(U, X)  + \beta_{p,g}(X,G)\given X\right] / \EE\left[2\beta_{p, 2}(U, X) \given X\right]}, \text{ almost surely.}
\end{equation}
\end{proposition}
With a bit notation overload, let $\beta_{p, 1}(X) = \EE\left[\beta_{p, 1}(U, X) + \beta_{p,g}(X,G) \given X\right]$ and $\beta_{p, 2}(X) = \EE\left[\beta_{p, 2}(U, X) \given X\right]$. Since $U$ is not observed in our data, without any assumptions, $\beta_{p, 1}(X)$ and $\beta_{p, 2}(X)$ cannot be uniquely identified from the observed data nonparametrically. The non-identification issue indicates that there may exist two different expected revenues under the distribution of the observed random tuple $(X, G, P, Y)$. Directly regressing $Y$ on $(X, G, P, P^2)$ can lead to biased estimation of $\beta_{p, 1}(X)$ and $\beta_{p, 2}(X)$, and thus the resulting estimated policy is sub-optimal.  Example~\ref{ex:confounding} illustrates the source of this bias.
Thus, policy learning using Proposition~\ref{prop: optimal policy} requires us to identify $\beta_{p,1}(X)$ and $\beta_{p,2}(X)$ with further assumptions.

\subsection{Identification Assumptions}

Of the conditions in Assumption~\ref{ass: standard}, this section and the subsequent policy learning theory retain only consistency (Assumption~\ref{ass: standard}\ref{ass: consistency}). Assumption~\ref{ass: standard}\ref{ass: unconfoundedness} is replaced below by its latent information analogue, and positivity is not imposed for the invalid IV identification argument.
For simplicity, we assume that $\beta_{p, 2}(U, X) \leq -c$ almost surely for some constant $c >0$. Throughout, we also impose the integrability condition that all displayed conditional moments exist and are finite. In particular, the coefficient functions $\beta_{p,1}$, $\beta_{p,2}$, $\beta_g$, $\beta_{u,x}$, $\alpha_g$, and $\alpha_{u,x}$ are square-integrable conditioning on $(X,G)$, $\EE[|P|^3+|PY|\given G,X,U]<\infty$ almost surely, and $\EE[|G|^{2(K+3)}\given X]<\infty$, where $K$ is given in  Assumption~\ref{ass: beta_pg basis} below.
\begin{assumption}[Latent process conditional independence]\label{ass: process independence}
The potential-revenue process and the historical price satisfy
\[
Y(\cdot)\independent P\given(G,X,U).
\]
\end{assumption}
This restriction allows endogenous prices through $U$ but rules out assignment based on remaining information in the revenue path given $(G,X,U)$. Together with consistency and the regularity in Appendix~\ref{app: process bridge}, it converts Model~\eqref{Model: counterfactual} into the observed conditional mean by evaluating its jointly specified version at $p=P$.

We further impose the following identification assumptions on $\beta_{p, 1}(X)$ and $\beta_{p, 2}(X)$ by leveraging a candidate invalid price-setting instrument $G$.
\begin{assumption}\label{ass: identification}
	The following statements hold.
	\begin{enumerate}[label=(\alph*)]
		\item \label{ass:iv relavance} (IV relevance)  $G \dependent  P \given (X,U)$.
		\item \label{ass:iv independent} (IV independence) $G \independent U \given X$.
			\item \label{ass:Orthogonality condition} (Partial-information orthogonality) The following conditions hold almost surely: \begin{align}
			\Cov\big(\beta_g(U, X, G), \EE\left[P \given G, U, X \right] \given X, G\big) &= 0,\label{ass: ortho_alpha}\\
			\Cov\big(\beta_{p, 1}(U, X), \EE[P^k \given G, U, X ] \given X, G\big) &= 0,\quad k=1,2,\label{ass: ortho_p1}\\
			\Cov\big(\beta_{p, 2}(U, X), \EE[P^k \given G, U, X ] \given X, G\big) &= 0,\quad k=2,3.\label{ass: ortho_p2}
		    \end{align}
	\end{enumerate}
\end{assumption}

Assumption~\ref{ass: identification}\ref{ass:iv relavance} is distributional relevance of the candidate IV. It is not the algebraic relevance condition that solves the identifying system and is insufficient for identification. Actually, a pure conditional mean shift can satisfy it while being removed by residualization. The operative relevance condition is the full-column-rank condition in Theorem~\ref{thm: identification} below, or its nonzero-determinant specialization in Corollary~\ref{cor: identification 2}.
Assumption~\ref{ass: identification}\ref{ass:iv independent} separates IV from the latent state after conditioning on $X$ as one of the valid IV conditions. 

Assumption~\ref{ass: identification}\ref{ass:Orthogonality condition} is the core restriction. It requires that unobserved components modifying the marginal price-response coefficients $\beta_{p,1}$ and $\beta_{p,2}$ be orthogonal to the conditional price moments generated by the historical pricing rule in both Models \eqref{Model: counterfactual} and \eqref{Model: action}. The condition allows direct instrument effects on revenue through $\beta_g(U,X,G)$ and allows baseline confounding through $\beta_{u,x}(U,X)$ and the price rule. It does not permit the same latent demand primitive to both determine marginal revenue and be used by the historical pricing rule. Thus, the Assumption ~\ref{ass: identification}\ref{ass:Orthogonality condition} is appropriate for partial-information or rule based pricing systems, but not for fully informed optimal monopoly pricing based on the same unobserved demand state. %

Assumption~\ref{ass: identification} is verified for Example~\ref{ex:running} in Appendix Section~\ref{app: assumption verification}.

\subsection{Identification Results}
By the aforementioned assumption, we establish our identification results for $ \beta_{p, 1}(X)$ and $\beta_{p, 2}(X)$ using our offline data. The identification result below recovers the policy-relevant functions $\beta_{p,1}(X)$ and $\beta_{p,2}(X)$ in Proposition~\ref{prop: optimal policy}. This relies on the following two lemmas and Assumption \ref{ass: beta_pg basis}.
The proofs of Proposition~\ref{prop: optimal policy}, Lemmas \ref{lm: identification eq 1} and \ref{lm: identification eq 2}, Theorem~\ref{thm: identification}, and Corollary~\ref{cor: identification 2} can be found in Section
\ref{proof: identification} of the Appendix.
\begin{lemma}\label{lm: identification eq 1}
	Under Models \eqref{Model: counterfactual} and \eqref{Model: action}, Assumption~\ref{ass: process independence}, and consistency, and under Assumption \ref{ass: identification}, we have that, almost surely,

\begin{equation}
\label{eq: lemma1}
\begin{aligned}
	&\EE\left[(G - \EE\left[G \given X\right])(P - \EE\left[P \given G, X \right])Y\given  X \right]\\
	&=  \EE\left[(G - \EE\left[G \given X\right])(P- \EE\left[P \given G, X \right])P \given X\right] \times \EE\left[\beta_{p, 1}(U,X) \given X\right] \\
	&\quad	+ \EE\left[(G - \EE\left[G \given X\right])(P - \EE\left[P \given G, X \right])P^2 \given X\right] \times \beta_{p, 2}(X)\\
    &\quad + \EE\left[(G - \EE\left[G\given X\right])(P - \EE\left[P \given X,G\right])P\beta_{p, g}(X,G)\given X\right].
\end{aligned}
\end{equation}
\end{lemma}

\begin{lemma}\label{lm: identification eq 2}
	Under Models \eqref{Model: counterfactual} and \eqref{Model: action}, Assumption~\ref{ass: process independence}, and consistency, and under Assumption~\ref{ass: identification}, for any positive integer $k$ such that the moments below are finite (only $k \leq K+2$ is used in the sequel), we have, almost surely,
\begin{equation}
\label{eq: lemma2}
    \begin{aligned}
    & \EE\left[G^k(G - \EE\left[G \given X\right])(P - \EE\left[P \given G, X \right])Y\given  X \right] \\
		&\ = \Cov(G^k(G - \EE\left[G\given X \right]), P(P - \EE\left[P \given X,G \right])\given X)\times \EE\left[\beta_{p,1}(U,X)\given X\right] \\
	  &\quad + \Cov(G^k(G - \EE\left[G\given X \right]), P^2(P - \EE\left[P \given X,G \right])\given X)\times \beta_{p, 2}(X)\\
		 &\quad + \EE\left[G^k(G - \EE\left[G \given X\right])\given X\right] \times \EE\left[(P - \EE\left[P \given X,G\right])Y\given X\right]\\
    &\quad + \Cov\left(G^k(G - \EE\left[G\given X\right]), (P - \EE\left[P \given X,G\right])P\beta_{p, g}(X,G)\given X\right).
    \end{aligned}
\end{equation}
\end{lemma}

\begin{assumption}
    \label{ass: beta_pg basis} Let $\mathcal{G}_0 \subseteq \RR$ denote the support of $G$. There exist known, linearly independent basis functions $B_1(x,g), \dots, B_K(x,g)$ on $\mX\times\mathcal{G}_0$ such that $\beta_{p,g} = \sum_{j=1}^{K}\theta_j B_j$ for some $(\theta_1,\dots,\theta_K)^\top \in \RR^K$.
\end{assumption}
Assumption \ref{ass: beta_pg basis} provides a finite-dimensional model for the structured price-instrument interaction $\beta_{p,g}(X,G)$. As a direct result of Lemmas \ref{lm: identification eq 1} and \ref{lm: identification eq 2} and Assumption \ref{ass: beta_pg basis}, we have the following theorem for identifying $\beta_{p, 1}(X)$ and $\beta_{p, 2}(X)$.
\begin{theorem}\label{thm: identification}
Suppose that Models \eqref{Model: counterfactual} and \eqref{Model: action}, Assumption~\ref{ass: process independence}, and consistency hold. If Assumption~\ref{ass: identification} and Assumption \ref{ass: beta_pg basis} hold with a fixed $K$, and the $(K+3)\times(K+2)$ coefficient matrix of the linear system formed by \eqref{eq: lemma1} and \eqref{eq: lemma2} for $k=1,\dots,K+2$, whose columns multiply $\EE[\beta_{p,1}(U,X)\mid X]$, $\beta_{p,2}(X)$ and the basis coefficients $(\theta_1,\dots,\theta_K)$, has full column rank $K+2$ almost surely, then the system uniquely determines $\beta_{p,1}(X)$ and $\beta_{p,2}(X)$.
\end{theorem}
The equations identify the policy-relevant coefficients under partial-information orthogonality only when the full column rank condition of Theorem~\ref{thm: identification} holds. When there is no interactive effect of IV $G$ and price $P$ on the revenue, that is, $\beta_{p,g}(X,G)=0$, we have the following Corollary \ref{cor: identification 2}, which gives the explicit identification of $\beta_{p,1}(X)$ and $\beta_{p,2}(X)$.
\begin{corollary}\label{cor: identification 2}
Suppose that Models \eqref{Model: counterfactual} and \eqref{Model: action} hold with $\beta_{p,g}=0$, together with Assumption~\ref{ass: process independence} and consistency. If Assumption \ref{ass: identification} holds, then we have
\begin{align*}
  \beta_{p, 1}(X) &= \left(\Upsilon_3(X)\Omega_1(X) - \Omega_3(X)\Upsilon_1(X)\right)/\left(\Upsilon_3(X)\Omega_2(X) - \Upsilon_2(X)\Omega_3(X)\right),\\
  \beta_{p, 2}(X) &= \left(\Omega_1(X)\Upsilon_2(X) - \Omega_2(X)\Upsilon_1(X)\right)/\left(\Omega_3(X)\Upsilon_2(X) - \Omega_2(X)\Upsilon_3(X)\right),
\end{align*}
    where
    \begin{align*}
	\Omega_1(X) &= \EE\left[(G - \EE\left[G \given X\right])(P - \EE\left[P \given G, X \right])Y\given  X \right],\\
	\Omega_2(X) &=  \EE\left[(G - \EE\left[G \given X\right])(P- \EE\left[P \given G, X \right])P \given X\right],\\
\Omega_3(X) &= \EE\left[(G - \EE\left[G \given X\right])(P - \EE\left[P \given G, X \right])P^2 \given X\right],\\
		\Upsilon_1(X) &= \Cov(G(G - \EE\left[G \given X  \right]), (P - \EE\left[P \given G, X \right])Y \given X),\\
	\Upsilon_2(X) &= \Cov(G(G - \EE\left[G  \given X  \right]), P(P - \EE\left[P \given G, X \right]) \given X),\\
	\Upsilon_3(X) &= \Cov(G(G - \EE\left[G  \given X  \right]), P^2(P - \EE\left[P \given G, X \right]) \given X),
    \end{align*}
provided the higher-moment relevance condition $\det A(X) \neq 0$ holds almost surely, where 
\begin{equation}\label{eq: rank matrix A}
A(X)=\begin{pmatrix}\Omega_2(X)&\Omega_3(X)\\\Upsilon_2(X)&\Upsilon_3(X)\end{pmatrix}.
\end{equation}
\end{corollary}

\begin{remark}[Higher-moment relevance]\label{rem: hmr}
The determinant condition in Corollary~\ref{cor: identification 2} requires more than an instrument-induced shift in the mean price. In particular,
$\Omega_2(X)=\Cov\{G,\Var(P\mid G,X)\mid X\}$, so a convenient check for variation relevant to the price slope is
\begin{equation}\label{eq: hmr}
\Cov\big(G,\,\Var(P\mid G,X)\mid X\big)\neq 0.
\end{equation}
A pure mean shift is removed by the residual $P-\EE[P\mid G,X]$ and therefore does not by itself identify $\beta_{p,1}(X)$. Condition~\eqref{eq: hmr} is sufficient to make the corresponding slope column nonzero, while full identification still requires the two columns in the determinant to be linearly independent. Appendix Section~\ref{app: higher moment details} provides the derivation and a counterexample.
\end{remark}

For the sake of simplicity, the remainder of the paper restricts to settings in Corollary~\ref{cor: identification 2}, i.e., the no-interaction specialization $\beta_{p,g}=0$ of Model~\eqref{Model: counterfactual}. 
Within the restricted model, policy learning first estimates $\beta_{p,1}(X)$ and $\beta_{p,2}(X)$ and then applies \eqref{eqn: model-based optimal policy formulation}.

\section{Estimation and Policy Learning}
\label{sec: estimation and policy learning}

We estimate the two price-response functions, then optimize revenue over the feasible interval. Corollary~\ref{cor: identification 2} allows a plug-in approach using sequential, cross-fitted regressions for its nested conditional moments. While such an approach is feasible, it inherits first-stage error and can be sensitive to a nearly singular moment matrix $A(X)$. PRINT instead estimates nuisances first and jointly enforces the remaining pricing moments with an adversarial criterion.

The six nuisance targets $h_1,\dots,h_6$ are defined as follows. Their true values are denoted $h_{0,j}$ when we distinguish them from candidate fits.
\begin{equation*}
\begin{aligned}
	h_1(X) &= \EE\left[G \given X\right],\quad &&h_2(X, G) = \EE\left[P \given X, G\right],\\
	h_3(X) &= \EE\left[G^2 \given X\right],\quad &&h_4(X) = \EE\left[(P - h_2(X, G))Y\given X\right],\\
	h_5(X) & = \EE\left[P(P - h_2(X, G)) \given X\right], \quad
	&& h_6(X)  = \EE\left[P^2(P - h_2(X, G)) \given X\right].
\end{aligned}
\end{equation*}
Among these, $h_1(X)=\EE[G\given X]$, $h_2(X,G)=\EE[P\given G,X]$, and $h_3(X)=\EE[G^2\given X]$ are ordinary regression functions of directly observed responses and can be estimated by any supervised-learning method, whereas $h_4$, $h_5$, $h_6$ involve the nested nuisance $h_2$, and $\beta_{p,1}$, $\beta_{p,2}$ are defined only through the identification system. 

We therefore adopt a two-stage design. In the \emph{first stage}, we estimate $\wh h=(\wh h_1,\wh h_2,\wh h_3)$ by supervised learning with cross-fitting. We randomly split the sample into folds, and the first-stage estimates applied to each fold are computed from the remaining folds, so that each fold-specific fit is independent of its own held-out observations. The fitted values are clipped to the range of the corresponding responses. 
The first stage consists of two standard supervised-learning tasks. We regress the two-dimensional response $(G,G^2)$ on $X$ to estimate $(h_1,h_3)$, and regress $P$ on $(X,G)$ to estimate $h_2$. These functions are estimated first because $h_2$ depends on $(X,G)$, whereas the second-stage unknowns and moment restrictions are indexed only by $X$. Cross-fitted predictions from these regressions provide the residuals used in the second stage.

In the \emph{second stage}, we collect the remaining unknowns as $\alpha=(\beta_{p,1},\beta_{p,2},h_4,h_5,h_6)$, a vector-valued function of $X$, with five components, and estimate them by the adversarial criterion below with $\wh h$ plugged in. Now let $Z=(X,G,P,Y)$ and
\begin{align*}
\wt W_1=\begin{pmatrix}
G - h_1(X), & P- h_2(X, G)\\
 G^2- h_3(X), & (P - h_2(X, G))Y- h_4(X)\\
P(P - h_2(X, G))- h_5(X), & P^2(P - h_2(X, G)) - h_6(X)
\end{pmatrix},
\end{align*}
and let $\wt W_2 = (w_7, w_8)^\top$, where
\begin{equation*}
	\begin{aligned}
		w_7(Z; \beta_{p,1}, \beta_{p,2}, h_1, h_2) &= \rho_1(Z; h_1, h_2) - \rho_2(Z; h_1, h_2)\beta_{p, 1}(X) \\
  &\quad - \rho_3(Z; h_1, h_2)\beta_{p, 2}(X),\\
		w_8(Z; \beta_{p,1}, \beta_{p,2}, h_1, \dots, h_6)  &= \rho_4(Z; h_1, h_2) -(h_3(X) - h_1^2(X))h_4(X)\\
		&  \quad- \left\{\rho_5(Z; h_1, h_2) - (h_3(X) - h_1^2(X))h_5(X) \right\}\beta_{p, 1}(X)\\
		& \quad-\left\{\rho_{6}(Z; h_1, h_2) - (h_3(X) - h_1^2(X))h_6(X)  \right\}\beta_{p, 2}(X).
	\end{aligned}
\end{equation*}
In particular,
\begin{align*}
    \rho =
    \begin{pmatrix}
(G - h_1(X))(P - h_2(X, G))Y, & (G - h_1(X))(P - h_2(X, G))P\\
 (G - h_1(X))(P - h_2(X, G))P^2, & G(G - h_1(X))(P - h_2(X, G))Y\\
G(G - h_1(X))(P - h_2(X, G))P, & G(G - h_1(X))(P - h_2(X, G))P^2
\end{pmatrix},
\end{align*}
with $\text{Vec}(\rho) = (\rho_1, \cdots, \rho_6)^\top$, where $\text{Vec}(A)$ denotes the vector that stacks the entries of $A$ row by row.

Then we have the following lemma that characterizes the property of all nuisance parameters. To lighten the notation, we write $h_0 = (h_{0,1}, h_{0,2}, h_{0,3})$ for the true first-stage components and $\alpha_0 = (\beta_{p, 1}, \beta_{p, 2}, h_{0,4}, h_{0,5}, h_{0,6})$ for the true second-stage components, and let $W(Z; \alpha, h) = \left(\text{Vec}^\top(\wt W_1), \text{Vec}^\top(\wt W_2) \right)^\top$ for generic $(\alpha, h)$.
We abbreviate $W_0(Z;\alpha) \triangleq W(Z;\alpha,h_0)$.

\begin{lemma}\label{lm: conditional moment restriction}
  Under Models \eqref{Model: counterfactual} and \eqref{Model: action} with $\beta_{p,g}=0$, Assumption~\ref{ass: process independence}, \ref{ass: identification}, causal consistency and the standing integrability condition of Section~\ref{sec: policy learning with invalid iv}, the following conditional moment restriction holds almost surely:
  \begin{align}\label{eqn: conditional moment restriction}
	\EE\left[W_0(Z; \alpha_0) \given X\right] = 0.
\end{align}

\end{lemma}
The proof is in Section~\ref{sec: Proof of Estimation and Policy Learning} of the Appendix. 
Equation \eqref{eqn: conditional moment restriction} is a nonparametric conditional moment problem \citep{chen2012estimation}, where
$Z=(X,G,P,Y)$ collects the observed variables entering the residual vector, and $X$ indexes the conditional moment restriction used to estimate the unknown functions. Here $X$ is not an instrument in the economic sense in which $G$ is a candidate pricing instrument. The residual is nonlinear in the candidate functions, unlike the linear residual in standard nonparametric IV regression. 
The general nonparametric conditional moment problem has been studied via penalized sieve minimum distance \citep{chen2012estimation,chen2015sieve} and, over flexible function classes, via the variational method of moments \citep{bennett2023variational}. Given a wide range of machine
learning approaches, it is essential to study this estimation problem under
flexible nonparametric function classes such as reproducing kernel Hilbert
spaces (RKHSs), neural networks, high-dimensional linear models, etc. Motivated
by these and also by \citet{dikkala2020minimax,bennett2023variational}, we reformulate Equation
\eqref{eqn: conditional moment restriction} into an unconditional moment
restriction via a min-max criterion, based on which we estimate $\alpha_0$ by
solving a zero-sum game.

Specifically, suppose that we have  $n$ independent and identically
distributed samples
$\calD_n  = \left\{Z_i = (X_i, G_i, P_i, Y_i)  \right\}_{i = 1}^n$, the cross-fitted first-stage estimates $\wh h$, and an initial guess $\wt{\alpha}$ of $\alpha_0$. We propose to obtain an estimator $\wh{\alpha}_n$  of $\alpha_0$ by solving
\begin{equation}\label{eqn: min-max estimation}
\min_{\alpha \in\mH}\sup_{f\in\mF} \quad  \wh\Psi_n(\alpha,f) - \norm{f}_{n,\wt{\alpha}}^2- \lambda_n
\norm{f}_{\mF}^2 + \mu_n \norm{\alpha}_{\mH}^2,
\end{equation}
where
$
\wh\Psi_n(\alpha,f) = \frac{1}{n} \sum_{i=1}^n W(Z_i; \alpha, \wh h)^{\top}f(X_i)
$
uses the shorthand $\wh h=\wh h^{(-k(i))}$. The same convention applies to the empirical quadratic penalty. It is the feasible version of $\Psi_n(\alpha,f) = \frac{1}{n} \sum_{i=1}^n W_0(Z_i;\alpha)^{\top}f(X_i)$,
\begin{equation*}
\calF = \{f: \calX \rightarrow \mathbb{R}^{d_W} \, \given \, f = (f_1, \cdots, f_{d_W})^\top \, \text{with} \, f_i \in \calF_i  \text{
  for
}  1 \leq i \leq d_W=8 \},
\end{equation*}
where $d_W=8$: vectorizing the $3\times2$ matrix $\wt W_1$ gives six residual components, and $\wt W_2=(w_7,w_8)^\top$ contributes two more.
The spaces $\calH$ and $\calF_i$ are user defined. Examples include RKHSs, random forests, neural networks, and many others. The norm
associated with $\calF$ is defined as
$\norm{f}_{\calF}^2 = \sum_{i = 1}^{d_W} \norm{f_i}_{\calF_i}^2$, where
$\norm{\bullet}_{\calF_i}$ is some functional norm, and $\norm{\bullet}_{\calH}$ is the norm
associated with  space~$\calH$, and $\lambda_n,\mu_n>0$ are regularization parameters. In addition, for generic $(\alpha,h)$ we define the weighted empirical norm
$\norm{f}_{n,\alpha;h}^2 = \frac{1}{n} \sum_{i=1}^n \{f(X_i)^{\top}W(Z_i;\alpha,h)\}^2$, and the criterion \eqref{eqn: min-max estimation} uses the feasible version $\norm{f}_{n,\wt{\alpha}}^2 \triangleq \norm{f}_{n,\wt{\alpha};\wh h}^2$, where $W(Z_i;\wt{\alpha},\wh h)$ is used to balance the weights of conditional moment restrictions (the theoretical analysis abbreviates the oracle version $\norm{f}_{n,\alpha;h_0}^2$ by the same subscript pattern with $h_0$ made explicit where the distinction matters). 

The following lemma shows that the population criterion is zero precisely at the true parameter when the adversary contains the exact conditional residuals. The learning results below allow approximate residual tests and quantify their error separately. The lemma is proved in Section~\ref{sec: Proof of Estimation and Policy Learning} of the Appendix.
\begin{lemma}[Identification through the population criterion]\label{lm: population criterion}
With the conditions in Corollary~\ref{cor: identification 2}, and $\det A(X)\neq 0$. Suppose $\mF$ is star-shaped about 0 and contains $m(\cdot;\alpha)$ for every $\alpha\in\mH$. Fix a bounded reference parameter function $\widetilde\alpha$ and define
$J(\alpha)=\sup_{f\in\mF}\{\Psi(\alpha,f)-\|f\|_{\widetilde\alpha}^2\}$.
There is a fixed $c_J>0$ such that
\[
J(\alpha)\geq c_J\|m(\cdot;\alpha)\|_2^2,\qquad
J(\alpha)=0\ \text{if and only if } \alpha=\alpha_0\text{ almost surely}.
\]
If $\Sigma_{\widetilde\alpha}(X)\succeq c_\Sigma I$ almost surely, then also
$J(\alpha)\leq\|m(\cdot;\alpha)\|_2^2/(4c_\Sigma)$.
When $\alpha_0\in\mH$, it is the unique minimizer almost surely. 
\end{lemma}

After obtaining $\widehat \alpha_n$ by solving the min-max problem \eqref{eqn: min-max estimation}, we compute an estimator of $\pi^\ast$ as
\begin{equation}
  \label{eqn: model-based optimal policy}
	\wh{\pi}_n = \clip{p_1}{p_2}{- \wh{\beta}_{p,1}/[2\wh{\beta}_{p,2}]}.
\end{equation}
We summarize the two-stage estimation procedure in Algorithm~\ref{alg: simplified two-stage}. 
Appendix~\ref{app:practical-algorithm} gives the detailed neural-network implementation. %
\begin{remark}
Throughout, $\mH$ is restricted so that its curvature component satisfies $-\beta_{p,2} \geq c_{p,2}$ for a user-chosen constant $0<c_{p,2}\leq c$, where $c$ is the model constant of Section~\ref{sec: policy learning with invalid iv}. This makes the clipped ratio in \eqref{eqn: model-based optimal policy} well defined. This restriction is enforced in the detailed algorithm of the Appendix.

With probability tending to one, the fixed reference must lie in a fixed norm ball, and its conditional residual second-moment matrix $\Sigma_{\widetilde\alpha}(X)$ must have eigenvalues uniformly bounded away from zero. See Assumption~\ref{ass: identification and spaces}(c). The proof is uniform over that ball and does not require the reference to converge to $\alpha_0$ in the $\mH$ norm. Iteratively replacing it during optimization is an additional practical procedure.
\end{remark}

\begin{algorithm}
  \caption{Simplified two-stage estimation procedure.}
  \label{alg: simplified two-stage}
  \textbf{Input:} Offline data $\calD_n=\{(X_i,G_i,P_i,Y_i)\}_{i=1}^n$, function classes and tuning parameters, reference estimate $\wt\alpha$, price range $[p_1,p_2]$, and curvature bound $c_{p,2}$.\\
  \textbf{Stage 1: Estimate ordinary regression nuisances.} Split the data into folds. On the complementary folds, regress $(G,G^2)$ on $X$ to estimate $(h_1,h_3)$ and regress $P$ on $(X,G)$ to estimate $h_2$. Use the fitted models to obtain clipped, out-of-fold predictions for every observation.\\
  \textbf{Stage 2: Estimate the policy-relevant functions.} Plug the cross-fitted predictions into $W(Z;\alpha,\wh h)$ and solve the min-max problem \eqref{eqn: min-max estimation} for $\alpha=(\beta_{p,1},\beta_{p,2},h_4,h_5,h_6)$, subject to $-\beta_{p,2}\geq c_{p,2}$.\\
  \textbf{Output:} Return $\wh\pi_n(X)=\clip{p_1}{p_2}{-\wh\beta_{p,1}(X)/[2\wh\beta_{p,2}(X)]}$.
\end{algorithm}

\section{Theoretical Analysis}
\label{sec: theory}
This section studies the no-interaction case $\beta_{p,g}=0$ under the conditions of Corollary~\ref{cor: identification 2}. We establish theoretical properties describing how accurately the data determine the pricing moments, how residual error controls derivative error, and how identification strength converts derivative error into coefficient error.

\subsection{Learning Conditions}
\label{subsec: preliminary assumptions}
Write $\|g\|_2^2=\EE[\|g(X)\|_{\ell^2}^2]$ for vector-valued functions of $X$, and $\|g\|_{n,2}^2=n^{-1}\sum_i\|g(X_i)\|_{\ell^2}^2$. 
Let $\mH_M=\{\alpha\in\mH:\|\alpha\|_{\mH}^2\leq M\}$ and $\mF_B=\{f\in\mF:\|f\|_{\mF}^2\leq B\}$ be the parameter and adversary norm balls, respectively. The admissible class $\mH$ inherits its norm from a normed function space and retains the curvature restriction $-\beta_{p,2}\geq c_{p,2}>0$. Let
$m(X;\alpha)=\EE[W_0(Z;\alpha)\mid X]$, 
$\Sigma_{\alpha}(X)=\EE[W_0(Z;\alpha)W_0(Z;\alpha)^\top\mid X]$.
Thus $\|m(\cdot;\alpha)\|_2$ measures violations of the pricing moments.

We allow the estimation classes to grow with the sample size. The classes may be sieves $\mH=\mH_n$ and $\mF=\mF_n$, where the index is suppressed except in the identification modulus. The classes $\mH_{n,M}$ and $\mF_{n,B}$ denote the corresponding parameter and adversary balls at sample size $n$. For such sieves, all fixed-ball envelopes, richness bounds, reference constants, and critical radius multipliers below are uniform in $n$. %

\begin{assumption}[Learning regularity]\label{ass: identification and spaces}
The following conditions hold.
\begin{enumerate}[label=(\alph*)]
\item \label{ass: identifiability}The ground truth $\alpha_0$ and the data $(Y,P,G)$ are bounded. 
For every fixed $0<M,B<\infty$, there are constants $D_M,E_B<\infty$, independent of $n$, such that
\[
\sup_{n\geq1}\sup_{\alpha\in\mH_{n,M}}\|\alpha\|_\infty\leq D_M,
\qquad
\sup_{n\geq1}\sup_{f\in\mF_{n,B}}\|f\|_\infty\leq E_B,
\]
where $\|g\|_\infty:=\operatorname*{ess\,sup}_x\|g(x)\|_{\ell^2}$. These output bounds are required to hold on a common $P_X$-full-measure set for all $n$ and all functions in the respective balls.
For fixed $M_\star<\infty$ and a deterministic sequence $\varepsilon_n\downarrow0$, there is a deterministic approximating function $\alpha_n^\dagger\in\mH_n$ such that
\begin{equation}
\|\alpha_n^\dagger\|_{\mH_n}^2\leq M_\star,\qquad
\|m(\cdot;\alpha_n^\dagger)\|_2\leq\varepsilon_n.
\label{eq: approximate realizability}
\end{equation}
\item \label{ass: f space}$\mF$ is star-shaped about zero, i.e., $tf\in\mF$ for every $f\in\mF$ and $t\in[0,1]$. The norm on its function space is absolutely homogeneous, i.e., $\|cf\|_{\mF}=|c|\|f\|_{\mF}$ for every real scalar $c$ and every $f$ in that space. There is a deterministic sequence $\xi_n\downarrow0$ such that, for each fixed $M<\infty$, and every $\alpha\in\mH_{n,M}$, there exists a $g_{\alpha,n}\in\mF_n$ satisfying
\begin{equation}
\|g_{\alpha,n}-m(\cdot;\alpha)\|_2\leq\xi_n,\qquad
\|g_{\alpha,n}\|_{\mF_n}\leq L_M<\infty,
\label{eq: fixed ball approximate residual tests}
\end{equation}
where $L_M$ is independent of $n$.
\item \label{ass: nondegenerate} The reference $\widetilde\alpha$ is held fixed during optimization. For fixed $M_{\rm ref}<\infty$ and $c_\Sigma>0$, with probability tending to one, $\widetilde\alpha\in\mH_{M_{\rm ref}}$ and $\lambda_{\min}\{\Sigma_{\widetilde\alpha}(X)\}\geq c_\Sigma$ for almost every $X$.
\end{enumerate}
\end{assumption}
Assumption~\ref{ass: identification and spaces}\ref{ass: identifiability} relaxes the realizability of the ground truth $\alpha_0$ and quantifies the approximation error. On common bounded sets, the polynomial moment map is $L^2$-Lipschitz, so an $L^2$ approximation $\|\alpha_n^\dagger-\alpha_0\|_2\leq e_n$ suffices with $\varepsilon_n\leq C e_n$. Exact realizability is the special case $\alpha_n^\dagger=\alpha_0$, $\varepsilon_n=0$. Such approximation--estimation decompositions have precedents in literature on sieve conditional moment estimation and minimax IV estimation \citep{chen2012estimation,dikkala2020minimax}. Finite sieves may exclude $\alpha_0$ at every sample size. In a fixed RKHS with a bounded kernel, however, approximation with uniformly bounded RKHS norms and vanishing $L^2$ error requires $\alpha_0$ to have an RKHS representative. Appendix~\ref{app: approximation extension} explains this limitation, and Proposition~\ref{prop: capacity oracle} allows growing bounds on the parameter and adversary test norms under common output bounds and complexity control.
Assumption~\ref{ass: identification and spaces}\ref{ass: f space} allows the adversary to approximate conditional residuals with uniformly bounded test norms on each fixed parameter ball. The error $\xi_n$ measures adversary approximation, 
whereas $\varepsilon_n$ bounds the moment approximation error $\|m(\cdot;\alpha_n^\dagger)\|_2$ of the approximating function $\alpha_n^\dagger$ in~\eqref{eq: approximate realizability}.
Exact residual tests are the special case $\xi_n=0$. Assumption~\ref{ass: identification and spaces}\ref{ass: nondegenerate} concerns the weighting matrix.

We next specify complexity, nuisance error, and optimization accuracy. For a bounded scalar class $\mathcal C$ on observations $O\sim P_O$, define its star hull by
$\operatorname{star}(\mathcal C)=\{tg:g\in\mathcal C,\ 0\leq t\leq1\}$ and its population-localized Rademacher complexity by
\[
\mathcal R_{n,P_O}(r,\mathcal C)=\EE_{O_{1:n},\epsilon}
\sup_{\substack{g\in\operatorname{star}(\mathcal C)\\
\|g\|_{L^2(P_O)}\leq r}}
\left|\frac1n\sum_{i=1}^n\epsilon_i g(O_i)\right|,
\]
where $O_1,\ldots,O_n$ are independent with law $P_O$, and the Rademacher signs are independent of the observations. We suppress $P_O$ when its meaning is clear. The radius $r$ is an $L^2$-norm radius. After dividing the class by a fixed uniform envelope, a critical-radius upper bound is a value $r$ satisfying $\mathcal R_{n,P_O}(r,\mathcal C)\leq r^2/64$ \citep{wainwright2019high}.

For the eight-dimensional adversary $f$, the scalar coordinate classes are
$\mathcal F_{B,k}=\{x\mapsto f_k(x):f\in\mF_B\}$, $k=1,\ldots,8$.
Their lifts to the observation space are
$\widetilde{\mathcal F}_{B,k}=\{z\mapsto f_k(x):f\in\mF_B\}$, where $z=(x,g,p,y)$.
Lifting preserves both population $L^2$ norms and sample evaluations, so
$\mathcal R_{n,P_Z}(r,\widetilde{\mathcal F}_{B,k})
=\mathcal R_{n,P_X}(r,\mathcal F_{B,k})$.

\begin{assumption}[Sampling, complexity, and optimization]\label{ass: parameters}
Let $b_n\downarrow0$ be a common critical-radius rate and let $r_n\downarrow0$ be a first-stage error rate.
\begin{enumerate}[label=(\alph*)]
\item \label{ass: complexity} For some fixed $B>0$ and every fixed $0< M<\infty$, the critical radii of $\mathcal F_{B,k}$ under $P_X$, $k=1,\ldots,8$, and of
$\mathcal G_{M,B}=\{z\mapsto W_0(z;\alpha)^\top f(x):\alpha\in\mH_M,\ f\in\mF_B\}$ under $P_Z$ are at most $C_M b_n$. 
\item \label{ass: first stage}Use a fixed number of cross-fitting folds, each of size proportional to $n$, and clip first-stage fits to fixed valid bounds. For a fixed $C_h<\infty$, with probability tending to one, every fold-specific fit satisfies
$\max_{j=1,2,3}\|\widehat h_j^{(-k)}-h_{0,j}\|_2\leq C_h r_n$.

\item \label{ass: sample criterion}Let $a_n=b_n+r_n+\sqrt{\log n/n}$ and $\bar a_n=a_n+\varepsilon_n$. Choose $\lambda_n\asymp a_n^2$, with its constant sufficiently large, and $\mu_n\asymp\bar a_n^2$. With probability tending to one, we have
$\widehat\Phi_n(\widehat\alpha_n)\leq\inf_{\alpha\in\mH_n}\widehat\Phi_n(\alpha)+C_{\rm opt}\eta_n$,
where $\eta_n=O(\bar a_n^2)$ and
$\widehat\Phi_n(\alpha) =\sup_{f\in\mF_n}\left\{ \widehat\Psi_n(\alpha,f)-\|f\|_{n,\widetilde\alpha}^2 -\lambda_n\|f\|_{\mF_n}^2\right\} +\mu_n\|\alpha\|_{\mH_n}^2.$
The empirical linear and quadratic terms use the cross-fitted first-stage estimates, with the reference $\widetilde\alpha$ fixed, as in~\eqref{eqn: min-max estimation}.
\end{enumerate}
\end{assumption}
In Assumption~\ref{ass: parameters}, the first-stage error bound of order $r_n$ holds with probability tending to one. Logarithmic inflation may be needed to obtain this bound. Unlike a requirement that nuisances converge strictly faster than the second stage, the result below keeps $r_n$ explicitly. Cross-fitting provides foldwise independence for nuisance evaluation.
The two penalties serve different purposes. The adversary penalty $\lambda_n$ absorbs sampling and first-stage fluctuations, whereas the parameter penalty $\mu_n$ should account for the approximation error $\varepsilon_n$. Keeping $\mu_n\asymp a_n^2$ without $\varepsilon_n=O(a_n)$ would give only a radius of order $1+(\varepsilon_n/a_n)^2$.
For the final error bounds, define
\begin{equation}
\delta_n=\bar a_n+\xi_n
=b_n+r_n+\sqrt{\log n/n}+\varepsilon_n+\xi_n.
\label{eq: total learning error}
\end{equation}
The error $\xi_n$ arises from approximating the conditional residuals $m(\cdot;\alpha)$ by adversary functions $g_{\alpha,n}$ in~\eqref{eq: fixed ball approximate residual tests}. The approximating function $\alpha_n^\dagger$, whose moment error is bounded by $\varepsilon_n$ in~\eqref{eq: approximate realizability}, determines $\bar a_n=a_n+\varepsilon_n$ and hence the scales $\mu_n\asymp\bar a_n^2$ and $\eta_n=O(\bar a_n^2)$ used to obtain a fixed parameter radius.

\subsection{Convergence Rates and Regret Bounds}
\label{subsec: regret bounds}
Learning the pricing moments and recovering the response functions are separate tasks. We first bound the conditional residual error, compare it with $\|T_0(\alpha-\alpha_0)\|_2$, and then use identification conditions to bound the coefficient error. The derivative-based approach is standard in conditional moment estimation \citep{chen2012estimation}.

Hold the first-stage functions at $h_0$, put $u=\alpha-\alpha_0$, and define the derivative at the truth by
\begin{equation*}
(T_0u)(x)=\frac{dm(x;\alpha_0)}{d\alpha}[u]
=\left.\frac{d}{dt}m(x;\alpha_0+tu)\right|_{t=0}.
\end{equation*}
The polynomial formula for $m(x;\alpha)$ defines the derivative along $\alpha_0+tu$, even when this path leaves the admissible class. For PRINT, $T_0u(x)=J_0(x)u(x)$, with the explicit matrix $J_0$ in Appendix~\ref{app: derivative link}. We measure its size by
$\|T_0u\|_2^2=\EE[\|(T_0u)(X)\|_{\ell^2}^2]$.
A small value of $\|T_0u\|_2$ means that the perturbation $u$ has a small first-order effect on the identifying moments. %

\begin{lemma}[Residual--derivative comparison]\label{lem: residual derivative link}
On each fixed ball $\mH_M$, with bounded structural target $\alpha_0$ not necessarily in $\mH_M$ or in $\mH_n$, there is $K_M<\infty$, uniform in $n$, such that for every $\alpha\in\mH_M$, 
\begin{equation}
K_M^{-1}\|T_0(\alpha-\alpha_0)\|_2
\leq\|m(\cdot;\alpha)\|_2
\leq K_M\|T_0(\alpha-\alpha_0)\|_2.
\label{eq: unweighted derivative bridge}
\end{equation}
\end{lemma}
Lemma~\ref{lem: residual derivative link} compares residual error with derivative error. Its proof is given in Appendix~\ref{app: derivative link}. To quantify how derivative error controls coefficient error on the same fixed ball, define the derivative modulus
\begin{equation}
\omega_{M,n}^D(t)=\sup\{\|\alpha-\alpha_0\|_2:\alpha\in\mH_{n,M},
                    \|T_0(\alpha-\alpha_0)\|_2\leq t\}.
\label{eq: modulus of continuity}
\end{equation}
Here $\mH_{n,M}=\{\alpha\in\mH_n:\|\alpha\|_{\mH_n}^2\leq M\}$, and the supremum is defined as zero when the candidate set is empty. The derivative is evaluated at ground truth $\alpha_0$ for every $n$, and $T_0$ acts on $L^2(P_X)^5$, even when $\alpha_0\notin\mH_n$. The modulus measures the largest coefficient error compatible with derivative error at most $t$.
Appendix~\ref{app: spectral link} bounds this modulus using derivative Gram matrices and explicit link conditions. Uniformly over the fixed-radius sieve balls, maintained almost-sure rank and common envelopes give $\sup_n\omega_{M,n}^D(t)\to0$ as $t\downarrow0$, which is proved in Appendix~\ref{app: coefficient stability} Lemma~\ref{lem: qualitative derivative stability}.

\begin{theorem}[Moment and derivative error bounds]\label{thm: convergence rate}
In the no-interaction case $\beta_{p,g}=0$, under Assumptions~\ref{ass: identification and spaces}--\ref{ass: parameters}, there are fixed finite constants $M\geq M_\star$ and $C$ such that, with probability tending to 1,
\begin{equation}
\|\widehat\alpha_n\|_{\mH_n}^2\leq M,\qquad
\|m(\cdot;\widehat\alpha_n)\|_2\leq C\delta_n,\qquad
\|T_0(\widehat\alpha_n-\alpha_0)\|_2\leq C\delta_n.
\label{eq: residual convergence}
\end{equation}
On the same event,
\begin{equation}
\|\widehat\alpha_n-\alpha_0\|_2\leq\omega_{M,n}^D(C\delta_n).
\label{eq: derivative modulus convergence}
\end{equation}
If, for every fixed $C_0>0$, a constant $K_{C_0,M}<\infty$ independent of $n$ and $t$ satisfies
$\omega_{M,n}^D(C_0t)\leq K_{C_0,M}\omega_{M,n}^D(t)$ for sufficiently large $n$ and sufficiently small $t>0$, 
then $\|\widehat\alpha_n-\alpha_0\|_2=O_p\{\omega_{M,n}^D(\delta_n)\}$.
\end{theorem}

Theorem~\ref{thm: convergence rate} gives a general convergence rate for the sieve estimator $\widehat\alpha_n$. When $\alpha_0$ actually belongs to each candidate ball and the ball is star-shaped about it, rescaling gives $\omega_{M,n}^D(C_0t)\leq\max\{1,C_0\}\omega_{M,n}^D(t)$. In section \ref{subsec: regret bounds for RKHSs}, we give a concise specialization with an explicit sample-size rate for RKHSs. The proof of Theorem~\ref{thm: convergence rate} is given in Appendix~\ref{app: learning proofs}.

\begin{theorem}[Quadratic revenue regret]\label{thm: regret bound}
Suppose the revenue model holds in the no-interaction case $\beta_{p,g}=0$ and $-\beta_{p,2}(X)\geq c>0$. For any estimated coefficients with $-\widehat\beta_{p,2}\geq c_{p,2}>0$, define $\widehat\pi$ by~\eqref{eqn: model-based optimal policy} and let $\pi^*$ optimize over all measurable feasible policies. With $e_j=\widehat\beta_{p,j}-\beta_{p,j}$, conditionally on the fitted functions, the regret
\begin{equation}
R(\widehat\pi) \triangleq \mathcal V(\pi^*)-\mathcal V(\widehat\pi)
\leq\frac{1}{4c}\EE_X[(e_1(X)+2\widehat\pi(X)e_2(X))^2]
\leq\frac{1+4p_2^2}{4c}\sum_{j=1}^2\|e_j\|_2^2.
\label{eq: quadratic regret}
\end{equation}
Under Theorem~\ref{thm: convergence rate}, the regret $R(\widehat\pi) = O_p \{\omega_{M,n}^D(C\delta_n)^2\}$ and is $O_p\{\omega_{M,n}^D(\delta_n)^2\}$ under the modulus regularity.
\end{theorem}
Theorem~\ref{thm: regret bound} gives a quadratic regret bound in coefficient error. Combining it with Theorem~\ref{thm: convergence rate} yields the regret rate. The proof is given in Appendix~\ref{app: learning proofs}. The result does not require identification of the price-independent revenue terms, which may be misspecified or unidentifiable. In section \ref{subsec: regret bounds for RKHSs}, we give a specialization with an explicit sample-size rate for RKHSs.

\subsection{Regret Bounds for RKHSs}
\label{subsec: regret bounds for RKHSs}
This section specializes Theorems~\ref{thm: convergence rate} and~\ref{thm: regret bound} to RKHSs to obtain explicit convergence and regret rates in the sample size. RKHSs provide a broad framework for modeling nonlinear covariate effects, encompassing both finite-dimensional function spaces and infinite-dimensional smoothness classes such as Sobolev spaces with sufficiently high smoothness \citep{steinwart2008support}. The choice of kernel accommodates different smoothness assumptions while providing a tractable characterization of function-class complexity \citep{mendelson2003performance}. We first introduce an assumption on the residual RKHS and its eigenvalue decay.
\begin{assumption}[RKHS complexity bound]\label{ass: RKHS eigen-decay}
The estimation classes $\mH_n$ and $\mF_n$ are contained in fixed separable RKHSs of vector-valued functions $\alpha:\mathcal X\to\mathbb R^5$ and $f:\mathcal X\to\mathbb R^8$, respectively, with bounded measurable matrix-valued kernels $K_{\mH}$ and $K_{\mF}$, and inherit their norms. For every fixed $M,B>0$, there exist a bounded measurable scalar kernel $K_{\rm res}$ on $\mathcal Z$, with separable RKHS $\mathcal H_{\rm res}$, and $0<R_{M,B}<\infty$, all independent of $n$, such that, for every $n$,
\[
\mathcal G_{M,B}\ \cup\ \bigcup_{k=1}^8\widetilde{\mathcal F}_{B,k}
\ \subseteq\ \{g\in\mathcal H_{\rm res}:\|g\|_{\mathcal H_{\rm res}}\leq R_{M,B}\}.
\]
The lifts $\widetilde{\mathcal F}_{B,k}$ are defined in Section~\ref{subsec: preliminary assumptions}. The nonincreasing eigenvalues of the integral operator of $K_{\rm res}$ on $L^2(P_Z)$ satisfy
\begin{equation*}
\nu_j(K_{\rm res};P_Z)\leq C j^{-2s/d},\qquad s>d/2,
\end{equation*}
where $C<\infty$ is independent of $n$ and $j$, and $d=q$ is the dimension of the covariates used by the second-stage functions.
\end{assumption}
For a coordinate kernel $L_k$ on $\mathcal X$, the lifted kernel is $\widetilde L_k(z,z')=L_k(x,x')$.
The RKHS $\mathcal H_{\rm res}$ contains the scalar functions $z\mapsto W_0(z;\alpha)^\top f(x)$ and $z\mapsto f_k(x)$ for empirical-process analysis. The parameter and adversary RKHSs contain vector-valued functions of $X$. Their scalar coordinate spaces need not contain the weighted products appearing in the criterion. Assumption~\ref{ass: RKHS eigen-decay} therefore requires containment in a scalar RKHS ball and a polynomial upper bound on its eigenvalues. Appendix~\ref{app: residual kernel} establishes containment using finite sums of weighted product kernels. Appendix~\ref{app: sobolev construction} verifies both requirements for coordinate spaces embedded in $H^s([0,1]^d)$, $s>d/2$, under its embedding and design conditions, with bounded observation weights preserving the spectral exponent. These constructions show that Assumption~\ref{ass: RKHS eigen-decay} covers commonly used smoothness classes. 

Then we are ready to state the convergence and regret rates for RKHSs. The following corollaries are direct consequences of Theorems~\ref{thm: convergence rate} and \ref{thm: regret bound}.
\begin{corollary}[RKHS derivative-norm rate]\label{cor: RKHS derivative convergence rate}
Suppose Assumptions~\ref{ass: identification and spaces}, \ref{ass: parameters}(b)--(c), and \ref{ass: RKHS eigen-decay} hold. The residual RKHS assumption implies Assumption~\ref{ass: parameters}(a) with the admissible common critical-radius rate $b_n=n^{-s/(2s+d)}$. With tuning based on $a_n=n^{-s/(2s+d)}+r_n+\sqrt{\log n/n}$ and $\bar a_n=a_n+\varepsilon_n$,
\[
\|m(\cdot;\widehat\alpha_n)\|_2+
\|T_0(\widehat\alpha_n-\alpha_0)\|_2
=O_p\{n^{-s/(2s+d)}+r_n+\varepsilon_n+\xi_n\}.
\]
\end{corollary}
The critical-radius assertion is an upper bound. The localized spectral calculation is given in Appendix~\ref{app: rkhs rate proofs}.
A coefficient rate requires one further identification step. The following uniform-rank condition gives a linear derivative link. More general derivative Gram and full spectral links are stated and proved in Appendix~\ref{app: spectral link}.

\begin{corollary}[Stone-type estimation and regret]\label{cor: RKHS convergence rate}
Suppose the conditions of Corollary~\ref{cor: RKHS derivative convergence rate} and Theorem~\ref{thm: regret bound} hold. If, in addition,
\begin{equation}
\kappa_A:=\operatorname*{ess\,inf}_x\sigma_{\min}\{A(x)\}>0,
\label{eq: uniform price rank}
\end{equation}
then Lemma~\ref{lem: price moment stability} gives $\|\alpha-\alpha_0\|_2\leq C\|T_0(\alpha-\alpha_0)\|_2$ on the fixed ball and, with the same tuning,
\begin{align}
\|\widehat\alpha_n-\alpha_0\|_2
&=O_p\{n^{-s/(2s+d)}+r_n+\varepsilon_n+\xi_n\},\label{eq: stone coefficients}\\
\mathcal V(\pi^*)-\mathcal V(\widehat\pi_n)
&=O_p\{n^{-2s/(2s+d)}+r_n^2+\varepsilon_n^2+\xi_n^2\}.\label{eq: stone regret}
\end{align}
The constants depend on instrument strength, including $\kappa_A^{-2}$ in the regret conversion.
\end{corollary}
When $r_n+\varepsilon_n+\xi_n=O\{n^{-s/(2s+d)}\}$, coefficient error has the familiar Stone-type norm exponent, while regret has the squared exponent \citep{stone1982optimal}. The localization calculation uses standard RKHS eigenvalue bounds \citep{mendelson2003performance,dikkala2020minimax}.

Condition~\eqref{eq: uniform price rank} means that instrument-induced variation separates the two price-response coefficients throughout the customer population. It implies $\|\alpha-\alpha_0\|_2\leq C\|T_0(\alpha-\alpha_0)\|_2$ on the fixed ball. This identification condition is stronger than almost-sure rank and differs from the reference nondegeneracy in Assumption~\ref{ass: identification and spaces}(c). We also consider other identification conditions in the Appendix. Appendix~\ref{app: spectral link} gives derivative Gram matrices, their cross-block condition, and the resulting polynomial and exponential rate conversions. Appendix~\ref{app: weak identification} controls the population share of weakly identified customer profiles. Example~\ref{ex: thin weak rank pricing} verifies that condition for a bounded pricing design. 

\begin{remark}[Root-$n$ revenue regret]\label{rem: root n regret} ~

\begin{enumerate}[label=(\roman*)] 
\item \emph{Conditions for root-$n$ regret.} Under the conditions of Theorems~\ref{thm: convergence rate}--\ref{thm: regret bound}, $R(\widehat\pi_n)=O_p(n^{-1/2})$ if, on the fixed ball supplied by the convergence theorem, for every fixed $C_0>0$,
\begin{equation*}
\omega_{M,n}^D(C_0\delta_n)=O(n^{-1/4}).
\end{equation*}
Under uniform rank~\eqref{eq: uniform price rank}, $b_n,r_n,\varepsilon_n,\xi_n=O(n^{-1/4})$ suffice. In the RKHS setting of Assumption~\ref{ass: RKHS eigen-decay}, $b_n=n^{-s/(2s+d)}$ with $s>d/2$ satisfies this requirement, so under uniform rank it suffices that $r_n+\varepsilon_n+\xi_n=O(n^{-1/4})$.

\item \emph{Results under other link conditions.} Similar Stone-type regret bounds hold under other link conditions, with rates that depend on identification strength. Appendix~\ref{app: root n proof} gives these bounds and the corresponding root-$n$ thresholds.
\end{enumerate}
\end{remark}

\section{Numerical Studies}
\label{sec: numerical}
We conduct numerical studies to evaluate the performance of the proposed method, PRINT, in terms of revenue regret. We compare the PRINT method with two benchmark offline pricing policy learning methods in various simulation settings and on a dataset from an online auto loan company.
\begin{enumerate}
    \item \citet{kallus2018policy} (and also \citet{chen2016personalized}) consider an inverse propensity score-based policy
learning method for continuous treatment without unmeasured confounding. The method is implemented by first estimating the generalized
propensity score $f_{P\given X,G}(p\given x,g)$, which is given by the ratio of two kernel
density estimators $\wh{f}_{P\given X,G}(p\given x,g) = \wh{f}_{P,X,G}(p,x,g)/\wh{f}_{X,G}(x,g)$.
Then, one can  learn a linear policy
$\wh{\pi}_{KZ}$ that maximizes the estimated value
\begin{equation*}
\sum_{i=1}^n Y_i \frac{K \left( h^{-1}[P_i - \pi_{KZ}(X_i)] \right)}{\wh{f}_{P\given X,G}(P_i\given X_i,G_i)} \Bigg/\sum_{i=1}^n \frac{K \left( h^{-1}[P_i - \pi_{KZ}(X_i)]\right)}{\wh{f}_{P\given X,G}(P_i\given X_i,G_i)}.
\end{equation*}
\item
We also compare with a regression-based method using the model
\begin{align*}
\EE[Y \given P,G,X] &= \beta_{p,1}^{reg}(X) P + \beta_{p,2}^{reg}(X) P^2 + \beta_g^{reg}(X,G),
\end{align*}
where $\beta_{p,1}^{reg}$, $\beta_{p,2}^{reg}$, and $\beta_{g}^{reg}$ are some nonparametric functions.
Then, the estimated optimal policy is
\begin{equation*}
\wh{\pi}_n^{reg}(X) = \clip{p_1}{p_2}{- \wh{\beta}^{reg}_{p,1}(X)/[2\wh{\beta}^{reg}_{p,2}(X)]}.
\end{equation*}
\end{enumerate}

\subsection{Simulation}
\label{sec: simulation}
\subsubsection{Simulated Data Generation}
We first generate $(G,X,U)$ to satisfy Assumption~\ref{ass: identification}.
Specifically, let $X \sim \mN((0.25,0.25)^{\top}, \Sigma_x=\bI_2)$.
Then we generate $G$ and $U=(U_1,U_2)$ by
\begin{equation*}
  (G, U_1, U_2) \given X
\sim \mN \big(
  [\mu_g + c_g^{\top}X,
  \mu_{u_1} + c_{u_1}^{\top}X,
  \mu_{u_2} + c_{u_2}^{\top}X]^{\top},
  \text{diag} (\sigma_g^2, \sigma_{u_1}^2,\sigma_{u_2}^2)
\big).
\end{equation*}

According to structural equations \eqref{Model: counterfactual} and \eqref{Model: action}, together with Assumption~\ref{ass: process independence}, we generate the price $P$ and revenue $Y$ using the previously generated $(G,X,U)$.
The coefficient functions are
\begin{align*}
& \beta_{p,1}(U,X)  = U_1^2-c_1^{\top}X, & \beta_{p,2} \equiv -c_{2,0}<0,\\
& \beta_g(U,X,G) = (U_1^2+c_3^{\top}X)G + c_4G, & \beta_{u,x}(U,X) = \kappa_1 U_2 + \kappa_2 U_2^2 + \cos(U_1U_2 + c_5^{\top}X),\\
& \alpha_g(X,G)= (c_6^{\top}X + c_7)G, & \alpha_{u,x}(U,X) = a\, U_2.
\end{align*}

They are constructed to satisfy Assumption \ref{ass: identification}.
The curvature $\beta_{p,2}$ is a negative constant, so the personalized price varies with $X$ through $\beta_{p,1}(X)$. Conditional on $X$, the diagonal covariance gives $U_1\independent U_2$ and $G\independent(U_1,U_2)$. The coefficient $\beta_{p,1}$ and the $U$-dependent part of $\beta_g$ depend on $U_1$, whereas the conditional price moments depend on $U$ through $U_2$. Consequently, the covariance restrictions in Assumption~\ref{ass: identification}\ref{ass:Orthogonality condition} hold by construction. Based on structural equations \eqref{Model: counterfactual} and \eqref{Model: action}, we add two independent noises to generate $Y$ and $P$, respectively, i.e.,
\begin{align*}
Y(p) &=\beta_{p,1}(U,X)p+\beta_{p,2}(U,X)p^2+\beta_g(U,X,G)+\beta_{u,x}(U,X)+\epsilon_y, \\
	P &= \alpha_g(X, G) + \alpha_{u, x}(U, X) + \epsilon_p,
\end{align*}
where $\epsilon_y\sim \text{Uniform}[-1,1]$, and the price noise $\epsilon_p = c_{\mathrm{het}}\, G\, \zeta + \nu$ is heteroskedastic in the instrument, with $\zeta\sim\mathcal{N}(0,1)$ and $\nu\sim\text{Uniform}[-1,1]$ independent of $(U,G,X)$. The $G$-dependent dispersion supplies variation in the first moment column, while the instrument-dependent component $\alpha_g(X,G)$ of the conditional mean is free of $U$. Identification also requires the nonzero determinant in Corollary~\ref{cor: identification 2}.

Under the above setting, an optimal pricing policy is given below, which is used to compute the regret.
\begin{equation*}
\pi^{*}(X) = \clip{p_1}{p_2}{\frac{(\mu_{u_1}+c_{u_1}^{\top}X)^2 + \sigma_{u_1}^2 - c_1^{\top} X}{2\,c_{2,0}}}.
\end{equation*}

To control the  magnitude of the direct instrument effect, in the function $\beta_g(U,X,G)$, we set $c_4=1$ for a mild direct effect and $c_4=5$ for a severe direct effect.
To vary the magnitude of instrument-dependent price dispersion, we set $c_{\mathrm{het}}=3$ and $c_{\mathrm{het}}=4.5$ for weaker and stronger dispersion designs.
The values of other parameters and implementation details are provided in Section~\ref{sec: simulation details} in the Appendix.

\subsubsection{Simulation Results}

We compare PRINT against two benchmarks that do not correct for the unobserved confounding: (i) the direct regression benchmark above, which is quadratic in the price with nonparametric coefficient functions of $(X,G)$ and does not correct for the unobserved confounder $U$, and (ii) the continuous-action policy-learning method of \citet{kallus2018policy}, which performs kernel off-policy optimization weighted by a generalized propensity score under an unconfoundedness assumption. For each method and sample size $n$ we learn $\wh\pi$ over $50$ replicates and evaluate its regret $\mV(\pi^*)-\mV(\wh{\pi})$ on a fixed test set of size $20{,}000$.

Figure~\ref{fig: simu regret} reports regret against sample size on log--log axes for the four combinations of moment relevance and direct effect severity. PRINT's regret decreases with sample size, whereas the two benchmarks that ignore the constructed unobserved confounding remain comparatively flat. This pattern is consistent with the identification result and with asymptotic bias in the benchmark procedures under the simulation design.
\begin{figure}[!htbp]
  \centering
  \includegraphics[width=0.7\textwidth]{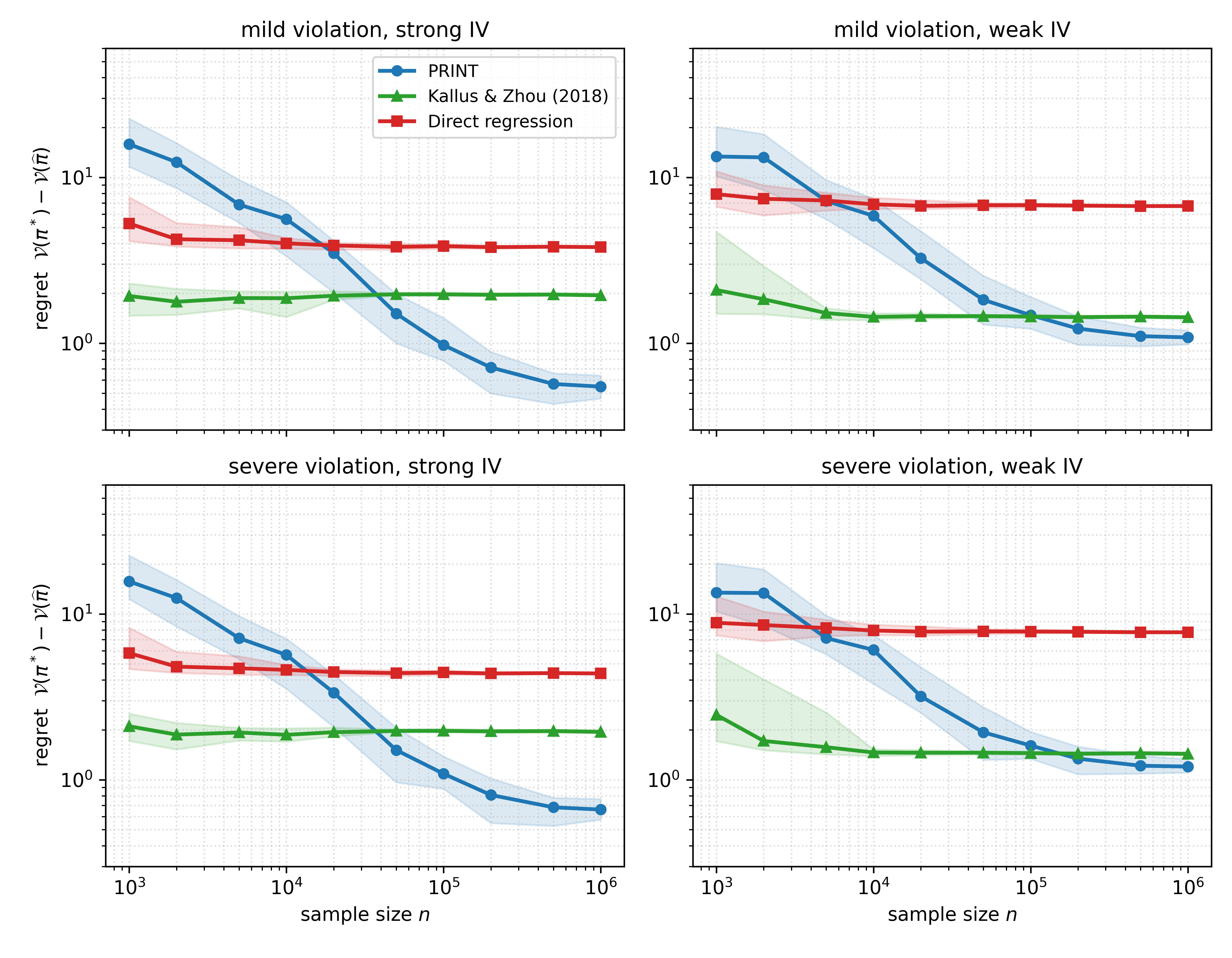}
    \caption{Regret $\mV(\pi^*)-\mV(\wh\pi)$ versus sample size $n$ (log--log scale).}
	    \label{fig: simu regret}
\end{figure}

\subsection{Auto-Loan Application}\label{subsec: real data}
In this subsection, we study the numerical performance of the pricing strategy
for personalized loans from an anonymous US auto lending company. We compare our
method with that of \citet{kallus2018policy}, a regression-based method, and
the historical decisions made by the company. A major difference between the real data application and the
previous simulation study is that Models \eqref{Model: counterfactual} and \eqref{Model: action}, the standing potential-outcome path regularity, and the latent process-level conditional independence in Assumption~\ref{ass: process independence} are potentially misspecified in the real data.

We obtain the dataset \textit{CPRM-12-001: On-Line Auto Lending}  from the
Center for Pricing and Revenue Management at Columbia University. It records the
online auto loan applications received by the company from Jul. 2002 to Nov.
2004. For each approved application, the requested term, loan amount, APR, monthly LIBOR,
whether contracted or not, and some personal information (e.g., FICO score) are recorded.
For detailed descriptions, we refer readers to \citet{phillips2015effectiveness} and \citet[][Table 3]{ban2021personalized}.

\subsubsection{Problem Settings and Evaluation Method}
For the pricing of the online auto loan company, we use a price definition adapted from \citet{ban2021personalized}, the value of the monthly payments accumulated at the monthly LIBOR rate over the term of the loan, subtracted by the loan amount, 
\begin{equation}
\label{eq: data price def}
P = \text{Monthly Payment} \times (1+r)\,\frac{(1+r)^{\text{Term}} - 1}{r} - \text{Loan Amount}, \qquad r = \text{monthly LIBOR}.
\end{equation}
In addition, we set the feasible price range to be $[\$0, \$40,000]$.

To evaluate the pricing policy, it is necessary to construct a generative model since the revenue depends on the price selected by the policy, which is not available in the dataset.
Since the outcomes of whether contracted or not are binary (accept/reject), instead of a linear
demand, we adopt a logistic demand model used by \citet{ban2021personalized} for generating the demand. Therefore the true expected revenue is not a quadratic function of the price, conditioning on other factors.
In particular, for a given feature vector $x$ (including the constant term) and a price $p$, the probability of accepting the contract, i.e., the expected demand, is modeled by
$
\frac{1}{1+ \exp \{ -\alpha^{\top} x - \beta^{\top}x\times p\}}.
$
Then the expected revenue is $p/[1+ \exp \{ -\alpha^{\top} x - \beta^{\top}x\times p \}]$.
The original evaluation fits this logistic model with an $\ell^2$ penalty using all records and selects the penalty by five-fold cross-validation.

\subsubsection{Implementation}
A competitor's rate is available in this dataset, although it may not fully represent competition during the analyzed period \citep{phillips2015effectiveness}. We deliberately exclude this observed variable to create an omitted-variable stress test.
Furthermore, given the loan amount, term, and LIBOR, the price of a loan can be purely determined by the monthly payment. Therefore, we exclude monthly payment because it deterministically encodes the constructed price. Conditioning on it would effectively condition on the treatment itself.

For PRINT, we use APR as a candidate price-setting IV. Interest rate variation has provided identifying variation in consumer-credit studies \citep{karlan2008credit,attanasio2008credit}. APR is strongly related to the constructed price and may also affect contract acceptance and revenue directly, so it need not be a valid IV in the classical sense. The implemented estimator additionally imposes $\beta_{p,g}=0$.  APR may have a price-independent direct effect through $\beta_g$, but it may not modify the marginal price effect.
For comparison, in addition to the two benchmark methods in Section~\ref{sec: simulation}, we consider the company's historical pricing policy and a fitted-model benchmark obtained by maximizing revenue under the learned logistic demand model.
The first 60,000 time-ordered records are used for policy training. For PRINT, 50,000 records fit the estimator and 10,000 are used for moment-based early stopping. The remaining 148,085 records form the test dataset.

\subsubsection{Results and Discussion}

Table \ref{tab: data values} summarizes the expected revenue of the pricing policy learned by PRINT compared to the two benchmark learning methods, the model-based benchmark policy, and the policy used by the firm.
The firm's historical policy attains 75.2\% of the revenue that would be achieved by adopting the  model-based benchmark policy.
This satisfactory revenue is reasonable since the firm may have external information at the pricing time that is not available in the offline records.
The expected revenues obtained by the policies learned by direct regression and by \citet{kallus2018policy} do not reach the firm's historical revenue, partially due to the unobserved confounding issue. 
In model-based counterfactual evaluation, the policy learned by PRINT attains about 80.6\% of the expected revenue of the benchmark policy and exceeds the evaluated historical policy by 7.2\%. Modest model bias can still yield strong performance, consistent with \citet{besbes2010testing}.

\begin{table}[!htbp]
  \centering
  \caption{Fitted-model revenue scores on the testing dataset.}
  \label{tab: data values}
  \begingroup
  \small
  \setlength{\tabcolsep}{4pt}
  \begin{tabular}{@{}l|ccccc@{}}
  \toprule
      Pricing Policy &  Oracle & Firm   & \textbf{PRINT} & \citet{kallus2018policy} & Regression \\
  \midrule
      Revenue (\$1000)         & \textbf{1.1062}  & 0.8318 & \textbf{0.8921}   & 0.1452                   & 0.0456     \\
  \bottomrule
  \end{tabular}
  \endgroup
\end{table}

\begin{figure}[!htbp]
  \centering
  \includegraphics[width=0.4\textwidth]{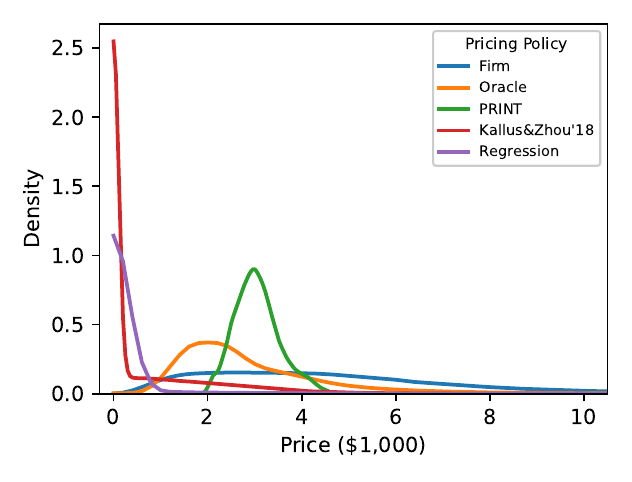}
    \caption{Personalized Pricing distributions on the testing dataset.}
  \label{fig: data price distribution}
\end{figure}

\begin{figure}[htbp]
  \centering
    \begin{subfigure}[b]{0.245\textwidth}
      \centering
      \includegraphics[width=\textwidth]{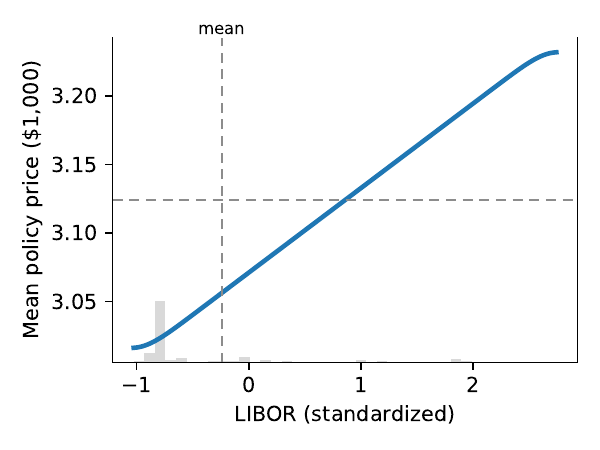}
      \caption{LIBOR}
    \end{subfigure}
    \begin{subfigure}[b]{0.245\textwidth}
      \centering
      \includegraphics[width=\textwidth]{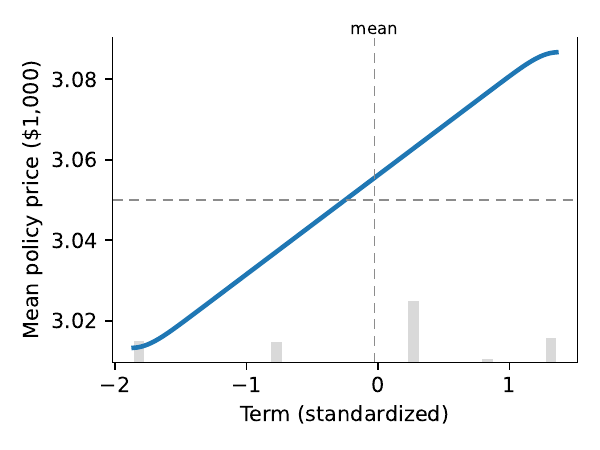}
      \caption{Term}
    \end{subfigure}
    \begin{subfigure}[b]{0.245\textwidth}
      \centering
      \includegraphics[width=\textwidth]{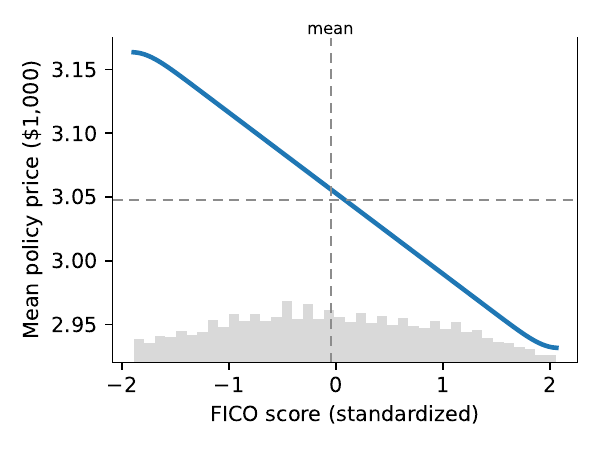}
      \caption{FICO score}
    \end{subfigure}
    \begin{subfigure}[b]{0.245\textwidth}
      \centering
      \includegraphics[width=\textwidth]{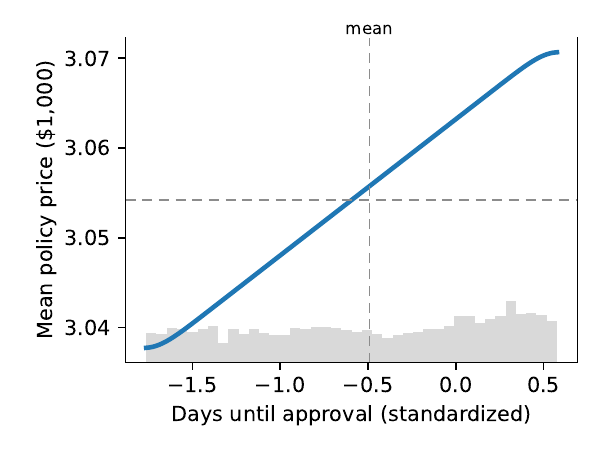}
      \caption{Days until approval}
    \end{subfigure}
    \caption{Partial dependence plots of the pricing strategy learned by PRINT on
      important factors. All factors are centered at their means and scaled by their standard deviations. Histograms are distributions of the standardized
      factors and blue lines are the partial dependence of the recommended price on each factor.}
    \label{fig: data partialdep}
\end{figure}

In Figure \ref{fig: data price distribution}, we provide the density plots of the prices by the five pricing policies on the testing dataset. The PRINT policy prices between roughly \$1,800 and \$5,400.  Among the learned policies, PRINT's price distribution is by far the closest to the benchmark's. In contrast, the policies learned by direct regression and by \citet{kallus2018policy} place most of their mass near zero prices, while the firm's historical prices are more dispersed and higher on average.

To study the interpretability of the policy learned by the PRINT method, in Figure \ref{fig: data partialdep}, we provide the partial dependence plots of the learned policy on the four most important features. First, as shown in Figure \ref{fig: data partialdep}(a) and (b), the price increases with LIBOR and term, which agrees with the definition of the price in \eqref{eq: data price def}. Second, Figure \ref{fig: data partialdep}(c) shows that for customers with higher FICO scores, we suggest lower prices as those customers have lower risks of bad debts. Finally, as shown in Figure \ref{fig: data partialdep}(d), our policy tends to set higher prices for applications with longer processing times. This is likely due to the fact that those applications have potential risks that require longer review processes.

\section{Conclusions}
\label{sec: conclusion}

To conclude, we study offline personalized pricing under endogeneity by
leveraging a candidate invalid price-setting instrument. Our target class is partial-information offline pricing, where the historical pricing rule does not act on the same latent demand state that determines marginal price response, rather than fully informed monopoly pricing. The managerial lesson is that an imperfect price-setting lever such as APR-type variation can still be used to identify the price-response functions when this partial-information orthogonality is credible and the residualized higher-price-moment system satisfies the operative full-rank condition.

The key challenges are (a)
identification of the heterogeneous effect of continuous price on revenue under unmeasured
confounding, (b) dealing with a possibly invalid IV that may violate the exclusion restriction, and (c) solving conditional moment restrictions of generalized residual functions. For (a) and (b), we generalize the
identification results in the causal inference literature on relaxing exclusion
restriction for discrete treatment to continuous treatment. For (c), we develop
an adversarial min-max algorithm for learning the optimal pricing strategy.
Theoretically, we establish the consistency and the convergence rate of the proposed policy learning algorithm, as well as an asymptotic bound on the regret of the learned pricing policy. 

For future work, it would be interesting to extend our work to learning a multistage policy from offline data under endogeneity with
invalid IVs.
Discrete-demand settings (e.g., assortment and multinomial logit choice models) call for different identification and learning tools than our continuous price IV framework. See, for instance, \citet{kofisher2007demand}, \citet{farias2013nonparametric}, and \citet{talluri2004theory} for approaches tailored to discrete choice. Extending invalid IV ideas to those discrete models is an interesting direction for future research.

\bibliographystyle{plainnat}
\bibliography{reference}

\clearpage
\appendix
\section*{Appendices}
\addcontentsline{toc}{section}{Appendices}
\renewcommand{\theHsection}{app.\arabic{section}}
\renewcommand{\theHsubsection}{app.\arabic{section}.\arabic{subsection}}
\renewcommand{\theHsubsubsection}{app.\arabic{section}.\arabic{subsection}.\arabic{subsubsection}}
\newcounter{appendixstatement}[section]
\renewcommand{\theappendixstatement}{\thesection.\arabic{appendixstatement}}
\makeatletter
\def\theHappendixstatement{\theHsection.\arabic{appendixstatement}}
\@for\appendixtheorem:=theorem,lemma,proposition,corollary,claim,conjecture,hypothesis,assumption,remark,example,problem,definition,question,answer,exercise\do{%
  \expandafter\let\csname c@\appendixtheorem\endcsname\c@appendixstatement
  \@addtoreset{\appendixtheorem}{section}%
  \expandafter\let\csname the\appendixtheorem\endcsname\theappendixstatement
  \expandafter\let\csname theH\appendixtheorem\endcsname\theHappendixstatement
}
\makeatother
\allowdisplaybreaks

\section{A Detailed Practical Algorithm}
\label{app:practical-algorithm}
In this section, we give a detailed practical implementation of  the min-max estimator when the function classes $\mH$ and $\mF$ are neural
networks (NNs), in which the functions can be parameterized by some NN parameters.
Specifically, suppose that $\mH = \{ \alpha_{\theta_{\alpha}}(\bullet) = \alpha(\bullet;\theta_{\alpha}): \theta_{\alpha}\in\Theta_{\mH} \}$
and that $\mF = \{ f_{\theta_f}(\bullet)=f(\bullet;\theta_f): \theta_f\in\Theta_{\mF} \}$ for some Euclidean subspaces
$\Theta_{\mH}$ and $\Theta_{\mF}$. In the case of a large sample size, one can apply the stochastic
gradient ascent/descent with random minibatches from the full dataset
$\mD_n$ in each iteration. Instead of updating
$\wh{f}_n^{(k)}$ and $\wh{\alpha}_n^{(k)}$ to their exact inner and outer optima, we apply one gradient step to each player per iteration. See details in Algorithm \ref{alg:
  detailed min-max}.

\begin{algorithm}
  \caption{Estimation of $\wh{\pi}_n$ by solving a zero-sum game.}
  \label{alg: detailed min-max}

  \textbf{Input:} Batch data
  $\mD_n = \left\{ Z_i=(X_i,G_i,P_i,Y_i) \right\}_{i=1}^n$. Cross-fitted first-stage estimates $\wh h=(\wh h_1,\wh h_2,\wh h_3)$. Initial
  $\wh{\theta}_f^{(-1)},\wh{\theta}_{\alpha}^{(0)}$. Feasible parameter set $\Theta_{\mH}(c_{p,2})$ chosen so that the curvature component satisfies $-\beta_{p,2}(x)\geq c_{p,2}$ for all $x$. Tuning
  parameters $\lambda_{n_B},\mu_{n_B}>0$. Price range $[p_1,p_2]$. Maximum iteration $K$. Batch
  size $1\leq n_B \leq n$. Step sizes $\gamma_{\alpha},\gamma_f>0.$\\
  \textbf{For} $k\in \left\{ 0,\dots, K-1 \right\}$:\\
  \Indp
  Randomly sample a  mini-batch $D_I = \left\{ Z_i \right\}_{i\in I}$ with $|I| = n_B$.
  \[\wh{\theta}_f^{(k)} \leftarrow \wh{\theta}_f^{(k-1)}+\gamma_f \nabla_{\theta_f}\left[ \wh\Psi_{n_B}(\alpha_{\wh{\theta}_{\alpha}^{(k)}},f_{\wh{\theta}_f^{(k-1)}}) - \norm{f_{\wh{\theta}_f^{(k-1)}}}_{n_B,\alpha_{\wh{\theta}_{\alpha}^{(k)}}}^2 - \lambda_{n_B}\nmF{f_{\wh{\theta}_f^{(k-1)}}}^2 \right].\]\\
  $\wh{\theta}_{\alpha}^{(k+1)}\leftarrow \operatorname{Proj}_{\Theta_{\mH}(c_{p,2})}\!\left\{\wh{\theta}_{\alpha}^{(k)} - \gamma_{\alpha}\nabla_{\theta_{\alpha}} \left[ \wh\Psi_{n_B}(\alpha_{\wh{\theta}_{\alpha}^{(k)}}, f_{\wh{\theta}_f^{(k)}}) + \mu_{n_B}\nmH{\alpha_{\wh{\theta}_{\alpha}^{(k)}}}^2 \right]\right\}$.\\
  \Indm
  \textbf{End For}\\
  \textbf{Output:} Pricing policy
  $\wh{\pi}_n = \clip{p_1}{p_2}{- \wh{\beta}_{p,1}^{(K)}/[2\wh{\beta}_{p,2}^{(K)}]}$.
\end{algorithm}
Here $(\wh\beta_{p,1}^{(K)},\wh\beta_{p,2}^{(K)})$ are the first two coordinates of $\alpha_{\wh\theta_\alpha^{(K)}}$, and $\mH$ is restricted so that the estimated curvature satisfies $-\wh\beta_{p,2}^{(K)}\geq c_{p,2}$ (as required by Theorem~\ref{thm: regret bound}). Here $\wh\Psi_{n_B}$ and $\norm{f}^2_{n_B,\alpha}$ are the mini-batch versions of the feasible criterion components of \eqref{eqn: min-max estimation}, computed with the cross-fitted $\wh h$ plugged into $W$. Algorithm~\ref{alg: detailed min-max} is a practical heuristic: its $f$-update uses the current iterate $\alpha_{\wh\theta_\alpha^{(k)}}$ in place of the fixed reference $\wt\alpha$, it uses generally nonconvex finite-width neural-network classes. 

\section{Details of Numerical Study} \label{sec: simulation details}
Following the data generating procedure in Section \ref{sec: simulation},
the ground truth nuisance functions of policy are
\begin{align*}
\wt{\beta}_{p,1}(X) & = \EE \left[ \beta_{p,1}(U,X)\given X \right] = -c_1^{\top}X + (\mu_{u_1}+c_{u_1}^{\top}X)^2 + \sigma_{u_1}^2,\\
\wt{\beta}_{p,2}(X) & = \EE \left[ \beta_{p,2}\given X \right] = -c_{2,0}.
\end{align*}
Then the oracle optimal pricing policy is
\begin{align*}
\pi^{*}(X) &= \clip{p_1}{p_2}{-\wt{\beta}_{p,1}(X)/[2\wt{\beta}_{p,2}(X)]} \\
&= \clip{p_1}{p_2}{\frac{(\mu_{u_1}+c_{u_1}^{\top}X)^2 + \sigma_{u_1}^2 - c_1^{\top} X}{2\,c_{2,0}}}.
\end{align*}

All parameters used for the data generating procedure are listed in Table \ref{tab: parameters}.

The simulation results in Section~\ref{sec: simulation} estimate $\beta_{p,1}(X)$ and $\beta_{p,2}(X)$ by solving the identifying conditional-moment restrictions of Corollary~\ref{cor: identification 2}, using the heteroskedasticity-based instrument set $\{G-\EE[G\mid X],G(G-\EE[G\mid X]),G^2(G-\EE[G\mid X])\}$, the price-response functions modeled in a quadratic basis of $X$, and the curvature restriction $\wh\beta_{p,2}(X)\le -c_{p,2}$. This is the finite-dimensional, closed-form instance of the min--max estimator of Section~\ref{sec: estimation and policy learning} (a linear adversary spanned by the moment instruments), and is the computationally light route used in the experiments. The neural-network adversarial implementation of Appendix~\ref{app:practical-algorithm} targets high-dimensional nuisances, and the regret guarantees of Section~\ref{sec: theory} apply to the min--max estimator.

\begin{table}[H]
  \caption{Simulation data parameters.}
  \label{tab: parameters}
\resizebox{\columnwidth}{!}{
  \begin{tabular}{c|ccccccc}
	\toprule
	Parameter & $\mu_x$              & $\Sigma_x$              & $\mu_g$  & $\mu_{u_1}$ & $\mu_{u_2}$ & $c_g$ & $c_{u_1}$ \\
	\midrule
	Value     & $[0.25, 0.25]^{\top}$ & $\text{diag}(1,1)$ & $2$    & $0.5$    & $0.3$ & $[0.25,0.25]^{\top}$ & $[0.3,0.4]^{\top}$  \\
	\midrule
	\midrule
	Parameter    & $c_{u_2}$ & $\sigma_g^2$ & $\sigma_{u_1}^2$ & $\sigma_{u_2}^2$        & $c_1$              & $c_{2,0}$     & $c_3$ \\
	\midrule
	Value        & $[0.2,0.2]^{\top}$ & 1 & 3 & 1                 & $[0.3,0.2]^{\top}$    & $0.2$ & $[0.2,-0.1]^{\top}$\\
	\midrule
	\midrule
    Parameter    & $c_4$ & $c_5$             & $c_6$              & $c_7$      & $c_{\mathrm{het}}$        & $(\kappa_1,\kappa_2)$ & $a$ \\
	\midrule
    Value      & 1 or 5 & $[0.4,0.1]^{\top}$   & $[1.2, 0.4]^{\top}$   & 4     & 3 or 4.5        & $(50,5)$ & 3 \\
	\midrule
	\bottomrule
  \end{tabular}
}
\end{table}

\section{Running Example and Higher-Moment Relevance}
\label{app: assumption verification}

This appendix verifies the structural assumptions for the running example, states its rank requirement, and explains the role of higher-price moments. The process regularity used in the model verification is given in Appendix~\ref{app: process bridge}.

\subsection{Model Specification and Potential Revenue}
\label{app: running example details}

In Example~\ref{ex:running}, the noises $\epsilon_d$ and $\epsilon_y$ are mean zero and independent of $(U,G,X,\epsilon_p)$, while $\epsilon_p$ satisfies $\EE[\epsilon_p\mid G,X]=0$ and is independent of $U$ given $(G,X)$. Its conditional dispersion $\Var(\epsilon_p\mid G,X)$ varies with $G$. We assume $\gamma_1\neq0$ and $\eta_2\neq0$, $G$ is nondegenerate given $X$, $\Var(\epsilon_p\mid G,X)>0$ almost surely, and $U^a$, $U^p$, $\epsilon_p$ have finite 3rd conditional moments.

Setting the price in the demand and revenue equations gives
\[
Y(p)=(\eta_0+\eta_1p+\eta_2G+\eta_3^{\top}X+\eta_4U^a+\epsilon_d)p+\epsilon_y.
\]
Substitution into $Y=D\cdot P+\epsilon_y$ recovers Models~\eqref{Model: counterfactual} and~\eqref{Model: action}, with
\[
\beta_{p,1}(U,X)=\eta_0+\eta_3^{\top}X+\eta_4U^a,\qquad
\beta_{p,g}(X,G)=\eta_2G,\qquad
\beta_{p,2}(U,X)=\eta_1,
\]
and $\beta_g=\beta_{u,x}=0$. The quadratic revenue path is continuous and has an integrable supremum on the compact $\calP$ under the stated moment conditions. Because $(\epsilon_d,\epsilon_y)\independent \epsilon_p \given (G,X,U)$, the entire path $Y(\cdot) \independent P \given (G,X,U)$, as required by Assumption~\ref{ass: process independence}. Then the consistency follows by evaluating the potential outcome at $P$. Thus $\beta_{p,2}=\eta_1<0$, as required by the curvature bound. Because $\eta_2G$ enters the coefficient on $P$, APR shifts revenue even when price is held fixed. If $\epsilon_p$ has a positive conditional density on $\RR$ given $(G,X)$, the conditional law of $G$ given $(P,X,U)$ is nondegenerate, and hence $G\independent Y\mid(P,X,U)$ fails.

\subsection{Verification of Partial-Information Orthogonality}
\label{app: running example orthogonality}

The historical price responds to $U^p$ but not $U^a$. Under $U^a\independent U^p\mid X$ and $G\independent U\mid X$, this separation yields the partial-information orthogonality verified below.
We verify the conditions of Assumption~\ref{ass: identification} for the partial-information auto-loan model in Example~\ref{ex:running}.

We first check instrument assignment and relevance.
The maintained assignment condition is $G\independent(U^a,U^p)\mid X$.
The historical price has conditional mean
$\gamma_0+\gamma_1G+\gamma_2^\top X+\gamma_3U^p$
given $(G,X,U)$. Since $\gamma_1\ne0$ and $G$ is nondegenerate given $X$, IV independence also makes $G$ nondegenerate given $(X,U)$. Its nonconstant conditional-mean effect therefore implies distributional relevance $G\dependent P\mid(X,U)$.

Then we show which latent component enters the price moments.
Write $A(G,X,U^p)=\gamma_0+\gamma_1G+\gamma_2^\top X+\gamma_3U^p$.
The price-noise independence assumed in Appendix~\ref{app: running example details} gives, for $k=1,2,3$,
\[
\EE[P^k\mid G,X,U]
=\sum_{\ell=0}^k {k\choose\ell}
 A(G,X,U^p)^{k-\ell}\EE[\epsilon_p^\ell\mid G,X].
\]
Thus each required moment is a function of $(G,X,U^p)$ alone. 

Then we factor the relevant conditional distribution.
The two maintained independence statements imply
\[
P_{U^a,U^p\mid X,G}=P_{U^a,U^p\mid X}
=P_{U^a\mid X}\otimes P_{U^p\mid X}
\quad\text{almost surely}.
\]
Hence, conditional on $(X,G)$, the coefficient
$\beta_{p,1}(U,X)=\eta_0+\eta_3^\top X+\eta_4U^a$
is independent of every price-moment function $\EE[P^k\mid G,X,U]$. Its conditional covariance with those moments is zero, in particular for the required orders $1,2$.
The coefficient $\beta_{p,2}=\eta_1$ is constant, so its covariance is zero for the required orders $2,3$. Finally $\beta_g=0$ in this example, so~\eqref{ass: ortho_alpha} holds immediately.

\subsection{Rank Conditions for the Running Example}
\label{app: running example rank}

For higher-moment variation, let $\Var(\epsilon_p\mid G,X)$ be nonconstant and monotone in $G$, so that $\Omega_2(X)=\Cov\{G,\Var(P\mid G,X)\mid X\}\neq0$. This makes one coefficient column nonzero, but does not establish the rank needed for identification. To state the remaining condition explicitly, write $\bar G(X)=\EE[G\mid X]$, $R=P-\EE[P\mid G,X]$, $C_0(V)=\EE[(G-\bar G(X))V\mid X]$, and $C_k(V)=\Cov\{G^k(G-\bar G(X)),V\mid X\}$ for $k=1,2,3$. Since $\beta_{p,g}(X,G)=\eta_2G$, Theorem~\ref{thm: identification} uses $B_1(x,g)=g$, $K=1$, and the complete coefficient matrix
\[
M_{\rm run}(X)=\begin{pmatrix}
C_0(RP)&C_0(RP^2)&C_0(RPG)\\
C_1(RP)&C_1(RP^2)&C_1(RPG)\\
C_2(RP)&C_2(RP^2)&C_2(RPG)\\
C_3(RP)&C_3(RP^2)&C_3(RPG)
\end{pmatrix}.
\]
We therefore additionally maintain $\operatorname{rank}\{M_{\rm run}(X)\}=3$ almost surely. Monotone dispersion alone does not imply this condition. In the no-interaction specialization $\eta_2=0$, Corollary~\ref{cor: identification 2} applies only under the separately maintained condition $\det A(X)=\Omega_2(X)\Upsilon_3(X)-\Omega_3(X)\Upsilon_2(X)\neq0$ almost surely.

\subsection{Higher-Moment Calculations and Identification Counterexamples}
\label{app: higher moment details}

The following calculations explain why a mean shift and quantitative identification are different. Write
$\mu(X,G)=\EE[P\mid X,G]$, $R=P-\mu(X,G)$,
$v(X,G)=\EE[R^2\mid X,G]$, and $\tau(X,G)=\EE[R^3\mid X,G]$.

Since $\EE[R\mid X,G]=0$ and $P=\mu+R$, we have that
\begin{align*}
\EE[RP\mid X,G]
&=\EE[R(\mu+R)\mid X,G]=v(X,G),\\
\EE[RP^2\mid X,G]
&=\EE[R(\mu^2+2\mu R+R^2)\mid X,G]
=2\mu(X,G)v(X,G)+\tau(X,G).
\end{align*}
Let $H_0=G-\EE[G\mid X]$ and $H_1=GH_0$. The tower property and conditional covariance identities now give
\begin{align*}
\Omega_2&=\Cov(G,v\mid X),&\qquad
\Upsilon_2&=\Cov(H_1,v\mid X),\\
\Omega_3&=\Cov(G,2\mu v+\tau\mid X),&
\Upsilon_3&=\Cov(H_1,2\mu v+\tau\mid X).
\end{align*}
Thus a pure mean shift with $v$ constant in $G$ makes the entire first column zero. It can still change the second column through $2\mu v$.

Note that Condition~\eqref{eq: hmr} ensures $\Omega_2\ne0$, so it is sufficient to make the slope column nonzero. But it is not necessary, because with a symmetric instrument and even variance, $\Omega_2$ may vanish while $\Upsilon_2$ does not.

\begin{example}[Mean relevance does not identify the price slope]\label{ex: mean-insufficient}

Suppress $X$. Let centered bounded $G,U^a,U^p,\epsilon$ be mutually independent, with $\Var(G)>0$ and $\Var(U^p)+\Var(\epsilon)>0$. Take
\[
P=p_0+G+U^p+\epsilon,\qquad
\beta_{p,1}(U)=U^a,\qquad \beta_{p,2}=-1,\qquad \beta_{p,g}=0.
\]
Choose $p_0$ and a compact nonnegative feasible interval to contain this bounded price support. This level shift makes the illustration compatible with the maintained price domain. Let the baseline and direct-effect terms obey the stated structural restrictions. For example, a baseline depending on $U^p$ and a direct effect depending only on $G$ are allowed. Partial-information orthogonality holds because the marginal linear response depends on $U^a$, whereas the conditional price moments depend on $U^p$ and the quadratic coefficient is constant.

The instrument changes $\EE[P\mid G]=p_0+G$, but $R=P-\EE[P\mid G]=U^p+\epsilon$ is independent of $G$ and centered. Thus $v=\EE[R^2]$ is constant and $\Omega_2=\Upsilon_2=0$. To see what remains, expand
\[
\EE[RP^2\mid G]=2(p_0+G)\EE[R^2]+\EE[R^3].
\]
Multiplying by $G$ and averaging, using $\EE[G]=0$, gives
\[
\Omega_3=2\Var(G)\EE[R^2]>0.
\]
Hence the first identifying equation reduces to
$\Omega_1=\Omega_3\beta_{p,2}$, identifying the quadratic coefficient from these moments, while neither row of the system contains the linear coefficient. The system therefore does not identify the slope. 
\end{example}

\section{Foundations and Proofs of Identification}
\label{proof: identification}

We first justify evaluation of potential revenue at an observed price or a policy price. We then derive the optimal policy and prove the moment condition identities, the general identification theorem, and its no-interaction corollary.

\subsection{Potential-Revenue Processes and Random-Price Evaluation}
\label{app: process bridge}
This subsection gives the regularity underlying Assumptions~\ref{ass: standard} and~\ref{ass: process independence}. Let $Y(\cdot)$ be a measurable random element of a standard Borel space of real functions on $\calP$, with jointly measurable evaluation $(y,p)\mapsto y(p)$. Assume there is a measurable envelope $H(Y(\cdot))$ satisfying $\sup_{p\in\calP}|Y(p)|\leq H(Y(\cdot))$ and $\EE H<\infty$. The observed and latent conditioning variables take values in standard Borel spaces. The usual space of right-continuous paths with left limits, equipped with its Borel sigma algebra, is admissible. This permits individual acceptance thresholds and other jumps in realized revenue.

For a conditioning variable $S$, let $K_S(s,dy)$ be a regular conditional distribution of the entire path given $S=s$, and define the jointly measurable conditional-mean version
$q_S(s,p)=\int y(p)K_S(s,dy)$.
All structural mean models are understood as restrictions on this version, holding for every feasible $p$ $P_S$-almost surely. For the unconfounded benchmark of Section~\ref{sec: causal framework}, assume $q_X(x,\cdot)$ is continuous for almost every $x$. In the structural model of Section~\ref{sec: policy learning with invalid iv}, continuity of the conditional means follows instead from the jointly specified quadratic and the standing integrability conditions. However, we do not require the realized revenue path itself to be continuous.

\begin{lemma}[Conditional mean at a random price]\label{lem: random index}
Under the path-space regularity, if $Y(\cdot)\independent P\mid S$, then for every $\sigma(P,S)$-measurable price $T\in\calP$,
\begin{equation*}
\EE[Y(T)\mid P,S]=q_S(S,T).
\end{equation*}
For a $\sigma(S)$-measurable $T$, the identity $\EE[Y(T)\mid S]=q_S(S,T)$ needs no independence assumption involving $P$.
\end{lemma}
\subsubsection{Proof of Lemma~\ref{lem: random index}.}
Conditional independence makes $K_S(S,\cdot)$ a conditional law of the entire path $Y(\cdot)$ given $(P,S)$. Joint measurability of evaluation and the integrable envelope therefore give
\[
\EE[Y(T)\mid P,S]
=\int y(T)K_S(S,dy)
=q_S(S,T).
\]
For a $\sigma(S)$-measurable $T$, the same kernel formula conditional only on $S$ proves the second assertion without any independence assumption involving $P$.
\QED

Taking $S=X$, $T=P$, and using consistency, it is easy to show the observed conditional-mean identity for the unconfounded benchmark. Positive conditional density identifies its continuous-in-price version uniquely: two continuous functions that agree Lebesgue-almost everywhere on the feasible interval agree everywhere. Taking $T=\pi(X)$ then gives $\mathcal V(\pi)=\EE[q_X(X,\pi(X))]$. This explains the role of conditional-mean continuity, but we do not require sample-path continuity, in Section~\ref{sec: causal framework}.

Taking $S=(G,X,U)$ and $T=P$ proves that Assumption~\ref{ass: process independence}, consistency, and the counterfactual revenue model imply the observed revenue model. Taking $T=\pi(X)$ gives the policy-indexed mean. Assumption~\ref{ass: process independence} allows endogenous prices through $U$, but excludes assignment based on remaining path information after $(G,X,U)$ is held fixed. 

Throughout, $Y(p)$ denotes an intervention on price alone, leaving the pre-price $G$ and the reference law of $(U,G)\mid X$ unchanged. This is an estimand restriction, separate from conditional independence. 

\subsection{Proof of the Optimal Pricing Policy}
\label{app: optimal policy proof}

\subsubsection{Proof of Proposition~\ref{prop: optimal policy}.}

For any feasible policy $\pi$, Lemma~\ref{lem: random index}, Model~\eqref{Model: counterfactual}, and the tower property yield
\[
\calV(\pi)=C_0+\EE[\beta_{p,1}(X)\pi(X)+\beta_{p,2}(X)\pi(X)^2],
\]
where $C_0=\EE[\beta_g(U,X,G)+\beta_{u,x}(U,X)]$ is independent of $\pi$, and the coefficient averages are defined after Proposition~\ref{prop: optimal policy}. Since $\beta_{p,2}(X)\leq-c$, the conditional quadratic is strictly concave, with unique maximizer $\pi^*(X)=\clip{p_1}{p_2}{-\beta_{p,1}(X)/(2\beta_{p,2}(X))}$. This policy is measurable and belongs to $\Pi$. Integrating its pointwise optimality proves both formulas.
\QED

\subsection{Proofs of the Identifying Moment Identities}
\label{app: identifying moment proofs}

\subsubsection{Proof of Lemma~\ref{lm: identification eq 1}.}

For simplicity, define
\[
\mu=\EE[P\mid X,G],\quad R=P-\mu,\quad
M_j=\EE[P^j\mid X,G,U],\quad
\bar b_j=\EE[\beta_{p,j}(U,X)\mid X]\quad(j=1,2).
\]
Thus $\bar b_2=\beta_{p,2}(X)$.
Write $a(U,X)=\alpha_{u,x}(U,X)$ and $d_u(U,X)=\beta_{u,x}(U,X)$.
By consistency and Lemma~\ref{lem: random index} with $S=(G,X,U)$ and $T=P$,
\[
\EE[Y\mid P,G,X,U]
=\beta_{p,1}(U,X)P+\beta_{p,g}(X,G)P
 +\beta_{p,2}(U,X)P^2+\beta_g(U,X,G)+d_u(U,X).
\]

IV independence implies $\EE[\beta_{p,j}(U,X)\mid X,G]=\bar b_j$. For any required price-moment order $\ell$, conditioning first on $(X,G,U)$ gives
\begin{align*}
\EE[(\beta_{p,j}-\bar b_j)P^\ell\mid X,G]
&=\EE[(\beta_{p,j}-\bar b_j)M_\ell\mid X,G]\\
&=\Cov(\beta_{p,j},M_\ell\mid X,G).
\end{align*}
Therefore, 
\begin{align*}
\EE[(\beta_{p,1}-\bar b_1)RP\mid X,G]
&=\Cov(\beta_{p,1},M_2\mid X,G)
 -\mu\Cov(\beta_{p,1},M_1\mid X,G)=0,\\
\EE[(\beta_{p,2}-\bar b_2)RP^2\mid X,G]
&=\Cov(\beta_{p,2},M_3\mid X,G)
 -\mu\Cov(\beta_{p,2},M_2\mid X,G)=0.
\end{align*}
The first line uses exactly the two restrictions in~\eqref{ass: ortho_p1} and the second uses exactly those in~\eqref{ass: ortho_p2}. Consequently,
\[
\EE[\beta_{p,1}RP\mid X,G]=\bar b_1\EE[RP\mid X,G],\qquad
\EE[\beta_{p,2}RP^2\mid X,G]=\bar b_2\EE[RP^2\mid X,G].
\]

We then separate the direct instrument effect and baseline confounding.
For any integrable $d=d(U,X,G)$ entering the following products,
\[
\EE[Rd\mid X,G]
=\EE[(M_1-\mu)d\mid X,G]
=\Cov(M_1,d\mid X,G),
\]
because $\EE[M_1\mid X,G]=\mu$. With $d=\beta_g$, this covariance is $0$ by~\eqref{ass: ortho_alpha}.
For $d=d_u$, the historical-price model gives $M_1=\alpha_g(X,G)+a(U,X)$. Hence IV independence gives
\begin{align*}
\EE[Rd_u\mid X,G]
&=\Cov(a(U,X),d_u(U,X)\mid X,G)\\
&=\Cov(a(U,X),d_u(U,X)\mid X)=:C(X).
\end{align*}

Combining the above results, and using that $\beta_{p,g}(X,G)$ is measurable given $(X,G)$, yields
\begin{equation}
\EE[RY\mid X,G]
=\bar b_1\EE[RP\mid X,G]+\bar b_2\EE[RP^2\mid X,G]
 +\beta_{p,g}(X,G)\EE[RP\mid X,G]+C(X).
\label{eq: shared conditional identification}
\end{equation}
This also supplies the starting point for Lemma~\ref{lm: identification eq 2}.

Finally, let $H_0=G-\EE[G\mid X]$. Multiply~\eqref{eq: shared conditional identification} by $H_0$ and take conditional expectation given $X$. Since $\EE[H_0\mid X]=0$, the term $H_0C(X)$ integrates to zero. This then gives
\[
\EE[H_0RY\mid X]
=\bar b_1\EE[H_0RP\mid X]+\bar b_2\EE[H_0RP^2\mid X]
 +\EE[H_0RP\beta_{p,g}(X,G)\mid X].
\]
Substituting the definitions of $H_0,R,\bar b_1,\bar b_2$ is precisely~\eqref{eq: lemma1}.

\QED

\subsubsection{Proof of Lemma~\ref{lm: identification eq 2}.}

Use the notation of the preceding proof and fix a positive integer $k$ satisfying the stated integrability conditions. Define
\[
H_k=G^k\{G-\EE[G\mid X]\},\qquad \bar H_k=\EE[H_k\mid X],\qquad
B=\beta_{p,g}(X,G).
\]
Multiplying~\eqref{eq: shared conditional identification} by $H_k$ and averaging given $X$ gives
\begin{equation}
\EE[H_kRY\mid X]
=\bar b_1\EE[H_kRP\mid X]+\bar b_2\EE[H_kRP^2\mid X]
 +\EE[H_kRPB\mid X]+\bar H_k C(X).
\label{l2: eq1}
\end{equation}

Averaging~\eqref{eq: shared conditional identification} without the weight gives
\begin{equation}
\EE[RY\mid X]
=\bar b_1\EE[RP\mid X]+\bar b_2\EE[RP^2\mid X]
 +\EE[RPB\mid X]+C(X).
\label{l2: eq2}
\end{equation}

Subtract $\bar H_k$ times~\eqref{l2: eq2} from~\eqref{l2: eq1}. The two $\bar H_k C(X)$ terms cancel. For each integrable $V$, the definition
$\Cov(H_k,V\mid X)=\EE[H_kV\mid X]-\bar H_k\EE[V\mid X]$
then yields
\begin{align*}
\EE[H_kRY\mid X]-\bar H_k\EE[RY\mid X]
={}&\bar b_1\Cov(H_k,RP\mid X)
 +\bar b_2\Cov(H_k,RP^2\mid X)\\
&+\Cov(H_k,RPB\mid X).
\end{align*}
Move $\bar H_k\EE[RY\mid X]$ back to the right and substitute $H_k,R,B$. The resulting equation is~\eqref{eq: lemma2}. 

\QED

\subsection{Proofs of the Identification Theorem and Its Corollary}
\label{app: identification theorem proofs}

\subsubsection{Proof of Theorem~\ref{thm: identification}.}

Let $R=P-\EE[P\mid X,G]$, $H_0=G-\EE[G\mid X]$, and $H_k=G^kH_0$ for $k\geq1$. For any required integrable random variable $V$, define
\[
C_0(V)=\EE[H_0V\mid X],\qquad
C_k(V)=\Cov(H_k,V\mid X)\quad(k=1,\ldots,K+2).
\]
These operators depend only on the observed law. Define the $(K+3)$-vector $d(X)$ by $d_k(X)=C_k(RY)$, with rows indexed by $k=0,\ldots,K+2$. Define the $(K+3)\times(K+2)$ matrix $M(X)$ by its $k$th row,
\[
M_{k,\cdot}(X)=\big(C_k(RP),\ C_k(RP^2),\\
 C_k(RP B_1(X,G)),\ldots,C_k(RP B_K(X,G))\big).
\]
For $k\geq1$, the definition of $d_k$ has moved the observable centering term
$\EE[H_k\mid X]\EE[RY\mid X]$ to the left-hand side of~\eqref{eq: lemma2}.

Assumption~\ref{ass: beta_pg basis} gives $\beta_{p,g}=\sum_{j=1}^K\theta_jB_j$. By linearity, Lemma~\ref{lm: identification eq 1} supplies row zero and Lemma~\ref{lm: identification eq 2} supplies all remaining rows of
\[
M(X)v(X)=d(X),\qquad
v(X)=\big(\bar b_1(X),\beta_{p,2}(X),\theta_1,\ldots,\theta_K\big)^\top,
\]
where $\bar b_1(X)=\EE[\beta_{p,1}(U,X)\mid X]$. Thus the system is consistent under the maintained model.

On the probability 1 set where $M(X)$ has full column rank, $M(X)^\top M(X)$ is invertible and
\[
v(X)=\{M(X)^\top M(X)\}^{-1}M(X)^\top d(X).
\]
Indeed, if $Mv=Mv'$, then $M(v-v')=0$, and a full-column-rank matrix has a trivial null space. 

We then recover the coefficient relevant for price-only deployment.
The linear policy coefficient averages over the natural conditional distribution of $G$, so
\[
\beta_{p,1}(X)=\bar b_1(X)+\sum_{j=1}^K\theta_j\EE[B_j(X,G)\mid X].
\]
The right-hand side is now identified, as is $\beta_{p,2}(X)$. Proposition~\ref{prop: optimal policy} then identifies the optimal feasible price. 

\QED

\subsubsection{Proof of Corollary~\ref{cor: identification 2}.}

Under $\beta_{p,g}=0$, $\beta_{p,1}(X)=\EE[\beta_{p,1}(U,X)\mid X]$, so Lemma~\ref{lm: identification eq 1} gives
\[
\Omega_1=\Omega_2\beta_{p,1}+\Omega_3\beta_{p,2}.
\]
Let $H_1=G\{G-\EE[G\mid X]\}$ and $R=P-\EE[P\mid X,G]$. Lemma~\ref{lm: identification eq 2} with $k=1$ gives
\[
\EE[H_1RY\mid X]-\EE[H_1\mid X]\EE[RY\mid X]
=\Upsilon_2\beta_{p,1}+\Upsilon_3\beta_{p,2}.
\]
The left side is $\Cov(H_1,RY\mid X)=\Upsilon_1$. Thus the observable system is
\[
\begin{pmatrix}\Omega_2&\Omega_3\\\Upsilon_2&\Upsilon_3\end{pmatrix}
\binom{\beta_{p,1}}{\beta_{p,2}}=\binom{\Omega_1}{\Upsilon_1}.
\]

Write $D=\Omega_2\Upsilon_3-\Omega_3\Upsilon_2$. On the assumed event $D\ne0$,
\[
A^{-1}=\frac1D\begin{pmatrix}\Upsilon_3&-\Omega_3\\-\Upsilon_2&\Omega_2\end{pmatrix}.
\]
Multiplying by $(\Omega_1,\Upsilon_1)^\top$ gives
\[
\beta_{p,1}=\frac{\Upsilon_3\Omega_1-\Omega_3\Upsilon_1}{D},\qquad
\beta_{p,2}=\frac{\Omega_2\Upsilon_1-\Upsilon_2\Omega_1}{D}.
\]

\QED

\section{Moment Conditions and Identification Stability}
\label{app: moment stability}
\label{sec: Proof of Estimation and Policy Learning}

We first verify the conditional moments and their detection by the population criterion. We then compare residual and derivative errors and establish coefficient recovery under almost-sure and uniform rank. These deterministic results are used in the statistical proofs of Appendix~\ref{app: learning proofs}.

\subsection{Validity of the Conditional Moments}
\label{app: conditional moment proof}

\subsubsection{Proof of Lemma~\ref{lm: conditional moment restriction}.}

We check the eight coordinates in their row-by-row vectorization order. Throughout this proof, the first-stage functions and the second-stage parameters are evaluated at their true values.

Let $R=P-h_{0,2}(X,G)$, where $h_{0,2}(X,G)=\EE[P\mid X,G]$. The first three coordinates satisfy
\begin{align*}
\EE[W_{0,1}\mid X]&=\EE[G\mid X]-h_{0,1}(X)=0,\\
\EE[W_{0,2}\mid X]&=\EE[\EE[P-h_{0,2}(X,G)\mid X,G]\mid X]=0,\\
\EE[W_{0,3}\mid X]&=\EE[G^2\mid X]-h_{0,3}(X)=0.
\end{align*}

By definitions,
\begin{align*}
\EE[W_{0,4}\mid X]&=\EE[RY\mid X]-h_{0,4}(X)=0,\\
\EE[W_{0,5}\mid X]&=\EE[RP\mid X]-h_{0,5}(X)=0,\\
\EE[W_{0,6}\mid X]&=\EE[RP^2\mid X]-h_{0,6}(X)=0.
\end{align*}

Put $H_0=G-h_{0,1}(X)$, $H_1=GH_0$, and
$v_G(X)=h_{0,3}(X)-h_{0,1}(X)^2=\EE[H_1\mid X]$.
Writing $b_j=\beta_{p,j}(X)$ at the truth, the seventh residual is
$w_7=H_0R(Y-b_1P-b_2P^2)$, so
\[
\EE[w_7\mid X]=\Omega_1-\Omega_2b_1-\Omega_3b_2=0
\]
by Lemma~\ref{lm: identification eq 1} with $\beta_{p,g}=0$.
For the eighth coordinate, observe separately that
\begin{align*}
\EE[\rho_4\mid X]-v_G h_{0,4}
&=\EE[H_1RY\mid X]-\EE[H_1\mid X]\EE[RY\mid X]=\Upsilon_1,\\
\EE[\rho_5\mid X]-v_G h_{0,5}
&=\Cov(H_1,RP\mid X)=\Upsilon_2,\\
\EE[\rho_6\mid X]-v_G h_{0,6}
&=\Cov(H_1,RP^2\mid X)=\Upsilon_3.
\end{align*}
Substituting these three identities into $w_8$ yields
$\EE[w_8\mid X]=\Upsilon_1-\Upsilon_2b_1-\Upsilon_3b_2=0$
by Lemma~\ref{lm: identification eq 2} with $k=1$ and $\beta_{p,g}=0$.

\QED

\subsection{Identification through the Population Criterion}
\label{app: population criterion proof}

\subsubsection{Proof of Lemma~\ref{lm: population criterion}.}

Let $m_\alpha(X)=m(X;\alpha)$ and
$\Sigma_{\widetilde\alpha}(X)=\EE[W_0(Z;\widetilde\alpha)W_0(Z;\widetilde\alpha)^\top\mid X]$.
The tower property gives
\[
J(\alpha)=\sup_{f\in\mF}
\{\EE[m_\alpha(X)^\top f(X)]
 -\EE[f(X)^\top\Sigma_{\widetilde\alpha}(X)f(X)]\}.
\]
The matrix $\Sigma_{\widetilde\alpha}$ is positive semidefinite by definition. The assumed uniform bound on $W_0(Z;\widetilde\alpha)$ also gives a finite $C_W$ with
$\Sigma_{\widetilde\alpha}(X)\preceq C_W I$ almost surely.

We first show that the objective is zero exactly at zero residuals.
Residual realizability and star-shapedness imply $0\in\mF$, so $J(\alpha)\geq0$. If $m_\alpha=0$ almost surely, every objective value inside the supremum is nonpositive, and $f=0$ attains zero. Thus $J(\alpha)=0$.
Conversely, suppose $u^2=\EE\|m_\alpha(X)\|_{\ell^2}^2>0$. Choose
$t=\{2(1+C_W)\}^{-1}\in(0,1]$. Because $m_\alpha\in\mF$ and the class is star-shaped about zero, $f=t m_\alpha$ is feasible. Hence
\begin{align*}
J(\alpha)
&\geq t u^2-t^2\EE[m_\alpha^\top\Sigma_{\widetilde\alpha}m_\alpha]\\
&\geq (t-C_Wt^2)u^2\geq \tfrac{t}{2}u^2>0.
\end{align*}
This proves the lower bound with $c_J=t/2$. When $\alpha_0\in\mH$, it attains the global minimum zero by Lemma~\ref{lm: conditional moment restriction}. Otherwise no feasible zero is implied. If $\Sigma_{\widetilde\alpha}\succeq c_\Sigma I$, enlarging the test set to all $L^2$ functions and completing the square gives $J(\alpha)\leq\|m_\alpha\|_2^2/(4c_\Sigma)$.

Next, we identify all five second-stage functions from zero residuals.
For a candidate $\alpha=(b_1,b_2,h_4,h_5,h_6)$ with the first stage fixed at $h_0$, the fourth through sixth conditional means are
\[
m_4=h_{0,4}-h_4,\qquad m_5=h_{0,5}-h_5,\qquad m_6=h_{0,6}-h_6.
\]
Thus $m_\alpha=0$ forces the three nested nuisance functions to equal their true values. With these equalities imposed, the last two moments become
\[
A(X)\binom{b_1-\beta_{p,1}}{b_2-\beta_{p,2}}=\binom00.
\]

\QED

\subsection{Comparison of Residual Error and Derivative Error}
\label{app: derivative link}
This subsection separates the nonlinear comparison from strength of identification. All moment functions below keep the first-stage functions fixed at their true values $h_0$.

\subsubsection{Proof of Lemma~\ref{lem: residual derivative link}.}

We first write $\alpha=(b_1,b_2,h_4,h_5,h_6)$,
$\alpha_0=(b_{01},b_{02},h_{0,4},h_{0,5},h_{0,6})$, and
$u=(u_{b1},u_{b2},u_4,u_5,u_6)=\alpha-\alpha_0$.
Set $v_G=h_{0,3}-h_{0,1}^2=\Var(G\mid X)$. Conditional expectations of the first six residuals give
\[
m_1=m_2=m_3=0,\qquad m_4=-u_4,\quad m_5=-u_5,\quad m_6=-u_6.
\]
The identifying equations at the truth imply
$\Omega_1=\Omega_2b_{01}+\Omega_3b_{02}$ and
$\Upsilon_1=\Upsilon_2b_{01}+\Upsilon_3b_{02}$. Thus
\begin{align*}
m_7&=-\Omega_2u_{b1}-\Omega_3u_{b2},\nonumber\\
m_8&=-\Upsilon_2u_{b1}-\Upsilon_3u_{b2}
        +v_G(-u_4+b_1u_5+b_2u_6).
\end{align*}
For example, the second equality follows on substituting
$\EE[\rho_4\mid X]=\Upsilon_1+v_Gh_{0,4}$,
$\EE[\rho_5\mid X]=\Upsilon_2+v_Gh_{0,5}$, and
$\EE[\rho_6\mid X]=\Upsilon_3+v_Gh_{0,6}$ into $w_8$.

Set $\alpha$ to $\alpha_0+tu$ and differentiate at $t=0$. The resulting matrix is
\begin{equation}
(T_0u)(x)=J_0(x)u(x),\qquad
J_0=\begin{pmatrix}
0&0&0&0&0\\
0&0&0&0&0\\
0&0&0&0&0\\
0&0&-1&0&0\\
0&0&0&-1&0\\
0&0&0&0&-1\\
-\Omega_2&-\Omega_3&0&0&0\\
-\Upsilon_2&-\Upsilon_3&-v_G&v_Gb_{01}&v_Gb_{02}
\end{pmatrix}.
\label{eq: explicit derivative matrix}
\end{equation}
All entries are bounded, hence $T_0$ is a bounded linear map on $L^2(P_X)^5$, and in particular on the parameter RKHS with bounded kernel used below. Along the path
\[
m(\cdot;\alpha_0+tu)=tT_0u+t^2v_G(u_{b1}u_5+u_{b2}u_6)e_8,
\]
where $e_8$ denotes the eighth coordinate vector of $\mathbb R^8$. Dividing by $t$ and letting $t\to0$ justifies the derivative in $L^2$. In particular,
\begin{equation*}
m(\cdot;\alpha)=T_0u+v_G(u_{b1}u_5+u_{b2}u_6)e_8.
\end{equation*}

Since $u_5=-m_5$, $u_6=-m_6$, and the same two coordinates of $T_0u$ equal $-u_5,-u_6$, we obtain
\begin{align}
T_0u&=m+v_G(u_{b1}m_5+u_{b2}m_6)e_8,\label{eq: triangular derivative forward}\\
m&=T_0u-v_G\{u_{b1}(T_0u)_5+u_{b2}(T_0u)_6\}e_8.
\label{eq: triangular derivative inverse}
\end{align}
Choose fixed $B_M,V<\infty$ such that
\[
\sup_{\alpha\in\mH_M}\operatorname*{ess\,sup}_x
\{u_{b1}(x)^2+u_{b2}(x)^2\}^{1/2}\leq B_M,
\qquad\|v_G\|_\infty\leq V.
\]
Cauchy--Schwarz applied to the last terms of~\eqref{eq: triangular derivative forward}--\eqref{eq: triangular derivative inverse} gives, with $K_M:=1+VB_M$, the pointwise bounds
\[
\|T_0u(x)\|_{\ell^2}\leq K_M\|m(x;\alpha)\|_{\ell^2},\qquad
\|m(x;\alpha)\|_{\ell^2}\leq K_M\|T_0u(x)\|_{\ell^2}.
\]
Integrating the squared inequalities gives~\eqref{eq: unweighted derivative bridge}.

\QED

\subsection{Coefficient Recovery under Almost-Sure and Uniform Rank}
\label{app: coefficient stability}

This section we write $s_A(X)=\sigma_{\min}\{A(X)\}$. The first lemma gives qualitative stability on a fixed bounded ball under almost-sure rank. The second gives a linear error bound under the stronger uniform-rank condition.

Throughout this subsection, fix $M<\infty$ and consider the
admissible parameter sets
\[
\mH_{n,M}
=\{\alpha\in\mH_n:\|\alpha\|_{\mH_n}^2\le M\}.
\]
Assume that
$\sup_n\sup_{\alpha\in\mH_{n,M}}\|\alpha\|_\infty
\le D_M<\infty$, where
$\|\alpha\|_\infty
=\operatorname*{ess\,sup}_x\|\alpha(x)\|_{\ell^2}$.
The structural target $\alpha_0$ is bounded but need not
belong to $\mH_n$.
Write $s_A(X)=\sigma_{\min}\{A(X)\}$ and let
$B_M=D_M+\|\alpha_0\|_\infty$.

\begin{lemma}[Qualitative stability under almost-sure rank]\label{lem: qualitative derivative stability}
Under the preceding boundedness conditions, suppose $A(X)$
is nonsingular almost surely. There exists $C_M<\infty$,
independent of $n$, $\alpha$, and $\tau$, such that, for every
$n$, $\alpha\in\mH_{n,M}$, and $0<\tau\le1$,
\begin{equation}
\|\alpha-\alpha_0\|_2^2\leq C_M\tau^{-2}\|T_0(\alpha-\alpha_0)\|_2^2
+B_M^2\Pr\{s_A(X)\leq\tau\},
\label{eq: qualitative derivative split}
\end{equation}
Consequently, $\sup_n\omega_{M,n}^D(t)\to0$ as $t\downarrow0$.
\end{lemma}

\noindent\textbf{Proof of Lemma~\ref{lem: qualitative derivative stability}.}

Set $d=T_0u$. The derivative matrix~\eqref{eq: explicit derivative matrix} gives $u_j=-d_j$ for $j=4,5,6$ and
\begin{equation}
A(X)\binom{u_{b1}}{u_{b2}}
=-\binom{d_7}{d_8-v_G(d_4-b_{01}d_5-b_{02}d_6)}.
\label{eq: pointwise derivative inverse}
\end{equation}
The right side is bounded pointwise by $C\|d(X)\|_{\ell^2}$ because true coefficients and $v_G$ are bounded. On $\{s_A>\tau\}$, inversion of $A$ therefore gives
$\|u(X)\|_{\ell^2}^2\leq C_M\tau^{-2}\|d(X)\|_{\ell^2}^2$,
where the nuisance coordinates are included using $\tau\leq1$.

On $\{s_A\leq\tau\}$, use the fixed-ball bound $\|u(X)\|_{\ell^2}\leq B_M$. Adding and integrating the two bounds gives~\eqref{eq: qualitative derivative split}. 

Since $s_A>0$ almost surely, continuity of probability from above gives
$\Pr(s_A\leq\tau)\downarrow0$ as $\tau\downarrow0$. For any target $\epsilon>0$, choose a fixed $\tau>0$ so small that $B_M^2\Pr(s_A\leq\tau)<\epsilon^2/2$. With that $\tau$ held fixed, choose the derivative-norm radius $t$ so small that $C_M\tau^{-2}t^2<\epsilon^2/2$. Every feasible candidate then has $\|u\|_2<\epsilon$. Taking the supremum proves $\omega_{M,n}^D(t)\to0$ uniformly.

\QED

\begin{lemma}[Uniform-rank derivative link]\label{lem: price moment stability}
Under the preceding boundedness conditions, suppose the
uniform-rank condition~\eqref{eq: uniform price rank} holds.
There exists $C<\infty$, independent of $n$ and $\alpha$,
such that, for every $n$ and $\alpha\in\mH_{n,M}$,
\begin{equation}
\|\alpha-\alpha_0\|_2^2\leq(1+C\kappa_A^{-2})\|T_0(\alpha-\alpha_0)\|_2^2.
\label{eq: uniform derivative stability}
\end{equation}
It follows that $\omega_{M,n}^D(t)\leq C't$, where $C'=\sqrt{1+C\kappa_A^{-2}}$.
Lemma~\ref{lem: residual derivative link} also gives $\|\alpha-\alpha_0\|_2\leq C_M\|m(\cdot;\alpha)\|_2$.
\end{lemma}

\noindent\textbf{Proof of Lemma~\ref{lem: price moment stability}.}

Put $u=\alpha-\alpha_0$, $d=T_0u$, and $v_G=\Var(G\mid X)$. From~\eqref{eq: explicit derivative matrix}, $u_4=-d_4$, $u_5=-d_5$, and $u_6=-d_6$. The last two derivative coordinates give
\[
A(X)\binom{u_{b1}}{u_{b2}}
=-\binom{d_7}{d_8-v_G(d_4-b_{01}d_5-b_{02}d_6)},
\]
exactly the same as~\eqref{eq: pointwise derivative inverse}.

The true coefficients and $v_G$ are bounded. By AM-GM and Cauchy-Schwarz inequalities, the squared norm of the right side of \eqref{eq: pointwise derivative inverse} is at most
\[
d_7^2+2d_8^2+2v_G^2(1+b_{01}^2+b_{02}^2)(d_4^2+d_5^2+d_6^2)
\leq C\|d(X)\|_{\ell^2}^2.
\]
Uniform rank implies $\|A(X)^{-1}\|_{\rm op}\leq\kappa_A^{-1}$ almost surely. Hence
$\|(u_{b1},u_{b2})(X)\|_{\ell^2}^2\leq C\kappa_A^{-2}\|d(X)\|_{\ell^2}^2$.
Adding $u_4^2+u_5^2+u_6^2=d_4^2+d_5^2+d_6^2$ and integrating proves~\eqref{eq: uniform derivative stability}.

For every candidate with $\|T_0u\|_2\leq t$, we have
$\|u\|_2\leq\sqrt{1+C\kappa_A^{-2}}\,t$.
Taking the supremum over candidates proves $\omega_{M,n}^D(t)\leq C't$.
For the residual bound, substitute
$\|T_0u\|_2\leq K_M\|m(\cdot;\alpha)\|_2$ from Lemma~\ref{lem: residual derivative link} into~\eqref{eq: uniform derivative stability}.

\QED

\section{Proofs of Convergence and Regret Bounds}
\label{app: learning proofs}

The proofs proceed from localized empirical bounds to first-stage error control, the feasible criterion, and the main convergence theorem. The final proof converts coefficient error into revenue regret. We use the moment and stability results in Appendix~\ref{app: moment stability}. 

\subsection{Localized Empirical-Process Bounds}
\label{app: empirical process bounds}

\begin{lemma}[Localized empirical means and second moments]\label{lem: localized empirical bounds}
Let $O_1,\ldots,O_n$ be i.i.d. with law $\mathbb P$ on a measurable observation space $\mathcal O$. Let $\mathcal C$ be a nonempty measurable scalar function class satisfying
\[
\mathcal C\subseteq\{g:\mathcal O\to[-1,1]\mid g\text{ is measurable}\},
\qquad
\mathcal C=\operatorname{star}(\mathcal C)
=\{tg:g\in\mathcal C,\ 0\leq t\leq1\}.
\]
Write $\mathbb P_ng=n^{-1}\sum_{i=1}^ng(O_i)$ and $\|g\|_2^2=\mathbb Pg^2$. Suppose $\rho_n>0$ decreases to zero and satisfies $\mathcal R_n(\rho_n,\mathcal C)\leq\rho_n^2/64$. Put $d_n=\rho_n+\sqrt{\log n/n}$. For every fixed $\epsilon>0$, with probability tending to one, simultaneously for all $g\in\mathcal C$,
\begin{align}
|(\mathbb P_n-\mathbb P)g|
&\leq C\{d_n\|g\|_2+d_n^2\},\label{eq: localized mean}\\
|\mathbb P_ng^2-\mathbb Pg^2|
&\leq\epsilon\|g\|_2^2+C_\epsilon d_n^2.\label{eq: localized squares}
\end{align}
\end{lemma}

\subsubsection{Proof of Lemma~\ref{lem: localized empirical bounds}.}

\noindent\textbf{Proof of \eqref{eq: localized mean}}.
For $r\geq\rho_n$, multiplying each $g$ with $\|g\|_2\leq r$ by $\rho_n/r$ keeps it in the star-shaped class and puts it in the ball of radius $\rho_n$. Therefore
\[
\mathcal R_n(r,\mathcal C)
\leq\frac r{\rho_n}\mathcal R_n(\rho_n,\mathcal C)
\leq\frac{\rho_n r}{64}.
\]
For smaller radii the complexity is at most $\mathcal R_n(\rho_n,\mathcal C)$ by inclusion. 

For $\mathcal C_r=\{g\in\mathcal C:\|g\|_2\leq r\}$, the bounded empirical-process inequality of \citet[Theorem~2.1]{bartlett2005local}, applied to both signs, gives with probability at least $1-2e^{-u}$
\[
\sup_{g\in\mathcal C_r}|(\mathbb P_n-\mathbb P)g|
\leq C\left\{\mathcal R_n(r,\mathcal C)+r\sqrt{u/n}+u/n\right\}.
\]
Its conditions hold because $|g|\leq1$ and $\Var(g(O))\leq\mathbb Pg^2\leq r^2$.
To derive the symmetrization bound, write $O=(O_1,\ldots,O_n)$, let $O'=(O_1',\ldots,O_n')$ be an independent copy, and let $\epsilon_1,\ldots,\epsilon_n$ be independent Rademacher signs, independent of both samples. Since $\mathcal C_r$ is defined using the population norm, it is deterministic, and
\begin{align*}
\EE_O\sup_{g\in\mathcal C_r}|(\mathbb P_n-\mathbb P)g|
&\leq\EE_{O,O'}\sup_{g\in\mathcal C_r}
\left|\frac1n\sum_{i=1}^n\{g(O_i)-g(O_i')\}\right|\\
&=\EE_{O,O',\epsilon}\sup_{g\in\mathcal C_r}
\left|\frac1n\sum_{i=1}^n\epsilon_i\{g(O_i)-g(O_i')\}\right|\\
&\leq2\EE_{O,\epsilon}\sup_{g\in\mathcal C_r}
\left|\frac1n\sum_{i=1}^n\epsilon_i g(O_i)\right|
=2\mathcal R_n(r,\mathcal C).
\end{align*}
The first inequality is Jensen's inequality conditional on $O$, the equality follows by exchanging each i.i.d. pair $(O_i,O_i')$ according to $\epsilon_i$, and the last inequality follows from the triangle inequality and the identical laws of the two samples.

We then obtain simultaneous bounds by dyadic localization.
For all sufficiently large $n$, $d_n<1$. Take $u=4\log n$ and define
\[
J_n=\left\lceil\log_2(1/d_n)\right\rceil,
\qquad r_j=2^jd_n\quad(0\leq j\leq J_n).
\]
Then $r_{J_n}\geq1$ and, since $d_n\geq\sqrt{\log n/n}$,
\[
J_n+1\leq2+\log_2(1/d_n)
\leq2+\tfrac12\log_2 n\leq C(1+\log n).
\]
For every $j$, $r_j\geq d_n\geq\rho_n$, so 
\begin{align*}
C\left\{\mathcal R_n(r_j,\mathcal C)+r_j\sqrt{u/n}+u/n\right\}
&\leq C\left\{\frac{\rho_n r_j}{64}
+2r_j\sqrt{\frac{\log n}{n}}+\frac{4\log n}{n}\right\}\\
&\leq C'\{d_nr_j+d_n^2\}.
\end{align*}
Consequently, the events
\[
\mathcal E_j=\left\{
\sup_{g\in\mathcal C_{r_j}}|(\mathbb P_n-\mathbb P)g|
\leq C'\{d_nr_j+d_n^2\}\right\}
\]
satisfy $\Pr(\mathcal E_j^c)\leq2e^{-u}=2n^{-4}$. The union bound gives
\[
\Pr\left\{\left(\bigcap_{j=0}^{J_n}\mathcal E_j\right)^c\right\}
\leq\sum_{j=0}^{J_n}\Pr(\mathcal E_j^c)
\leq2(J_n+1)n^{-4}
\leq C(1+\log n)n^{-4}\longrightarrow0.
\]
We just consider $\bigcap_{j=0}^{J_n}\mathcal E_j$.
If $\|g\|_2\leq d_n$, use $j=0$, obtaining $Cd_n^2$. Otherwise choose the smallest $j$ with $\|g\|_2\leq r_j$, then $r_j<2\|g\|_2$. This proves~\eqref{eq: localized mean} for every $g$ on the same event.

\noindent\textbf{Proof of \eqref{eq: localized squares}}.
For the squared class $\{g^2:g\in\mathcal C_r\}$, $0\leq g^2\leq1$ and
$\Var(g^2(O))\leq\mathbb Pg^4\leq r^2$.
The map $v\mapsto v^2$ is $2$-Lipschitz on $[-1,1]$ and vanishes at zero. Rademacher contraction therefore bounds the expected complexity of this squared class by a fixed multiple of $\mathcal R_n(r,\mathcal C)$. Applying previous steps to that class, and repeating the same finite union over $r_j$, gives
\[
|\mathbb P_ng^2-\mathbb Pg^2|
\leq C\{d_n\|g\|_2+d_n^2\}
\quad\text{for all }g\in\mathcal C.
\]
Finally use $C d_n\|g\|_2\leq\epsilon\|g\|_2^2+C^2d_n^2/(4\epsilon)$ to obtain~\eqref{eq: localized squares}. Both assertions hold on the intersection of the two events, whose probability tends to one. 

\QED

\subsection{First-Stage Error and Bounds for the Feasible Criterion}
\label{app: feasible criterion bounds}

\begin{lemma}[Uniform feasible-criterion bounds]\label{lem: first stage stability}
For i.i.d. observations $Z_1,\ldots,Z_n$, assume the boundedness and measurability conditions of Assumption~\ref{ass: identification and spaces}(a), the star-shapedness and norm homogeneity conditions of Assumption~\ref{ass: identification and spaces}(b), the fixed-reference condition of Assumption~\ref{ass: identification and spaces}(c), and the complexity and cross-fitting conditions of Assumption~\ref{ass: parameters}(a)--(b). Let $B>0$ be as in Assumption~\ref{ass: parameters}(a), put $a_n=b_n+r_n+\sqrt{\log n/n}$, and choose $\lambda_n\asymp a_n^2$ with a sufficiently large constant, as in Assumption~\ref{ass: parameters}(c). Define $s_f=\max\{1,\|f\|_{\mF}/\sqrt B\}$ and
$Q_n(f)=n^{-1}\sum_i\{f(X_i)^\top W(Z_i;\widetilde\alpha,\widehat h^{(-k(i))})\}^2$.
For any fixed parameter radius $M$, the following inequalities hold simultaneously with probability tending to one, uniformly for $\alpha\in\mH_M$ and $f\in\mF$:
\begin{align}
|\widehat\Psi_n(\alpha,f)-\EE[m(X;\alpha)^\top f(X)]|
&\leq C_M\{a_n\|f\|_2+a_n^2 s_f\},\label{eq: feasible linear}\\
\|f\|_{n,2}
&\leq C\{\|f\|_2+a_n s_f\},\label{eq: f empirical}\\
Q_n(f)+\lambda_n\|f\|_{\mF}^2
&\geq c_0\|f\|_2^2+\tfrac12\lambda_n\|f\|_{\mF}^2-Ca_n^2,\label{eq: feasible lower}\\
Q_n(f)
&\leq C\{\|f\|_2^2+a_n^2s_f^2\}.\label{eq: feasible upper}
\end{align}
Here $c_0>0$ and constants in the last three displays depend on the fixed reference radius, not the candidate radius $M$.
\end{lemma}

\subsubsection{Proof of Lemma~\ref{lem: first stage stability}.}

In the following, set $d_n=b_n+\sqrt{\log n/n}\leq a_n$. All constants below are independent of $n$. Dependence on a fixed candidate radius is denoted by a subscript $M$. 

\textbf{Step 1: Proof of \eqref{eq: f empirical}.}
For arbitrary $f\in\mF$, put
$s_f=\max\{1,\|f\|_{\mF}/\sqrt B\}$ and $\bar f=f/s_f$.
Star-shapedness and absolute homogeneity ensure $\bar f\in\mF_B$.
Apply Lemma~\ref{lem: localized empirical bounds} to the star hull of $\mathcal G_{M,B}$ and to each coordinate class of $\mF_B$. Their fixed envelopes and critical radii are controlled by Assumptions~\ref{ass: identification and spaces}(a) and~\ref{ass: parameters}(a). The factor $C_M$ in a critical radius only changes the constants multiplying $d_n$.
For $g_{\alpha,\bar f}(Z)=W_0(Z;\alpha)^\top\bar f(X)$, bounded residuals give
$\|g_{\alpha,\bar f}\|_{L^2(P_Z)}\leq C_M\|\bar f\|_2$.
Multiplying the localized mean inequality by $s_f$ gives
\begin{equation}
|\Psi_n(\alpha,f)-\EE[m(X;\alpha)^\top f(X)]|
\leq C_M\{d_n\|f\|_2+d_n^2s_f\}.
\label{eq: oracle linear detailed}
\end{equation}
Applying the square inequality separately to the eight coordinates and adding gives
\[
\|f\|_{n,2}^2\leq C\|f\|_2^2+Cd_n^2s_f^2,
\qquad
\|f\|_{n,2}\leq C\{\|f\|_2+d_ns_f\}.
\]
This proves~\eqref{eq: f empirical}, since $d_n\leq a_n$.

\textbf{Step 2: Compare the oracle quadratic penalty to its population version.}
For a deterministic reference $\gamma\in\mH_{M_{\rm ref}}$, define
\[
Q_{n,0}(f;\gamma)=\frac1n\sum_i\{f(X_i)^\top W_0(Z_i;\gamma)\}^2,
\qquad Q_0(f;\gamma)=\EE_X[f(X)^\top\Sigma_\gamma(X)f(X)].
\]
The expectation integrates over the population law of $X$ with $\gamma$ and $f$ held fixed. At $\gamma=\widetilde\alpha$, abbreviate these quantities by $Q_{n,0}(f)$ and $Q_0(f)$.
Use Lemma~\ref{lem: localized empirical bounds} for $\mathcal G_{M_{\rm ref},B}$, then multiply its square bound by $s_f^2$. Uniformly over every reference in $\mH_{M_{\rm ref}}$,
\begin{equation}
|Q_{n,0}(f)-Q_0(f)|\leq\epsilon Q_0(f)+C_\epsilon d_n^2s_f^2.
\label{eq: oracle quadratic detailed}
\end{equation}
The uniform oracle event is
\[
\mathcal E_n^Q(\epsilon)=\left\{
\begin{aligned}
|Q_{n,0}(f;\gamma)-Q_0(f;\gamma)|
&\leq\epsilon Q_0(f;\gamma)+C_\epsilon d_n^2s_f^2\\
&\text{for all }\gamma\in\mH_{M_{\rm ref}},\ f\in\mF
\end{aligned}\right\},
\]
and $\Pr\{\mathcal E_n^Q(\epsilon)\}\to1$. Assumption~\ref{ass: identification and spaces}(c) gives the reference event
\[
\mathcal E_n^{\rm ref}=\left\{
\widetilde\alpha\in\mH_{M_{\rm ref}},\quad
\operatorname*{ess\,inf}_x\lambda_{\min}
\{\Sigma_{\widetilde\alpha}(x)\}\geq c_\Sigma
\right\},\qquad
\Pr(\mathcal E_n^{\rm ref})\to1.
\]
On $\mathcal E_n^Q(\epsilon)\cap\mathcal E_n^{\rm ref}$, substitute $\gamma=\widetilde\alpha$ to obtain~\eqref{eq: oracle quadratic detailed}. Uniformity over the reference ball permits this substitution even when $\widetilde\alpha$ depends on the data.

\textbf{Step 3: Transfer first-stage population error to held-out empirical error.}
Let $K_{\rm cf}$ be the fixed number of folds, $I_k$ a held-out fold, $n_k=|I_k|$, and $\mathcal T_k$ its training information, including the independently chosen fold partition. Conditional on $\mathcal T_k$, the fit $\widehat h^{(-k)}$ is fixed and the observations in $I_k$ are i.i.d. with the original law. Write
\[
D_i=\|\widehat h^{(-k(i))}(X_i,G_i)-h_0(X_i,G_i)\|_{\ell^2},
\qquad \Delta_n^2=\frac1n\sum_iD_i^2.
\]
Here $h_1,h_3$ are evaluated at $X_i$. The notation $(X_i,G_i)$ accommodates $h_2$.
Define the training event
\[
\mathcal A_{n,k}=\left\{
\max_{j=1,2,3}\|\widehat h_j^{(-k)}-h_{0,j}\|_2\leq C_h r_n
\right\}\in\mathcal T_k.
\]
Assumption~\ref{ass: parameters}(b) gives $\Pr(\bigcap_{k=1}^{K_{\rm cf}}\mathcal A_{n,k})\to1$. On $\mathcal A_{n,k}$,
$\EE[D_i^2\mid\mathcal T_k]\leq3C_h^2r_n^2$ for $i\in I_k$.
Clipping supplies a deterministic $D_{\max}$ with $0\leq D_i^2\leq D_{\max}^2$, and hence
\[
\Var(D_i^2\mid\mathcal T_k)
\leq\EE[D_i^4\mid\mathcal T_k]
\leq D_{\max}^2\EE[D_i^2\mid\mathcal T_k]
\leq C r_n^2.
\]
On $\mathcal A_{n,k}$, conditional Bernstein concentration at $u=3\log n$ gives, except on a conditional event of probability at most $n^{-3}$,
\[
\frac1{n_k}\sum_{i\in I_k}D_i^2
\leq C\left\{r_n^2+r_n\sqrt{\frac{\log n}{n_k}}
                   +\frac{\log n}{n_k}\right\}
\leq C'\{r_n^2+\log n/n\}.
\]
The last step uses $n_k\asymp n$ and $2uv\leq u^2+v^2$.
Define the held-out evaluation event
\[
\mathcal B_{n,k}=\left\{
\frac1{n_k}\sum_{i\in I_k}D_i^2
\leq C_D\{r_n^2+\log n/n\}\right\},
\]
where $C_D$ is a fixed constant large enough for the preceding bound. Conditional Bernstein and the tower property give
\begin{align*}
\ind_{\mathcal A_{n,k}}\Pr(\mathcal B_{n,k}^c\mid\mathcal T_k)
&\leq n^{-3}\ind_{\mathcal A_{n,k}},\\
\Pr(\mathcal A_{n,k}\cap\mathcal B_{n,k}^c)
&=\EE[\ind_{\mathcal A_{n,k}}\Pr(\mathcal B_{n,k}^c\mid\mathcal T_k)]
\leq n^{-3},\\
\Pr\left\{\left(\bigcap_{k=1}^{K_{\rm cf}}\mathcal B_{n,k}\right)^c\right\}
&\leq\sum_{k=1}^{K_{\rm cf}}\Pr(\mathcal A_{n,k}^c)
+K_{\rm cf}n^{-3}=o(1).
\end{align*}
On $\bigcap_{k=1}^{K_{\rm cf}}\mathcal B_{n,k}$,
\[
\Delta_n^2=\sum_{k=1}^{K_{\rm cf}}\frac{n_k}{n}
\left(\frac1{n_k}\sum_{i\in I_k}D_i^2\right)
\leq C_D\{r_n^2+\log n/n\}.
\]
Taking square roots yields
\begin{equation}
\Delta_n\leq C\{r_n+\sqrt{\log n/n}\}\leq Ca_n
\label{eq: crossfit empirical detailed}
\end{equation}
with probability tending to one. 

\textbf{Step 4: Verify the residual's nuisance sensitivity on bounded sets.}
The basic factors are $R_h=P-h_2$, $T_h=G-h_1$, and $v_h=h_3-h_1^2$. If $\delta_j=h_j-h_{0,j}$, their differences satisfy
\[
|R_h-R_{h_0}|=|\delta_2|,\qquad
|T_h-T_{h_0}|=|\delta_1|,\qquad
|v_h-v_{h_0}|\leq|\delta_3|+C|\delta_1|.
\]
The constant in the last bound uses clipping of $h_1$. To expand the products explicitly, write $\alpha=(b_1,b_2,h_4,h_5,h_6)$ and set
\[
S_\alpha=Y-b_1P-b_2P^2,\qquad
H_\alpha=h_4-b_1h_5-b_2h_6.
\]
The last two residuals are $w_7=T_hR_hS_\alpha$ and $w_8=GT_hR_hS_\alpha-v_hH_\alpha$. The product-difference identity gives
\begin{align*}
T_hR_h-T_{h_0}R_{h_0}
&=(T_h-T_{h_0})R_h+T_{h_0}(R_h-R_{h_0})
=-\delta_1R_h-T_{h_0}\delta_2,\\
v_h-v_{h_0}
&=\delta_3-(h_1^2-h_{0,1}^2)
=\delta_3-(h_1+h_{0,1})\delta_1.
\end{align*}
Thus, for $\Delta W=W(Z;\alpha,h)-W_0(Z;\alpha)$, its eight coordinates satisfy
\begin{align*}
(\Delta W_1,\ldots,\Delta W_6)
&=(-\delta_1,-\delta_2,-\delta_3,-Y\delta_2,-P\delta_2,-P^2\delta_2),\\
\Delta W_7&=-(\delta_1R_h+T_{h_0}\delta_2)S_\alpha,\\
\Delta W_8&=-G(\delta_1R_h+T_{h_0}\delta_2)S_\alpha
-\{\delta_3-(h_1+h_{0,1})\delta_1\}H_\alpha.
\end{align*}
Bounded data, clipped nuisances, and the fixed parameter ball imply $|S_\alpha|+|H_\alpha|\leq C_M$ and $|R_h|+|T_{h_0}|+|h_1+h_{0,1}|\leq C$. Consequently,
\begin{align*}
\sum_{j=1}^6(\Delta W_j)^2
&=\delta_1^2+\delta_3^2+(1+Y^2+P^2+P^4)\delta_2^2
\leq C\|\delta\|_{\ell^2}^2,\\
|\Delta W_7|&\leq C_M(|\delta_1|+|\delta_2|),\qquad
|\Delta W_8|\leq C_M(|\delta_1|+|\delta_2|+|\delta_3|).
\end{align*}
Squaring and summing, using $(|\delta_1|+|\delta_2|+|\delta_3|)^2\leq3\|\delta\|_{\ell^2}^2$, gives $\|\Delta W\|_{\ell^2}\leq L_M\|\delta\|_{\ell^2}$. Evaluating at the fitted nuisances yields
\begin{equation}
\|W(Z_i;\alpha,\widehat h^{(-k(i))})-W_0(Z_i;\alpha)\|_{\ell^2}
\leq L_M D_i
\label{eq: residual lipschitz detailed}
\end{equation}
uniformly over $\alpha\in\mH_M$.

\textbf{Step 5: Bound the feasible linear criterion.}
Write $\Delta W_i(\alpha)=W(Z_i;\alpha,\widehat h^{(-k(i))})-W_0(Z_i;\alpha)$. Cauchy--Schwarz, the residual bound~\eqref{eq: residual lipschitz detailed}, the foldwise bound~\eqref{eq: crossfit empirical detailed}, and the empirical norm bound from Step 1 give, respectively,
\begin{align*}
|\widehat\Psi_n(\alpha,f)-\Psi_n(\alpha,f)|
&=\left|\frac1n\sum_i f(X_i)^\top\Delta W_i(\alpha)\right|\\
&\leq\left\{\frac1n\sum_i\|\Delta W_i(\alpha)\|_{\ell^2}^2\right\}^{1/2}\|f\|_{n,2}\\
&\leq L_M\Delta_n\|f\|_{n,2}
\leq C_Ma_n\|f\|_{n,2}\\
&\leq C_Ma_n\{\|f\|_2+d_ns_f\}\\
&\leq C_M\{a_n\|f\|_2+a_n^2s_f\}.
\end{align*}
The last inequality uses $d_n\leq a_n$.
Adding~\eqref{eq: oracle linear detailed} proves~\eqref{eq: feasible linear}. 

\textbf{Step 6: Bound the change in the quadratic reference penalty.}
First take $f\in\mF_B$ and a reference in $\mH_{M_{\rm ref}}$. Put
$u_i=f(X_i)^\top W(Z_i;\widetilde\alpha,\widehat h^{(-k(i))})$
and $v_i=f(X_i)^\top W_0(Z_i;\widetilde\alpha)$.
Step 4 gives $|u_i-v_i|\leq C D_i\|f(X_i)\|_{\ell^2}$.
Bounded residuals give $|u_i+v_i|\leq C\|f(X_i)\|_{\ell^2}$. Since the adversary ball is uniformly bounded, one of these two factors of $\|f(X_i)\|$ can be bounded by a fixed constant. Therefore
\begin{align*}
|Q_n(f)-Q_{n,0}(f)|
&\leq\frac1n\sum_i|u_i-v_i|\,|u_i+v_i|\\
&\leq\frac Cn\sum_iD_i\|f(X_i)\|_{\ell^2}
\leq C\Delta_n\|f\|_{n,2}\\
&\leq Ca_n\|f\|_2+Ca_n^2\\
&\leq\epsilon c_\Sigma\|f\|_2^2
+\left(C+\frac{C^2}{4\epsilon c_\Sigma}\right)a_n^2
=\epsilon c_\Sigma\|f\|_2^2+C_\epsilon a_n^2,
\end{align*}
where the last line is Young's inequality.
Applying this bound to $f/s_f$ and multiplying by $s_f^2$ gives, for arbitrary $f$,
\begin{equation}
|Q_n(f)-Q_{n,0}(f)|
\leq\epsilon c_\Sigma\|f\|_2^2+C_\epsilon a_n^2s_f^2.
\label{eq: feasible quadratic perturbation}
\end{equation}

\textbf{Step 7: Absorb the quadratic remainder using regularization.}
On $\mathcal E_n^{\rm ref}$, $Q_0(f)\geq c_\Sigma\|f\|_2^2$.
Use $\epsilon=1/4$ in~\eqref{eq: oracle quadratic detailed} and~\eqref{eq: feasible quadratic perturbation}, and note $d_n\leq a_n$. Their combination gives
\[
Q_n(f)\geq\tfrac12c_\Sigma\|f\|_2^2-Ca_n^2s_f^2
\geq\tfrac12c_\Sigma\|f\|_2^2-Ca_n^2-\frac{Ca_n^2}{B}\|f\|_{\mF}^2.
\]
Choose the fixed tuning constant so that $\lambda_n\geq2Ca_n^2/B$. Adding $\lambda_n\|f\|_{\mF}^2$ proves~\eqref{eq: feasible lower}, with $c_0=c_\Sigma/2$. 

\textbf{Step 8: Obtain the feasible upper bound and collect events.}
The feasible reference residual at observation $i$ is
\[
\widehat W_i^{\rm ref}=W(Z_i;\widetilde\alpha,\widehat h^{(-k(i))}),
\qquad
\widetilde\alpha=(\widetilde b_1,\widetilde b_2,
\widetilde h_4,\widetilde h_5,\widetilde h_6).
\]
On $\mathcal E_n^{\rm ref}$, boundedness of the fixed ball gives $\|\widetilde\alpha\|_\infty\leq A_{\rm ref}<\infty$. Clipping gives a deterministic $H_\infty<\infty$ such that
\[
h_0,\widehat h^{(-k)}\in\mathcal K_h
:=\{h=(h_1,h_2,h_3):\max_{j=1,2,3}\|h_j\|_\infty\leq H_\infty\}
\quad\text{for every }k.
\]
Here the supremum norm for $h_2$ is over $(X,G)$. Bounded data and the polynomial residual formulas in Step 4 therefore imply
\[
C_{\rm ref}:=\sup_{\gamma\in\mH_{M_{\rm ref}},\,h\in\mathcal K_h}
\operatorname*{ess\,sup}_z\|W(z;\gamma,h)\|_{\ell^2}<\infty,
\qquad \max_{1\leq i\leq n}\|\widehat W_i^{\rm ref}\|_{\ell^2}\leq C_{\rm ref}.
\]
Thus, on the empirical norm event from Step 1,
\[
Q_n(f)=\frac1n\sum_i\{f(X_i)^\top\widehat W_i^{\rm ref}\}^2
\leq C_{\rm ref}^2\|f\|_{n,2}^2
\leq C\{\|f\|_2^2+a_n^2s_f^2\},
\]
which proves~\eqref{eq: feasible upper}. To state the common event explicitly, denote the two uniform events from Step 1 by
\begin{align*}
\mathcal E_{n,M}^{\rm lin}
&=\left\{\sup_{\alpha\in\mH_M,\,f\in\mF}
\frac{|\Psi_n(\alpha,f)-\EE[m(X;\alpha)^\top f(X)]|}
{d_n\|f\|_2+d_n^2s_f}\leq C_M\right\},\\
\mathcal E_n^F
&=\left\{\sup_{f\in\mF}
\frac{\|f\|_{n,2}}{\|f\|_2+d_ns_f}\leq C\right\}.
\end{align*}
Both probabilities tend to one. All four asserted inequalities hold on
\[
\mathcal E_{n,M}=\mathcal E_{n,M}^{\rm lin}\cap\mathcal E_n^F
\cap\mathcal E_n^Q(1/4)\cap\mathcal E_n^{\rm ref}
\cap\bigcap_{k=1}^{K_{\rm cf}}\mathcal B_{n,k}.
\]
Indeed, the union bound and Step 3 give
\begin{align*}
\Pr(\mathcal E_{n,M}^c)
&\leq\Pr\{(\mathcal E_{n,M}^{\rm lin})^c\}
+\Pr\{(\mathcal E_n^F)^c\}
+\Pr\{(\mathcal E_n^Q(1/4))^c\}
+\Pr\{(\mathcal E_n^{\rm ref})^c\}\\
&\quad+\sum_{k=1}^{K_{\rm cf}}\Pr(\mathcal A_{n,k}^c)
+K_{\rm cf}n^{-3}=o(1).
\end{align*}

\QED

\subsection{Proofs of the Convergence and Regret Theorems}
\label{app: convergence proof}

\subsubsection{Proof of Theorem~\ref{thm: convergence rate}.}

We first bound the objective at the approximating function $\alpha_n^\dagger$, use approximate optimality to obtain a deterministic parameter radius, and then apply the uniform empirical bounds on that fixed ball.

\textbf{Step 1: Separate the two penalties and their scales.}
Write
\[
J_n(\alpha)=\sup_{f\in\mF_n}
\{\widehat\Psi_n(\alpha,f)-Q_n(f)-\lambda_n\|f\|_{\mF_n}^2\},\qquad
\widehat\Phi_n(\alpha)=J_n(\alpha)+\mu_n\|\alpha\|_{\mH_n}^2.
\]
The reference in $Q_n$ is held fixed as $\alpha$ varies. Since $0\in\mF_n$, $J_n(\alpha)\geq0$ for every candidate. For fixed positive constants and all sufficiently large $n$,
\[
c_\mu\bar a_n^2\leq\mu_n\leq C_\mu\bar a_n^2,\qquad
c_\lambda a_n^2\leq\lambda_n\leq C_\lambda a_n^2,\qquad
\eta_n\leq C_\eta\bar a_n^2.
\]
Choose $c_\lambda\geq2C_{\rm pen}/B$, where $C_{\rm pen}$ is the remainder constant in the bound
\[
Q_n(f)\geq\tfrac12c_\Sigma\|f\|_2^2-C_{\rm pen}a_n^2s_f^2
\]
from Step 7 of the proof of Lemma~\ref{lem: first stage stability}. Then $\lambda_n\geq2C_{\rm pen}a_n^2/B$, which ensures the lower-penalty inequality~\eqref{eq: feasible lower}. 

\textbf{Step 2: Upper-bound the objective at the approximating function $\alpha_n^\dagger$.}
For each fixed radius $R$, let $\mathcal E_{n,R}$ denote the common event defined at the end of Appendix~\ref{app: feasible criterion bounds}, on which~\eqref{eq: feasible linear}--\eqref{eq: feasible upper} hold uniformly over $\mH_{n,R}\times\mF_n$. Work first on $\mathcal E_{n,M_\star}$. The deterministic approximating function $\alpha_n^\dagger$ belongs to $\mH_{n,M_\star}$ by~\eqref{eq: approximate realizability}. Cauchy--Schwarz gives
\[
|\EE[m(X;\alpha_n^\dagger)^\top f(X)]|
\leq\varepsilon_n\|f\|_2.
\]
Put $u=\|f\|_2$, $v=\|f\|_{\mF_n}$ and use $s_f\leq1+v/\sqrt B$. On $\mathcal E_{n,M_\star}$, the feasible linear bound~\eqref{eq: feasible linear} and lower-penalty bound~\eqref{eq: feasible lower} give, for every $f\in\mF_n$,
\begin{align*}
\widehat\Psi_n(\alpha_n^\dagger,f)
&\leq\varepsilon_nu+C_La_nu+C_La_n^2(1+v/\sqrt B),\\
Q_n(f)+\lambda_nv^2
&\geq c_0u^2+\tfrac12\lambda_nv^2-C_Pa_n^2,
\end{align*}
where $C_L$ is the linear-bound constant for radius $M_\star$ and $C_P$ is the lower-penalty remainder constant. Subtracting and enlarging the range of the norm pair $(u,v)$ to $[0,\infty)^2$ yields
\begin{align*}
J_n(\alpha_n^\dagger)
&\leq\sup_{u,v\geq0}
\left\{A_nu-c_0u^2+D_nv-\frac{\lambda_n}{2}v^2\right\}
+(C_L+C_P)a_n^2,\\
A_n&=\varepsilon_n+C_La_n,\qquad D_n=C_La_n^2/\sqrt B.
\end{align*}
Complete the two squares separately:
\begin{align*}
A_nu-c_0u^2
&=\frac{A_n^2}{4c_0}-c_0\left(u-\frac{A_n}{2c_0}\right)^2,\\
D_nv-\frac{\lambda_n}{2}v^2
&=\frac{D_n^2}{2\lambda_n}
-\frac{\lambda_n}{2}\left(v-\frac{D_n}{\lambda_n}\right)^2.
\end{align*}
Their maximizers $A_n/(2c_0)$ and $D_n/\lambda_n$ are nonnegative. Hence, using $\lambda_n\geq c_\lambda a_n^2$,
\begin{align*}
J_n(\alpha_n^\dagger)
&\leq\frac{(\varepsilon_n+C_La_n)^2}{4c_0}
+\frac{C_L^2a_n^4}{2B\lambda_n}+(C_L+C_P)a_n^2\\
&\leq\frac{\varepsilon_n^2}{2c_0}
+\left(\frac{C_L^2}{2c_0}+\frac{C_L^2}{2Bc_\lambda}
+C_L+C_P\right)a_n^2.
\end{align*}
Taking
\[
C_*:=\max\left\{\frac1{2c_0},\,
\frac{C_L^2}{2c_0}+\frac{C_L^2}{2Bc_\lambda}+C_L+C_P\right\}
\]
we have, on $\mathcal E_{n,M_\star}$,
\begin{equation*}
0\leq J_n(\alpha_n^\dagger)
\leq C_*(a_n^2+\varepsilon_n^2)
\leq C_*\bar a_n^2.
\end{equation*}

\textbf{Step 3: Localize the estimator without circularity.}
Assumption~\ref{ass: parameters}(c), applied to the objective in~\eqref{eqn: min-max estimation}, gives the optimization event
\begin{equation}
\mathcal O_n=\left\{
\begin{aligned}
&\widehat\alpha_n\in\mH_n,\\
&\widehat\Phi_n(\widehat\alpha_n)
\leq\inf_{\alpha\in\mH_n}\widehat\Phi_n(\alpha)+C_{\rm opt}\eta_n
\end{aligned}\right\},
\qquad\Pr(\mathcal O_n)\to1.
\label{eq: convergence optimization event}
\end{equation}
On $\mathcal O_n\cap\mathcal E_{n,M_\star}$, feasibility of $\alpha_n^\dagger$ implies
\[
J_n(\widehat\alpha_n)+\mu_n\|\widehat\alpha_n\|_{\mH_n}^2
\leq J_n(\alpha_n^\dagger)+\mu_n M_\star+C_{\rm opt}\eta_n.
\]
Dropping the nonnegative first term and dividing by $\mu_n$ gives
\begin{equation*}
\|\widehat\alpha_n\|_{\mH_n}^2
\leq M_\star+
\frac{C_*(a_n^2+\varepsilon_n^2)+C_{\rm opt}\eta_n}{\mu_n}
\leq M_\star+C_{**},
\end{equation*}
where $C_{**}=(C_*+C_{\rm opt}C_\eta)/c_\mu$. Choose the deterministic radius $M=M_\star+C_{**}+1$. Then
\[
\mathcal O_n\cap\mathcal E_{n,M_\star}
\subseteq\{\widehat\alpha_n\in\mH_{n,M}\}.
\]
Having fixed $M$, define the common convergence event
\begin{equation}
\mathcal E_n^{\rm conv}
=\mathcal O_n\cap\mathcal E_{n,M_\star}\cap\mathcal E_{n,M}.
\label{eq: convergence common event}
\end{equation}
The union bound gives
\[
\Pr\{(\mathcal E_n^{\rm conv})^c\}
\leq\Pr(\mathcal O_n^c)+\Pr(\mathcal E_{n,M_\star}^c)
+\Pr(\mathcal E_{n,M}^c)=o(1).
\]
Thus $\widehat\alpha_n\in\mH_{n,M}$ and the feasible-criterion bounds hold throughout $\mathcal E_n^{\rm conv}$. The radius is determined using only the norm bounds on $\alpha_n^\dagger$ and the reference before invoking concentration on $\mH_{n,M}$.

\textbf{Step 4: Lower-bound the objective using an approximate residual test.}
For $\alpha\in\mH_{n,M}$, write $m_\alpha=m(\cdot;\alpha)$ and $u_\alpha=\|m_\alpha\|_2$. Assumption~\ref{ass: identification and spaces}(b) gives $g_{\alpha,n}\in\mF_n$ such that $\|g_{\alpha,n}-m_\alpha\|_2\leq\xi_n$ and $\|g_{\alpha,n}\|_{\mF_n}\leq L_M$, uniformly in $n$. Consequently,
\[
\EE[m_\alpha(X)^\top g_{\alpha,n}(X)]\geq u_\alpha^2-u_\alpha\xi_n,\qquad
\|g_{\alpha,n}\|_2\leq u_\alpha+\xi_n.
\]
Choose a fixed $t\in(0,1]$ with $2C_Qt\leq1/4$, where $C_Q$ is the constant in~\eqref{eq: feasible upper}. Star-shapedness makes $f_{\alpha,n}=tg_{\alpha,n}$ feasible, and
$s_{f_{\alpha,n}}\leq S_M:=\max\{1,tL_M/\sqrt B\}$.
On $\mathcal E_n^{\rm conv}$, inequalities~\eqref{eq: feasible linear} and~\eqref{eq: feasible upper}, together with the upper bound on $\lambda_n$, give
\begin{align*}
\widehat\Psi_n(\alpha,f_{\alpha,n})
&\geq t(u_\alpha^2-u_\alpha\xi_n)-C_Ma_nt(u_\alpha+\xi_n)-C_Ma_n^2S_M,\\
Q_n(f_{\alpha,n})&\leq C_Qt^2(u_\alpha+\xi_n)^2+C_Qa_n^2S_M^2,\\
\lambda_n\|f_{\alpha,n}\|_{\mF_n}^2&\leq C_\lambda a_n^2t^2L_M^2.
\end{align*}
Subtract the last two lines from the first. Using $(u_\alpha+\xi_n)^2\leq2u_\alpha^2+2\xi_n^2$ and Young's inequality for the remaining cross terms gives, uniformly on this ball,
\begin{equation}
J_n(\alpha)\geq\frac{t}{4}\|m(\cdot;\alpha)\|_2^2-C_M'(a_n^2+\xi_n^2).
\label{eq: residual test lower}
\end{equation}

\textbf{Step 5: Bound the moment and derivative errors.}
On $\mathcal E_n^{\rm conv}\subseteq\mathcal O_n$, approximate optimality~\eqref{eq: convergence optimization event} for~\eqref{eqn: min-max estimation} also gives
\begin{align*}
J_n(\widehat\alpha_n)
&=\widehat\Phi_n(\widehat\alpha_n)-\mu_n\|\widehat\alpha_n\|_{\mH_n}^2\\
&\leq\widehat\Phi_n(\alpha_n^\dagger)+C_{\rm opt}\eta_n
-\mu_n\|\widehat\alpha_n\|_{\mH_n}^2\\
&\leq J_n(\alpha_n^\dagger)+\mu_n M_\star+C_{\rm opt}\eta_n\\
&\leq(C_*+C_\mu M_\star+C_{\rm opt}C_\eta)\bar a_n^2.
\end{align*}
Because $\mathcal E_n^{\rm conv}$ implies $\widehat\alpha_n\in\mH_{n,M}$, we may combine this bound with~\eqref{eq: residual test lower} evaluated at $\alpha=\widehat\alpha_n$. It follows that
\[
\|m(\cdot;\widehat\alpha_n)\|_2\leq C'\sqrt{\bar a_n^2+\xi_n^2}\leq C'\delta_n.
\]
The candidates here are the five-component functions $\alpha=(b_1,b_2,h_4,h_5,h_6)\in\mH_{n,M}$. The structural truth is $\alpha_0=(\beta_{p,1},\beta_{p,2},h_{0,4},h_{0,5},h_{0,6})$ defined in Section~\ref{sec: estimation and policy learning}. Assumption~\ref{ass: identification and spaces}(a) supplies
\[
\sup_n\sup_{\alpha\in\mH_{n,M}}\|\alpha\|_\infty\leq A_M<\infty,
\qquad \|\alpha_0\|_\infty\leq A_0<\infty.
\]
Thus Lemma~\ref{lem: residual derivative link}, which permits $\alpha_0\notin\mH_n$, applies to $\widehat\alpha_n$ on $\mathcal E_n^{\rm conv}$. Its inequality~\eqref{eq: unweighted derivative bridge} gives
\[
\|T_0(\widehat\alpha_n-\alpha_0)\|_2
\leq K_M\|m(\cdot;\widehat\alpha_n)\|_2
\leq C_D\delta_n.
\]
This proves~\eqref{eq: residual convergence} after enlarging a fixed constant.

\textbf{Step 6: Recover coefficients through identification at the structural truth.}
On $\mathcal E_n^{\rm conv}$ from~\eqref{eq: convergence common event}, $\widehat\alpha_n\in\mH_{n,M}$ and $\|T_0(\widehat\alpha_n-\alpha_0)\|_2\leq C_D\delta_n$. Substituting this candidate into the supremum in~\eqref{eq: modulus of continuity} gives
\[
\|\widehat\alpha_n-\alpha_0\|_2
\leq\omega_{M,n}^D(C_D\delta_n),
\]
which is~\eqref{eq: derivative modulus convergence}.
The pointwise inverse and bounded envelopes give, for every fixed $0<\tau\leq1$,
\[
\sup_{\substack{\alpha\in\mH_{n,M}\\
\|T_0(\alpha-\alpha_0)\|_2\leq t}}
\|\alpha-\alpha_0\|_2^2
\leq C_M\tau^{-2}t^2+B_M^2\Pr\{s_A(X)\leq\tau\},
\]
with the nonnegative empty-supremum convention. The constants do not depend on $n$. First choose $\tau$ to make the second term small using almost-sure rank, then let $t\downarrow0$. Since $\delta_n\to0$, consistency follows.

Under the additional modulus scaling condition in Theorem~\ref{thm: convergence rate}, for every fixed $C_0>0$ there are $K_{C_0,M}<\infty$, $n_0$, and $t_0>0$ such that
\[
\omega_{M,n}^D(C_0t)\leq K_{C_0,M}\omega_{M,n}^D(t)
\qquad(n\geq n_0,\ 0<t\leq t_0).
\]
Take $C_0=C_D$. Since $\delta_n\to0$, for all sufficiently large $n$ and on $\mathcal E_n^{\rm conv}$,
\[
\|\widehat\alpha_n-\alpha_0\|_2
\leq\omega_{M,n}^D(C_D\delta_n)
\leq K_{C_D,M}\omega_{M,n}^D(\delta_n).
\]
Consequently,
\[
\Pr\{\|\widehat\alpha_n-\alpha_0\|_2
>K_{C_D,M}\omega_{M,n}^D(\delta_n)\}
\leq\Pr\{(\mathcal E_n^{\rm conv})^c\}\to0,
\]
which proves $\|\widehat\alpha_n-\alpha_0\|_2=O_p\{\omega_{M,n}^D(\delta_n)\}$. 

\QED

\subsubsection{Proof of Theorem~\ref{thm: regret bound}.}
\label{app: regret proof}

Condition on the fitted coefficient functions and fix $x$ outside the model's null set. Write $q(p)=b_1p+b_2p^2$, $\widehat q(p)=\widehat b_1p+\widehat b_2p^2$, and $e_j=\widehat b_j-b_j$, suppressing $x$. Let $p^*=\pi^*(x)$, $\widehat p=\widehat\pi(x)$, and $\Delta=p^*-\widehat p$. Since $\widehat b_2\leq-c_{p,2}<0$, the clipped price $\widehat p$ maximizes $\widehat q$ over $[p_1,p_2]$. 

Differentiating along the feasible path $\widehat p+t\Delta$, $0\leq t\leq1$, gives $\widehat q'(\widehat p)\Delta\leq0$, including at boundary optima. As $b_2\leq-c$, the quadratic expansion yields
\begin{align*}
0\leq q(p^*)-q(\widehat p)
&=\widehat q'(\widehat p)\Delta-(e_1+2\widehat p e_2)\Delta+b_2\Delta^2\\
&\leq-(e_1+2\widehat p e_2)\Delta-c\Delta^2\\
&\leq\frac{(e_1+2\widehat p e_2)^2}{4c}.
\end{align*}
The last inequality follows by completing the square. The policy-independent baseline in Proposition~\ref{prop: optimal policy} cancels when subtracting values, so integration over a fresh $X$ proves the first bound in~\eqref{eq: quadratic regret}. Since $0\leq\widehat\pi(X)\leq p_2$, Cauchy--Schwarz gives
\[
\{e_1(X)+2\widehat\pi(X)e_2(X)\}^2
\leq(1+4p_2^2)\{e_1(X)^2+e_2(X)^2\}.
\]
Integrating proves the second bound. Finally, $\sum_{j=1}^2\|e_j\|_2^2\leq\|\widehat\alpha_n-\alpha_0\|_2^2$. On the event $\mathcal E_n^{\rm conv}$ in~\eqref{eq: convergence common event}, the preceding proof therefore gives
\[
R(\widehat\pi)\leq\frac{1+4p_2^2}{4c}
\{\omega_{M,n}^D(C_D\delta_n)\}^2.
\]
Since $\Pr(\mathcal E_n^{\rm conv})\to1$, this proves the claimed $O_p$ rate. The stated modulus scaling condition yields $O_p\{\omega_{M,n}^D(\delta_n)^2\}$.

\QED

\section{RKHS Constructions and Explicit Learning Rates}
\label{sec: RKHS}

We construct a kernel for the residual class, give sufficient Sobolev conditions, and derive the critical radius and the convergence rates. We conclude with the additional learning conditions for the running example.

\subsection{Construction of the Residual Kernel}
\label{app: residual kernel}

We aim to construct a single scalar RKHS $\mathcal H_{\rm res}$ that contains every function
$Z\mapsto W_0(Z;\alpha)^\top f(X)$ on fixed parameter and adversary balls. The inclusion required by Assumption~\ref{ass: RKHS eigen-decay} can be verified by a finite sum of weighted product kernels built from the scalar coordinate kernels associated with the parameter and adversary kernels $K_{\mH}$ and $K_{\mF}$. This construction controls the actual nonlinear residual class. The additional eigenvalue-decay requirement is verified in the familiar Sobolev setting of Appendix~\ref{app: sobolev construction}.

We first reduce vector-valued coordinates to scalar kernels.
Let $K_j$ be a scalar kernel for parameter coordinate $\alpha_j$, and let $L_k$ be a scalar kernel for adversary coordinate $f_k$. In a vector-valued RKHS with matrix kernel $K$, the coordinate feature map is $x\mapsto K(\cdot,x)e_j$. Its inner product at $(x,x')$ is $e_j^\top K(x,x')e_j$, the diagonal scalar kernel. The reproducing identity represents $\alpha_j(x)$ using the original vector-valued element, so the scalar-coordinate RKHS norm is no larger than its vector-valued norm. This proves the coordinate reduction used below. See also \citet{de2013extension}.

Here we then list the products that actually occur.
Evaluate every $\rho_j$ at $h_0$ and put $v_G=h_{0,3}-h_{0,1}^2$. Expanding the eight residuals, in the stated row-by-row order, gives
\begin{equation}
\begin{aligned}
W_0(Z;\alpha)^\top f(X)={}&(G-h_{0,1})f_1+(P-h_{0,2})f_2+(G^2-h_{0,3})f_3\\
&+(P-h_{0,2})Yf_4+P(P-h_{0,2})f_5
+P^2(P-h_{0,2})f_6+\rho_1f_7+\rho_4f_8\\
&-h_4f_4-h_5f_5-h_6f_6-\rho_2\beta_{p,1}f_7-\rho_3\beta_{p,2}f_7\\
&-v_Gh_4f_8-\rho_5\beta_{p,1}f_8-\rho_6\beta_{p,2}f_8\\
&+v_Gh_5\beta_{p,1}f_8+v_Gh_6\beta_{p,2}f_8,
\end{aligned}
\label{eq: full residual expansion}
\end{equation}
where every term has the form
$q_\ell(z)f_{k_\ell}(x)\prod_{j\in I_\ell}\alpha_j(x)$.
The weights $q_\ell$ are fixed, bounded functions of observations and $h_0$. Each multiset $I_\ell$ has at most two elements. For example, the final two terms require products of a nested nuisance, a price coefficient, and an adversary coordinate. 

We can represent each product by a tensor-product feature map.
Choose feature maps $\phi_j$ and $\psi_k$ for $K_j$ and $L_k$. For one term, use
\[
\Phi_\ell(z)=q_\ell(z)\,\psi_{k_\ell}(x)\otimes
                  \bigotimes_{j\in I_\ell}\phi_j(x).
\]
The inner product of these maps at $z,z'$ equals
$q_\ell(z)q_\ell(z')L_{k_\ell}(x,x')\prod_{j\in I_\ell}K_j(x,x')$,
so it is a positive-semidefinite kernel. If $f_{k_\ell}$ and $\alpha_j$ are represented by coefficient vectors $v_{k_\ell}$ and $u_j$, then the product function is represented by
$v_{k_\ell}\otimes\bigotimes_{j\in I_\ell}u_j$.
Its norm is the product of the factor norms. Thus the term's RKHS norm is no greater than
$\|f_{k_\ell}\|_{L_{k_\ell}}\prod_{j\in I_\ell}\|\alpha_j\|_{K_j}$.
We also define the empty product to be one. 

Take the direct sum of the preceding feature maps, and include each unweighted adversary feature map as well. The resulting kernel is
\begin{equation*}
K_*(z,z')=\sum_{k=1}^{8}L_k(x,x')+
\sum_{\ell}q_\ell(z)q_\ell(z')L_{k_\ell}(x,x')
                         \prod_{j\in I_\ell}K_j(x,x').
\end{equation*}
It is bounded and positive semidefinite, since it is a finite sum of bounded feature-map kernels. On $\mH_M\times\mF_B$, the squared norm of a coefficient vector representing the residual expansion is bounded by a fixed multiple of
$B\sum_\ell M^{|I_\ell|}$, hence by $C B(1+M+M^2)$.
If some coordinates have continuous embedding constants, those constants are absorbed into $C$. A single lifted adversary coordinate is represented in its own summand and has norm at most $\sqrt B$. Consequently the entire residual class and all lifted coordinate classes lie in fixed-radius balls of the $K_*$-RKHS.

This verifies the inclusion part of Assumption~\ref{ass: RKHS eigen-decay}. 

\subsection{Sufficient Conditions Using Sobolev Spaces}
\label{app: sobolev construction}

This subsection shows that the inclusion and eigenvalue-decay requirements in Assumption~\ref{ass: RKHS eigen-decay} hold together for standard Sobolev coordinate spaces under the design conditions given below. The Sobolev multiplication bound controls the finite collection of products appearing in the PRINT residuals, and bounded observation-dependent weights preserve the spectral exponent. Thus Assumption~\ref{ass: RKHS eigen-decay} admits a familiar smoothness model even when these weights are nonsmooth functions of $Z$.

\textbf{Part 1: Control the function products in the residual.}
Suppose every parameter and adversary coordinate embeds continuously into
$\mathcal A=H^s([0,1]^d)$ with $s>d/2$, or into its periodic version used below, equipped with an equivalent bounded Sobolev kernel. The Sobolev multiplication theorem gives
\[
\|uv\|_{\mathcal A}\leq C_{\rm mult}\|u\|_{\mathcal A}\|v\|_{\mathcal A}
\]
\citep[Theorem~3.9]{behzadan2021multiplication}.
Applying this inequality twice controls products of three functions. Therefore each unweighted factor
$f_{k_\ell}\prod_{j\in I_\ell}\alpha_j$
in~\eqref{eq: full residual expansion} belongs to a fixed $\mathcal A$-ball on fixed parameter and adversary balls.

\textbf{Part 2: Specify the spectral condition and an explicit example.}
Assume the embedding $J:\mathcal A\to L^2(P_X)$ satisfies
\[
\sigma_j(J)^2=\nu_j(K_{\mathcal A};P_X)\leq Cj^{-2s/d}.
\]
The following example verifies this condition for a periodic Sobolev space.

\begin{example}[Periodic Sobolev kernel and bounded-density designs]
\label{ex: periodic sobolev spectrum}
Identify opposite faces of $[0,1]^d$ to form the torus $\mathbb T^d$, and take $\mathcal A=H_{\rm per}^s(\mathbb T^d)$ with $s>d/2$. In this example, the coordinate spaces in Part 1 embed continuously into this periodic space. Let $U$ be the uniform law and use the Fourier basis
$e_k(x)=\exp(2\pi\mathrm{i}\,k^\top x)$, $k\in\mathbb Z^d$. Set
\[
a_k=(1+4\pi^2\|k\|^2)^{-s},\qquad
u_k=\int u(x)\overline{e_k(x)}\,dU(x),\qquad
\|u\|_{\mathcal A}^2=\sum_{k\in\mathbb Z^d}\frac{|u_k|^2}{a_k}.
\]
This Sobolev norm has reproducing kernel
\[
K_{\mathcal A}(x,x')=\sum_{k\in\mathbb Z^d}a_k e_k(x)\overline{e_k(x')}.
\]
The terms indexed by $k$ and $-k$ pair to give a real kernel. Splitting the nonzero frequencies into shells $2^m\leq\|k\|<2^{m+1}$ gives
\[
\sum_{k\in\mathbb Z^d}a_k
\leq1+C_{d,s}\sum_{m\geq0}2^{md}2^{-2sm}<\infty,
\]
because $2s>d$. Thus the kernel series converges absolutely and uniformly,
$\sup_xK_{\mathcal A}(x,x)=\sum_k a_k<\infty$, and
$|u(x)|\leq(\sum_k a_k)^{1/2}\|u\|_{\mathcal A}$.
Orthonormality under $U$ now yields
\[
\int K_{\mathcal A}(x,x')e_\ell(x')\,dU(x')
=\sum_k a_ke_k(x)\int\overline{e_k(x')}e_\ell(x')\,dU(x')
=a_\ell e_\ell(x).
\]
Hence the kernel eigenvalues under $U$ are exactly the $a_k$.

To order these eigenvalues, let $N(R)=\#\{k\in\mathbb Z^d:\|k\|\leq R\}$. Inscribed and enclosing cubes give
\[
\left(2\left\lfloor\frac R{\sqrt d}\right\rfloor+1\right)^d
\leq N(R)\leq(2\lfloor R\rfloor+1)^d,
\]
so $c_dR^d\leq N(R)\leq C_dR^d$ for all sufficiently large $R$. Enumerate the frequencies so that $r_j=\|k_j\|$ is nondecreasing. Since $j\leq N(r_j)$, the upper counting bound gives $r_j\geq(j/C_d)^{1/d}$ for large $j$. Conversely, $N\{(j/c_d)^{1/d}\}\geq j$ gives $r_j\leq(j/c_d)^{1/d}$. Therefore
\[
\nu_j(K_{\mathcal A};U)
=a_{k_j}=(1+4\pi^2r_j^2)^{-s}\asymp j^{-2s/d}.
\]

Now suppose $P_X$ has density $w\leq C_w$ relative to $U$. Write $J_P$ and $J_U$ for the embeddings into $L^2(P_X)$ and $L^2(U)$. For every $u\in\mathcal A$,
\[
\|J_Pu\|_{L^2(P_X)}^2
=\int|u(x)|^2w(x)\,dU(x)
\leq C_w\|J_Uu\|_{L^2(U)}^2.
\]
The change-of-measure map from $L^2(U)$ to $L^2(P_X)$ is bounded, so its composition with the compact $J_U$ makes $J_P$ compact. The min--max formula for squared singular values gives
\begin{align*}
\nu_j(K_{\mathcal A};P_X)
&=\inf_{\substack{E\subset\mathcal A\\\dim E=j-1}}
\sup_{\substack{u\perp_{\mathcal A}E\\\|u\|_{\mathcal A}=1}}
\|J_Pu\|_{L^2(P_X)}^2\\
&\leq C_w\inf_{\substack{E\subset\mathcal A\\\dim E=j-1}}
\sup_{\substack{u\perp_{\mathcal A}E\\\|u\|_{\mathcal A}=1}}
\|J_Uu\|_{L^2(U)}^2\\
&=C_w\nu_j(K_{\mathcal A};U)\leq Cj^{-2s/d}.
\end{align*}
Thus a bounded density preserves the required upper bound. No positive lower bound on $w$ is needed.
\end{example}

\textbf{Part 3: Add the observable weights without losing the exponent.}
Include the constant weight $q_0=1$ and let $L$ be the finite total number of weights. Define
\[
K_{\rm res}(z,z')=K_{\mathcal A}(x,x')
                  \sum_{\ell=0}^{L-1}q_\ell(z)q_\ell(z').
\]
Its direct-sum feature representation gives
\[
\left\|\sum_{\ell=0}^{L-1}q_\ell u_\ell\right\|_{\mathcal H_{\rm res}}^2
\leq\sum_{\ell=0}^{L-1}\|u_\ell\|_{\mathcal A}^2.
\]
Taking $u_\ell=f_{k_\ell}\prod_{j\in I_\ell}\alpha_j$ for the residual terms and using Part 1 proves fixed-ball containment. Taking $u_0=f_k$ and all other coordinates zero covers each lifted adversary coordinate.
To verify the spectrum, define
\[
M_q:(L^2(P_X))^L\longrightarrow L^2(P_Z),\qquad
M_q(u_0,\ldots,u_{L-1})=\sum_{\ell=0}^{L-1}q_\ell(Z)u_\ell(X).
\]
Cauchy--Schwarz gives
\[
\|M_q u\|_{L^2(P_Z)}^2
\leq\left(\sum_{\ell=0}^{L-1}\|q_\ell\|_\infty^2\right)
       \sum_{\ell=0}^{L-1}\|u_\ell\|_{L^2(P_X)}^2,
\]
so $M_q$ is bounded even if the weights are nonsmooth functions of $Z$.
The feature-space embedding of $K_{\rm res}$ factors as
$T=M_q\circ J^{\oplus L}:\mathcal A^L\to L^2(P_Z)$.
Put $S=J^{\oplus L}$. Its singular values repeat those of $J$ exactly $L$ times, so
\[
\sigma_j(S)^2=\sigma_{\lceil j/L\rceil}(J)^2
=\nu_{\lceil j/L\rceil}(K_{\mathcal A};P_X).
\]
Let $(\sigma_r(S),v_r,w_r)$ be a singular system of the compact operator $S$, with orthonormal domain and range vectors. Its truncation
\[
S_{j-1}u=\sum_{r=1}^{j-1}\sigma_r(S)\langle u,v_r\rangle w_r
\]
omits zero singular-value terms, has rank at most $j-1$, and satisfies $\|S-S_{j-1}\|=\sigma_j(S)$. For $j=1$, take $S_0=0$. The composed approximation $M_qS_{j-1}$ also has rank at most $j-1$. Applying the best finite-rank approximation characterization of singular values therefore gives
\begin{align*}
\sigma_j(T)
&=\inf_{\operatorname{rank}(R)<j}\|T-R\|\\
&\leq\|T-M_qS_{j-1}\|
=\|M_q(S-S_{j-1})\|\\
&\leq\|M_q\|\,\|S-S_{j-1}\|
=\|M_q\|\sigma_j(S),
\end{align*}
where all operator norms use the displayed Hilbert-space domains and ranges.
Therefore
\[
\sigma_j(T)^2
\leq\|M_q\|^2\nu_{\lceil j/L\rceil}(K_{\mathcal A};P_X)
\leq C\|M_q\|^2 L^{2s/d}j^{-2s/d},
\]
using $\lceil j/L\rceil\geq j/L$. Since $L$ and $\|M_q\|$ are fixed, the last expression is at most $C'j^{-2s/d}$. The integral operator of $K_{\rm res}$ is $TT^*$, so its nonzero eigenvalues are these squared singular values. This proves the spectrum required in Assumption~\ref{ass: RKHS eigen-decay} for this construction.

\textbf{Part 4: Check residual-test approximation for Assumption~\ref{ass: identification and spaces}(b).}
A sufficient check is that
$h_{0,4},h_{0,5},h_{0,6},v_G$ and the six conditional moment functions $\Omega_j,\Upsilon_j$ belong to the same Sobolev algebra, and that $\mathcal A^8$ is contained in each admissible adversary class $\mF_n$, with a continuous embedding bound uniform in $n$. Indeed, the first three coordinates of $m$ are zero, the next three are nuisance differences, and the last two are the sums and products in the proof of Lemma~\ref{lem: price moment stability}. Repeated multiplication then gives $\|m(\cdot;\alpha)\|_{\mF_n}\leq L_M$ on every fixed ball, verifying Assumption~\ref{ass: identification and spaces}(b) with $g_{\alpha,n}=m(\cdot;\alpha)$ and $\xi_n=0$.

\subsection{Critical-Radius Calculation and Proofs of the RKHS Corollaries}
\label{app: rkhs rate proofs}

\subsubsection{Proof of Corollaries~\ref{cor: RKHS derivative convergence rate} and~\ref{cor: RKHS convergence rate}.}

\textbf{Step 1: Derive an eigenvalue upper bound for local complexity.}
Consider an RKHS ball of fixed norm radius $R$ with integral-operator eigenvalues $\nu_j>0$ and associated orthonormal eigenfunctions $\phi_j$ in $L^2(P_Z)$. We work in the image of its embedding into $L^2(P_Z)$. If $g=\sum_j c_j\phi_j$ has RKHS norm at most $R$ and $\|g\|_2\leq t$, then
\[
\sum_j c_j^2\leq t^2,\qquad \sum_j\frac{c_j^2}{\nu_j}\leq R^2.
\]
Put $\tau_j^2=\min(t^2,R^2\nu_j)$. Since the maximum of two nonnegative terms is at most their sum,
\[
\sum_j\frac{c_j^2}{\tau_j^2}
\leq\frac1{t^2}\sum_jc_j^2+\frac1{R^2}\sum_j\frac{c_j^2}{\nu_j}\leq2.
\]
For $\xi_j=n^{-1}\sum_{i=1}^n\epsilon_i\phi_j(Z_i)$, independence of the Rademacher signs and orthonormality give $\EE\xi_j^2=1/n$. Cauchy--Schwarz, followed by Jensen's inequality, therefore yields
\begin{equation}
\mathcal R_n(t)
\leq\EE\left(2\sum_j\tau_j^2\xi_j^2\right)^{1/2}
\leq\left\{\frac2n\sum_{j\geq1}\min(t^2,R^2\nu_j)\right\}^{1/2}.
\label{eq: rkhs spectral complexity detailed}
\end{equation}
One may first truncate the eigenfunction expansion and then pass to the limit using the finite trace and separability. Null directions of the embedding contribute neither population norm nor sample evaluations on a common probability-one set. This is the upper-bound part of the localized eigenvalue calculation in \citet[Theorem~2.1]{mendelson2003performance}.

\textbf{Step 2: Split the spectrum and solve the critical-radius inequality.}
Write $\zeta=2s/d>1$. For sufficiently small $t>0$, choose
$J_t=\max\{1,\lceil(CR^2/t^2)^{1/\zeta}\rceil\}$.
Using $\nu_j\leq Cj^{-\zeta}$, split the sum into its head and tail:
\begin{align*}
\sum_{j\geq1}\min(t^2,R^2\nu_j)
&\leq J_t t^2+CR^2\sum_{j>J_t}j^{-\zeta}\\
&\leq J_t t^2+\frac{CR^2}{\zeta-1}J_t^{1-\zeta}\\
&\leq C_Rt^{2-d/s}.
\end{align*}
The second line uses $\sum_{j>J}j^{-\zeta}\leq\int_J^\infty x^{-\zeta}\,dx$. This is where $s>d/2$ is used. Consequently,
\[
\mathcal R_n(t)\leq C_R n^{-1/2}t^{1-d/(2s)}.
\]
To make this at most $t^2/64$, it suffices that
$t^{1+d/(2s)}\geq64C_Rn^{-1/2}$. Thus
\[
t=D n^{-s/(2s+d)}
\]
satisfies the critical-radius inequality for a sufficiently large fixed $D$. The constants can depend on the fixed ball radii, as allowed by Assumption~\ref{ass: parameters}(a).

\textbf{Step 3: Apply the bound to the actual residual and adversary classes.}
Fix $M,B$ and collect the nine scalar classes in
\[
\mathcal C_{M,B}=\mathcal G_{M,B}\cup
\bigcup_{k=1}^8\widetilde{\mathcal F}_{B,k}.
\]
Assumption~\ref{ass: RKHS eigen-decay} places this union in the residual RKHS ball of radius $R_{M,B}$. The reproducing inequality gives a common fixed envelope
\[
E_{M,B}:=\max\left\{1,R_{M,B}\sup_z\sqrt{K_{\rm res}(z,z)}\right\},
\qquad \sup_{g\in\mathcal C_{M,B}}\|g\|_\infty\leq E_{M,B}.
\]
Normalize by this envelope, as in the critical-radius definition preceding Assumption~\ref{ass: parameters}. Since an RKHS ball is star-shaped around $0$,
\[
\operatorname{star}(\mathcal C_{M,B}/E_{M,B})
\subseteq\left\{g\in\mathcal H_{\rm res}:
\|g\|_{\mathcal H_{\rm res}}\leq R_{M,B}/E_{M,B}\right\}.
\]
Restricting both sets to $\|g\|_{L^2(P_Z)}\leq t$ preserves the inclusion. Hence~\eqref{eq: rkhs spectral complexity detailed} and Step 2 imply
\begin{align*}
\mathcal R_{n,P_Z}(t,\mathcal C_{M,B}/E_{M,B})
&\leq\left\{\frac2n\sum_{j\geq1}
\min\left(t^2,\frac{R_{M,B}^2}{E_{M,B}^2}
\nu_j(K_{\rm res};P_Z)\right)\right\}^{1/2}\\
&\leq C_{M,B}n^{-1/2}t^{1-d/(2s)}.
\end{align*}
All constants are uniform in $n$. Put $b_n=n^{-s/(2s+d)}$ and $t_n=D_{M,B}b_n$. Since $\{s/(2s+d)\}\{1+d/(2s)\}=1/2$,
\[
\frac{\mathcal R_{n,P_Z}(t_n,\mathcal C_{M,B}/E_{M,B})}{t_n^2}
\leq C_{M,B}D_{M,B}^{-1-d/(2s)}\leq\frac1{64}
\]
whenever $D_{M,B}^{1+d/(2s)}\geq64C_{M,B}$, for all sufficiently large $n$. Every subclass inherits this bound. Moreover, lifting a coordinate preserves its norm and sample values, so
\[
\mathcal R_{n,P_X}(t,\mathcal F_{B,k}/E_{M,B})
=\mathcal R_{n,P_Z}(t,\widetilde{\mathcal F}_{B,k}/E_{M,B}),
\qquad k=1,\ldots,8.
\]
Thus the residual class and all eight adversary-coordinate classes have critical radii at most $D_{M,B}b_n$. With $B$ fixed, this verifies Assumption~\ref{ass: parameters}(a) for every fixed $M$.

Together with Assumption~\ref{ass: identification and spaces}, the nuisance conditions in Assumption~\ref{ass: parameters}(b), and the tuning and optimization conditions in part (c), Theorem~\ref{thm: convergence rate} now applies. It supplies a deterministic radius $M_0$ and the event $\mathcal E_n^{\rm conv}$ in~\eqref{eq: convergence common event}, with $\Pr(\mathcal E_n^{\rm conv})\to1$, on which
\begin{align*}
\|\widehat\alpha_n\|_{\mH_n}^2&\leq M_0,\\
\|m(\cdot;\widehat\alpha_n)\|_2+
\|T_0(\widehat\alpha_n-\alpha_0)\|_2
&\leq C\delta_n\\
&=C\{n^{-s/(2s+d)}+r_n+\sqrt{\log n/n}+\varepsilon_n+\xi_n\}.
\end{align*}
The moment and derivative errors are therefore $O_p(\delta_n)$.

\textbf{Step 4: Remove the smaller confidence term in the moment and derivative error bounds.}
Set $\gamma=s/(2s+d)<1/2$. Then
\[
\frac{\sqrt{\log n/n}}{n^{-\gamma}}
=n^{\gamma-1/2}\sqrt{\log n}\longrightarrow0.
\]
Hence $\delta_n=O(n^{-\gamma}+r_n+\varepsilon_n+\xi_n)$, proving Corollary~\ref{cor: RKHS derivative convergence rate}.

\textbf{Step 5: Convert derivative error into coefficient error under uniform rank.}
For $\alpha\in\mH_{n,M_0}$, put $u=\alpha-\alpha_0$, $u_b=(u_{b1},u_{b2})$, and $d_u=T_0u$. The inverse identity~\eqref{eq: pointwise derivative inverse} used in Lemma~\ref{lem: price moment stability} is
\[
A(X)u_b(X)
=-\binom{d_{u,7}}{d_{u,8}-v_G(d_{u,4}-b_{01}d_{u,5}-b_{02}d_{u,6})}
=:-r_u(X).
\]
Here $b_{01},b_{02}$ are the true price coefficients and $v_G=\Var(G\mid X)$. Cauchy--Schwarz applied to the second coordinate gives
\[
|r_{u,2}|^2
\leq\{1+v_G^2(1+b_{01}^2+b_{02}^2)\}
(d_{u,4}^2+d_{u,5}^2+d_{u,6}^2+d_{u,8}^2).
\]
Consequently $\|r_u(X)\|_{\ell^2}^2\leq C_{\rm inv}\|d_u(X)\|_{\ell^2}^2$, where
\[
C_{\rm inv}:=1+\|v_G\|_\infty^2
\{1+\|b_{01}\|_\infty^2+\|b_{02}\|_\infty^2\}.
\]
Under~\eqref{eq: uniform price rank}, $\|A(X)^{-1}\|_{\rm op}\leq\kappa_A^{-1}$ almost surely. Squaring and integrating $u_b=-A^{-1}r_u$, then using $u_j=-d_{u,j}$ for $j=4,5,6$, yields
\begin{align*}
\|u_b\|_2^2&\leq C_{\rm inv}\kappa_A^{-2}\|d_u\|_2^2,\\
\|u\|_2^2
&=\|u_b\|_2^2+\sum_{j=4}^6\|u_j\|_2^2
\leq(1+C_{\rm inv}\kappa_A^{-2})\|d_u\|_2^2.
\end{align*}
The residual--derivative comparison~\eqref{eq: unweighted derivative bridge} makes the corresponding residual conversion explicit:
\[
\|u_b\|_2\leq\frac{K_{M_0}\sqrt{C_{\rm inv}}}{\kappa_A}
\|m(\cdot;\alpha)\|_2,\qquad
\|u\|_2\leq K_{M_0}\sqrt{1+C_{\rm inv}\kappa_A^{-2}}
\|m(\cdot;\alpha)\|_2.
\]

On $\mathcal E_n^{\rm conv}$, Step 3 gives
$\|T_0(\widehat\alpha_n-\alpha_0)\|_2\leq C_D\delta_n$. Substitution into the full coefficient bound yields
\[
\|\widehat\alpha_n-\alpha_0\|_2
\leq C_D\sqrt{1+C_{\rm inv}\kappa_A^{-2}}\,\delta_n,
\]
which, with Step 4, proves~\eqref{eq: stone coefficients}. The regret bound~\eqref{eq: quadratic regret} uses only the two price coefficients. Writing $\widehat u_b=(\widehat\beta_{p,1}-\beta_{p,1},\widehat\beta_{p,2}-\beta_{p,2})$, we obtain on the same event
\begin{align*}
\mathcal V(\pi^*)-\mathcal V(\widehat\pi_n)
&\leq\frac{1+4p_2^2}{4c}\|\widehat u_b\|_2^2\\
&\leq\frac{(1+4p_2^2)C_{\rm inv}}{4c\,\kappa_A^2}
\|T_0(\widehat\alpha_n-\alpha_0)\|_2^2\\
&\leq\frac{(1+4p_2^2)C_{\rm inv}C_D^2}{4c\,\kappa_A^2}\delta_n^2.
\end{align*}
This displays the $\kappa_A^{-1}$ factor in the price-coefficient conversion and its square in the regret conversion. Since $\Pr(\mathcal E_n^{\rm conv})\to1$, Step 4 and $(u+v+w+z)^2\leq4(u^2+v^2+w^2+z^2)$ give
\[
\mathcal V(\pi^*)-\mathcal V(\widehat\pi_n)
=O_p\{(n^{-\gamma}+r_n+\varepsilon_n+\xi_n)^2\}
=O_p\{n^{-2\gamma}+r_n^2+\varepsilon_n^2+\xi_n^2\},
\]
which is~\eqref{eq: stone regret}. 

\QED

\section{Weak Identification and Further Regret Rates}
\label{app: further regret rates}

This appendix gives quantitative rate conversions beyond uniform rank. We first bound the probability that the smallest singular value of $A(X)$ is close to zero and illustrate the condition with a bounded pricing design. We then give alternative spectral links and compare their implications for regret.

\subsection{Small Singular Values and a Bounded Pricing Example}
\label{app: weak identification}

Write $s_A(X)=\sigma_{\min}\{A(X)\}$. Building on Lemma~\ref{lem: qualitative derivative stability}, a bound on the probability of small $s_A(X)$ gives an explicit coefficient and regret rate.

\begin{corollary}[Rates under a tail bound for the smallest singular value]\label{cor: weak rank}
Suppose the conditions of Theorems~\ref{thm: convergence rate}--\ref{thm: regret bound} hold and, for constants $C_A,\nu,t_0>0$,
\begin{equation}
\Pr\{s_A(X)\leq t\}\leq C_A t^\nu,\qquad 0<t<t_0.
\label{eq: weak rank tail}
\end{equation}
Then
$\|\widehat\alpha_n-\alpha_0\|_2=O_p\{\delta_n^{\nu/(\nu+2)}\}$
and
$\mathcal V(\pi^*)-\mathcal V(\widehat\pi_n)=O_p\{\delta_n^{2\nu/(\nu+2)}\}$.
The same argument gives $\omega_{M,n}^D(t)\lesssim_M t^{\nu/(\nu+2)}$ for sufficiently small $t$. Under the residual RKHS condition, substitute
$\delta_n=n^{-s/(2s+d)}+r_n+\varepsilon_n+\xi_n+\sqrt{\log n/n}$.
\end{corollary}

\noindent\textbf{Proof of Corollary~\ref{cor: weak rank}.}
Fix $M<\infty$ and a candidate $\alpha\in\mH_{n,M}$ with
$\|T_0(\alpha-\alpha_0)\|_2\leq\epsilon$. The split inequality~\eqref{eq: qualitative derivative split} in Lemma~\ref{lem: qualitative derivative stability} and the tail condition~\eqref{eq: weak rank tail} give, for $0<\tau<\min(1,t_0)$,
\[
\|\alpha-\alpha_0\|_2^2
\leq C_M\tau^{-2}\epsilon^2+B_M^2 C_A\tau^\nu
\leq C'_M\{\tau^{-2}\epsilon^2+\tau^\nu\}.
\]
For sufficiently small $\epsilon>0$, choose $\tau=\epsilon^{2/(\nu+2)}$. Both terms then equal $\epsilon^{2\nu/(\nu+2)}$. For $\epsilon=0$, let $\tau\downarrow0$. Taking the candidate supremum proves
$\omega_{M,n}^D(\epsilon)\leq C_M\epsilon^{\nu/(\nu+2)}$, uniformly in $n$.
On $\mathcal E_n^{\rm conv}$ from~\eqref{eq: convergence common event}, the convergence theorem gives $\widehat\alpha_n\in\mH_{n,M}$ and $\|T_0(\widehat\alpha_n-\alpha_0)\|_2\leq C\delta_n$ for fixed $M,C<\infty$. Thus, with $C_R=(1+4p_2^2)/(4c)$,
\begin{align*}
\|\widehat\alpha_n-\alpha_0\|_2
&\leq\omega_{M,n}^D(C\delta_n)
\leq C'_M\delta_n^{\nu/(\nu+2)},\\
\mathcal V(\pi^*)-\mathcal V(\widehat\pi_n)
&\leq C_R\|\widehat\alpha_n-\alpha_0\|_2^2
\leq C''_M\delta_n^{2\nu/(\nu+2)}.
\end{align*}
The second line uses~\eqref{eq: quadratic regret}. Since $\Pr(\mathcal E_n^{\rm conv})\to1$, these inequalities prove the two stated $O_p$ rates.
The RKHS specialization follows by substituting its critical-radius rate into~\eqref{eq: total learning error}.

\QED

\begin{example}[A bounded pricing design with nonuniform identification strength]\label{ex: thin weak rank pricing}

Let $X\in(0,1)$ have distribution $\Pr(X\leq x)=x^\eta$, $0<x<1$, with $\eta>0$. Let $G$ be uniform on $\{-1,0,1\}$ and let $V$ be an independent Rademacher variable. Assume $(G,V)$ is independent of $X$. For $\kappa>0$, set $a(x)=x^\kappa$ and
\begin{equation*}
P=4+G+\sqrt{1+a(X)G^2}\,V.
\end{equation*}
The support of $P$ is contained in $[3-\sqrt2,5+\sqrt2]\subset[1,7]$, so the design uses bounded nonnegative prices. Writing $\mu(x,g)=\EE[P\mid X=x,G=g]$ and $R_P=P-\mu(X,G)$, its conditional moments are
\[
\mu(x,g)=4+g,\qquad
v(x,g)=\EE[R_P^2\mid X=x,G=g]=1+a(x)g^2,\qquad
\EE[R_P^3\mid X,G]=0.
\]
Using the observable moment identities in Appendix~\ref{app: higher moment details}, symmetry gives
\begin{align*}
\Omega_2(x)&=\Cov(G,v\mid X=x)=0,&
\Upsilon_2(x)&=\Cov(G^2,v\mid X=x)=\frac29a(x),\\
\Omega_3(x)&=\Cov(G,2\mu v\mid X=x)=\frac43\{1+a(x)\},&
\Upsilon_3(x)&=\Cov(G^2,2\mu v\mid X=x)=\frac{16}{9}a(x).
\end{align*}
Consequently the price-moment matrix in~\eqref{eq: rank matrix A} is
\begin{equation*}
A(x)=\begin{pmatrix}
0&\frac43\{1+a(x)\}\\[2pt]
\frac29a(x)&\frac{16}{9}a(x)
\end{pmatrix},\qquad
\det A(x)=-\frac8{27}a(x)\{1+a(x)\}.
\end{equation*}
For $0<a(x)\leq1$, $\sigma_{\max}\{A(x)\}$ is bounded above and bounded away from zero. The identity
$\sigma_{\min}(A)=|\det A|/\sigma_{\max}(A)$ therefore gives fixed $0<c_1\leq c_2<\infty$ such that
\[
c_1a(x)\leq s_A(x)\leq c_2a(x).
\]
Thus, for $0<t<c_1$,
\begin{equation*}
\Pr\{s_A(X)\leq t\}
\leq\Pr\{c_1X^\kappa\leq t\}
=c_1^{-\eta/\kappa}t^{\eta/\kappa}.
\end{equation*}
Condition~\eqref{eq: weak rank tail} holds with $\nu=\eta/\kappa$. The matrix is nonsingular almost surely, but its essential minimum singular value is zero, so uniform rank fails.

The exponent $\kappa$ controls how quickly instrument-dependent price dispersion disappears near $X=0$, while $\eta$ controls the customer mass there. The mean-price shift $4+G$ persists even as identification weakens. This example verifies the observable rank and tail conditions. Under the remaining hypotheses of Corollary~\ref{cor: weak rank}, taking $\eta=\kappa=1$ gives coefficient error $O_p(\delta_n^{1/3})$ and regret $O_p(\delta_n^{2/3})$.

\end{example}

\subsection{Derivative Gram Matrices and Alternative Spectral Links}
\label{app: spectral link}
The residual RKHS in Assumption~\ref{ass: RKHS eigen-decay} controls statistical complexity. The parameter RKHS and unweighted derivative operator below control identification. Their spectra describe different objects. The following links impose additional sufficient identification conditions.

\noindent\textbf{Common setting for the spectral links.}
Fix $M<\infty$. The two spectral lemmas below and the intervening rate corollary use the following additional common setting.
\begin{enumerate}[label=(\alph*)]
\item \textbf{Parameter RKHS and spectrum.}
Let $\mathcal H(K_{\mH})$ denote a common separable parameter RKHS with bounded measurable matrix-valued kernel $K_{\mH}$. Write $\lambda_j>0$ for the nonincreasing eigenvalues of its integral operator and $e_j$ for the corresponding $L^2(P_X)^5$-orthonormal eigenfunctions. These eigenfunctions span the closure of the RKHS image in $L^2(P_X)^5$. Bounded kernel diagonal gives $\sum_j\lambda_j<\infty$.

\item \textbf{Common norm bound for candidates and truth.}
Assume, uniformly in $n$,
\[
\sup_{\alpha\in\mH_{n,M}}\|\alpha\|_{\mathcal H(K_{\mH})}\leq\sqrt M,
\qquad \|\alpha_0\|_{\mathcal H(K_{\mH})}\leq\sqrt M.
\]
The truth must belong to the common parameter RKHS, but need not lie in the estimation sieve $\mH_n$. For $u=\alpha-\alpha_0$ and $a_j=\langle u,e_j\rangle_2$, the triangle inequality and spectral representation give
\begin{equation}
u=\sum_{j\geq1}a_je_j,\qquad
\|u\|_2^2=\sum_{j\geq1}a_j^2,\qquad
\sum_{j\geq1}\frac{a_j^2}{\lambda_j}\leq\|u\|_{\mathcal H(K_{\mH})}^2\leq4M.
\label{eq: parameter spectral expansion}
\end{equation}
Here the expansion converges in $L^2$. The first inequality allows RKHS functions that vanish $P_X$-almost surely. It is an equality for the minimum-norm representative of $u$.

\item \textbf{Derivative operator and Gram matrices.}
The unweighted derivative is $T_0u(X)=J_0(X)u(X)$. Bounded $J_0$ makes $T_0:L^2(P_X)^5\to L^2(P_X)^8$ bounded, so the spectral partial sums also converge after applying $T_0$. Define
\begin{equation*}
\gamma_{ij}^{D}=\EE[(T_0e_i)(X)^\top(T_0e_j)(X)],
\qquad\Gamma_J^{D}=(\gamma_{ij}^{D})_{1\leq i,j\leq J}.
\end{equation*}
The matrix $\Gamma_J^D$ measures derivative strength on the first $J$ parameter directions, while $\gamma_{ij}^D$ with $i\leq J<j$ measures interaction between retained and tail directions.
\end{enumerate}
The first link controls both finite-block eigenvalues and cancellation between low and high frequencies.

\begin{lemma}[Derivative Gram link with cross-block control]\label{lem: modulus of continuity}
Suppose the spectral representation~\eqref{eq: parameter spectral expansion} is available. Suppose there are $c_\times<\infty$ and positive $\tau_J$ such that, for every $J$,
\begin{equation}
\lambda_{\min}(\Gamma_J^{D})\geq\tau_J,\qquad
|\gamma_{ij}^{D}|\leq c_\times\tau_J\quad\text{for all }i\leq J<j.
\label{eq: derivative gram link}
\end{equation}
Then, for every $t\geq0$,
\begin{align}
\{\omega_{M,n}^D(t)\}^2
&\leq\inf_{J\geq1}\left\{\frac{t^2}{\tau_J}
+4M\left[2c_\times\sqrt{\sum_{i\geq1}\lambda_i}
\sqrt{\sum_{j>J}\lambda_j}+\lambda_{J+1}\right]\right\}.
\label{eq: derivative gram modulus}
\end{align}
For a finite image of dimension $D$, take $1\leq J\leq D$, with zero tail and $\lambda_{D+1}=0$ at $J=D$.
\end{lemma}

\noindent\textbf{Proof of Lemma~\ref{lem: modulus of continuity}.}
Fix $J$ and a candidate difference $u$ with $\|T_0u\|_2\leq t$, and write
$u_L=\sum_{i\leq J}a_ie_i$ and $u_H=\sum_{j>J}a_je_j$.
By~\eqref{eq: parameter spectral expansion}, monotonicity of $\lambda_j$ and weighted Cauchy--Schwarz give
\begin{align*}
\sum_{j>J}a_j^2
&\leq\lambda_{J+1}\sum_{j>J}a_j^2/\lambda_j\leq4M\lambda_{J+1},\\
\sum_{i\leq J}|a_i|
&\leq2\sqrt M\left(\sum_{i\geq1}\lambda_i\right)^{1/2},\qquad
\sum_{j>J}|a_j|\leq2\sqrt M\left(\sum_{j>J}\lambda_j\right)^{1/2}.
\end{align*}
These bounds ensure absolute convergence of the cross-block series. Using the population $L^2$ inner product, expand the derivative norm and apply both conditions in~\eqref{eq: derivative gram link}:
\begin{align*}
t^2\geq\|T_0u\|_2^2
&=\|T_0u_L\|_2^2+2\langle T_0u_L,T_0u_H\rangle_2+\|T_0u_H\|_2^2\\
&\geq\tau_J\sum_{i\leq J}a_i^2
-2c_\times\tau_J\left(\sum_{i\leq J}|a_i|\right)\left(\sum_{j>J}|a_j|\right).
\end{align*}
Thus
\[
\|u\|_2^2\leq\frac{t^2}{\tau_J}
+8c_\times M\sqrt{\sum_{i\geq1}\lambda_i}\sqrt{\sum_{j>J}\lambda_j}
+4M\lambda_{J+1}.
\]
Taking the candidate supremum and then the infimum over $J$ proves~\eqref{eq: derivative gram modulus}. The cross-block condition is what controls possible cancellation in the norm expansion. In finite dimension, $J=D$ eliminates both tails. In infinite dimension, choose $J$ to make the tails small, then take $t$ small with that $J$ fixed. At $t=0$, taking $J\to\infty$ gives zero modulus.

\QED

\begin{corollary}[Rates under polynomial and exponential derivative links]\label{cor: derivative spectral rates}
Maintain Theorems~\ref{thm: convergence rate}--\ref{thm: regret bound} and the link conditions of Lemma~\ref{lem: modulus of continuity}. Suppose $\lambda_j\leq Cj^{-2\beta_H}$ with $\beta_H>1/2$, and put $\chi=\beta_H-1/2>0$.
\begin{enumerate}[label=(\alph*)]
\item If $\tau_J\geq cJ^{-2b}$ for $b>0$, then
\[
\|\widehat\alpha_n-\alpha_0\|_2=O_p\{\delta_n^{\chi/(\chi+2b)}\},\qquad
R(\widehat\pi_n)=O_p\{\delta_n^{2\chi/(\chi+2b)}\}.
\]
\item If $\tau_J\geq c\exp(-J^b)$ for $b>0$, then
\[
\|\widehat\alpha_n-\alpha_0\|_2=O_p\{[\log(1/\delta_n)]^{-\chi/(2b)}\},\qquad
R(\widehat\pi_n)=O_p\{[\log(1/\delta_n)]^{-\chi/b}\}.
\]
\end{enumerate}
Here $R(\pi)=\mathcal V(\pi^*)-\mathcal V(\pi)$. Under the residual RKHS condition, take $\delta_n=n^{-s/(2s+d)}+r_n+\varepsilon_n+\xi_n+\sqrt{\log n/n}$. 
\end{corollary}

\noindent\textbf{Proof of Corollary~\ref{cor: derivative spectral rates}.}
In a finite image of dimension $D$, taking $J=D$ in~\eqref{eq: derivative gram modulus} gives $\omega_{M,n}^D(t)\leq t/\sqrt{\tau_D}$, so coefficient error is $O_p(\delta_n)$ and regret is $O_p(\delta_n^2)$, which imply both stated bounds. For an infinite-dimensional image,
\[
\sum_{j>J}\lambda_j\leq C\int_J^\infty x^{-2\beta_H}\,dx
=\frac{C}{2\beta_H-1}J^{1-2\beta_H}.
\]
Since $\chi=\beta_H-1/2$ and $\lambda_{J+1}\leq CJ^{-2\beta_H}\leq CJ^{-\chi}$, Lemma~\ref{lem: modulus of continuity} yields
\[
\{\omega_{M,n}^D(t)\}^2\leq C_M\inf_{J\geq1}\{t^2/\tau_J+J^{-\chi}\}.
\]
In case (a), choose $J=\lceil t^{-2/(\chi+2b)}\rceil$ for $0<t<1$. Then
\[
t^2/\tau_J+J^{-\chi}
\leq C\{t^2J^{2b}+J^{-\chi}\}
\leq C't^{2\chi/(\chi+2b)}.
\]
In case (b), for sufficiently small $t>0$, take $J=\lfloor[\log(1/t)]^{1/b}\rfloor\geq1$. This gives
\[
t^2/\tau_J+J^{-\chi}
\leq C\{t^2e^{J^b}+J^{-\chi}\}
\leq C'\{t+[\log(1/t)]^{-\chi/b}\}
\leq C''[\log(1/t)]^{-\chi/b}.
\]
Taking square roots and applying~\eqref{eq: derivative modulus convergence} at $t=C\delta_n$ proves the coefficient rates, since $\log(1/(C\delta_n))\sim\log(1/\delta_n)$. The quadratic regret bound gives their squares. 

\QED

A more directly diagonal link can improve the tail calculation. Conditions of this form bound $\|T_0u\|_2^2$ below by a weighted sum of the squared $L^2$ coefficients of $u$, as in \citet[Assumption~5.4]{chen2012estimation}.

\begin{lemma}[A full spectral link]\label{lem: full derivative spectral link}
In the setting of~\eqref{eq: parameter spectral expansion}, suppose there are $c_D>0$ and a nonincreasing positive sequence $\rho_j$ such that, for every candidate difference $u=\sum_ja_je_j$ on the fixed ball,
\begin{equation}
\|T_0u\|_2^2\geq c_D\sum_{j\geq1}\rho_j a_j^2.
\label{eq: full derivative spectral link}
\end{equation}
Then
\begin{equation}
\{\omega_{M,n}^D(t)\}^2\leq
\inf_{J\geq1}\left\{\frac{t^2}{c_D\rho_J}+4M\lambda_{J+1}\right\}.
\label{eq: diagonal derivative modulus}
\end{equation}
For a finite image of dimension $D$, restrict the infimum to $1\leq J\leq D$ and set $\lambda_{D+1}=0$.
Under the conditions of Theorems~\ref{thm: convergence rate}--\ref{thm: regret bound}, if $\lambda_j\leq Cj^{-2\beta_H}$ and $\rho_j\geq cj^{-2b}$, with $\beta_H>1/2$ and $b>0$, coefficient error is $O_p\{\delta_n^{\beta_H/(\beta_H+b)}\}$ and regret is $O_p\{\delta_n^{2\beta_H/(\beta_H+b)}\}$. If $\inf_j\rho_j>0$, then $\|u\|_2\leq C\|T_0u\|_2$.
\end{lemma}

\noindent\textbf{Proof of Lemma~\ref{lem: full derivative spectral link}.}
Fix $\alpha\in\mH_{n,M}$ with $u=\alpha-\alpha_0$ and $\|T_0u\|_2\leq t$. For every $J$, the full link and monotonicity of $\rho_j$ give
\[
c_D\rho_J\sum_{j\leq J}a_j^2
\leq c_D\sum_{j\geq1}\rho_ja_j^2
\leq\|T_0u\|_2^2\leq t^2.
\]
Combining this low-frequency bound with~\eqref{eq: parameter spectral expansion} yields
\begin{align*}
\|u\|_2^2
&=\sum_{j\leq J}a_j^2+\sum_{j>J}a_j^2\\*
&\leq\frac{t^2}{c_D\rho_J}
+\lambda_{J+1}\sum_{j>J}\frac{a_j^2}{\lambda_j}
\leq\frac{t^2}{c_D\rho_J}+4M\lambda_{J+1}.
\end{align*}
Taking the candidate supremum and then the infimum over $J$ proves~\eqref{eq: diagonal derivative modulus}. In infinite dimension, $t=0$ and $J\to\infty$ give $\omega_{M,n}^D(0)=0$, since $\lambda_{J+1}\to0$. In a finite image of dimension $D$, take $J=D$ and $\lambda_{D+1}=0$ to obtain
\[
\omega_{M,n}^D(t)\leq(c_D\rho_D)^{-1/2}t.
\]

For an infinite-dimensional image under the polynomial bounds, take $0<t<1$ and $J=\lceil t^{-1/(\beta_H+b)}\rceil$. Since
$t^{-1/(\beta_H+b)}\leq J\leq2t^{-1/(\beta_H+b)}$,
\begin{align*}
\{\omega_{M,n}^D(t)\}^2
&\leq C_M\{t^2J^{2b}+J^{-2\beta_H}\}\\*
&\leq C'_M\{t^{2-2b/(\beta_H+b)}
+t^{2\beta_H/(\beta_H+b)}\}
\leq C''_M t^{2\beta_H/(\beta_H+b)}.
\end{align*}
Taking square roots and substituting $t=C\delta_n$ in~\eqref{eq: derivative modulus convergence}, followed by the quadratic regret bound~\eqref{eq: quadratic regret}, gives
\[
\|\widehat\alpha_n-\alpha_0\|_2
=O_p\{\delta_n^{\beta_H/(\beta_H+b)}\},\qquad
R(\widehat\pi_n)
=O_p\{\delta_n^{2\beta_H/(\beta_H+b)}\}.
\]
The finite-dimensional linear bound gives the stronger rates $O_p(\delta_n)$ and $O_p(\delta_n^2)$, which also imply the stated rates. Finally, if $\rho_*:=\inf_j\rho_j>0$, then
\[
c_D\rho_*\|u\|_2^2
=c_D\rho_*\sum_ja_j^2
\leq c_D\sum_j\rho_ja_j^2
\leq\|T_0u\|_2^2,
\]
so $\omega_{M,n}^D(t)\leq(c_D\rho_*)^{-1/2}t$ directly.

\QED

Note that $\tau_J$ and $\rho_j$ cannot be read off from $\lambda_j$ or from $\nu_j(K_{\rm res};P_Z)$. They describe the derivative of the conditional moments. Here that derivative is the multiplication map $J_0(X)$, so a polynomial or exponential link is a separate restriction on how identification strength and parameter directions align. Uniform rank and the tail bound~\eqref{eq: weak rank tail} in Appendix~\ref{app: weak identification} give more directly interpretable sufficient conditions for the pricing model.

\subsection{Regret-Rate Comparisons and Root-\texorpdfstring{$n$}{n} Thresholds}
\label{app: root n proof}
Maintain Theorems~\ref{thm: convergence rate}--\ref{thm: regret bound}. For the rate comparisons below, suppose $\delta_n=O(n^{-\gamma})$, $\gamma>0$. Table~\ref{tab: identification regret rates} summarizes the power-law modulus bounds $\omega_{M,n}^D(t)\lesssim_M t^a$ established above. For the spectral rows, $\lambda_j\lesssim j^{-2\beta_H}$ with $\beta_H>1/2$, $\chi=\beta_H-1/2$, and the respective lower bounds are $\tau_J\gtrsim J^{-2b}$ and $\rho_j\gtrsim j^{-2b}$, $b>0$. Each row retains the hypotheses of the corresponding lemma or corollary, including the common-RKHS norm condition for the spectral links.

\begin{table}[H]
\centering
\caption{Regret upper bounds under the respective identification conditions.}
\label{tab: identification regret rates}
\small
\renewcommand{\arraystretch}{1.3}
\setlength{\tabcolsep}{5pt}
\begin{tabular}{@{}lccc@{}}
\toprule
Identification condition & Exponent $a$ & Regret bound & Sufficient for root-$n$ \\
\midrule
Uniform rank~\eqref{eq: uniform price rank}
& $1$ & $O_p(n^{-2\gamma})$ & $\gamma\geq1/4$ \\
Small singular values~\eqref{eq: weak rank tail}
& $\nu/(\nu+2)$ & $O_p(n^{-2\gamma\nu/(\nu+2)})$ & $\gamma\nu/(\nu+2)\geq1/4$ \\
Gram link~\eqref{eq: derivative gram link}
& $\chi/(\chi+2b)$ & $O_p(n^{-2\gamma\chi/(\chi+2b)})$ & $\gamma\chi/(\chi+2b)\geq1/4$ \\
Full spectral link~\eqref{eq: full derivative spectral link}
& $\beta_H/(\beta_H+b)$ & $O_p(n^{-2\gamma\beta_H/(\beta_H+b)})$ & $\gamma\beta_H/(\beta_H+b)\geq1/4$ \\
\bottomrule
\end{tabular}
\end{table}

\subsubsection{Proof of Remark~\ref{rem: root n regret}.}
On $\mathcal E_n^{\rm conv}$ from~\eqref{eq: convergence common event}, Theorem~\ref{thm: convergence rate} and the quadratic regret bound give, with $C_R=(1+4p_2^2)/(4c)$,
\begin{equation}
0\leq R(\widehat\pi_n)\leq C_R\{\omega_{M,n}^D(C\delta_n)\}^2.
\label{eq: root n proof event}
\end{equation}
Since $\Pr(\mathcal E_n^{\rm conv})\to1$, the condition $\omega_{M,n}^D(C_0\delta_n)=O(n^{-1/4})$ for every fixed $C_0>0$ implies coefficient error $O_p(n^{-1/4})$ and regret $O_p(n^{-1/2})$. Replacing $O$ by $o$ gives the corresponding little-$o_p$ conclusions. Under the stated modulus scaling condition, it suffices to check the modulus at $\delta_n$.
For a power-law link,~\eqref{eq: root n proof event} gives
\[
R(\widehat\pi_n)=O_p(\delta_n^{2a})=O_p(n^{-2\gamma a}),
\]
which proves the table and its sufficient threshold $\gamma a\geq1/4$. Strict inequality gives $o_p(n^{-1/2})$ regret.

Under uniform rank, $a=1$ and $\sqrt{\log n/n}=o(n^{-1/4})$, so $b_n,r_n,\varepsilon_n,\xi_n=O(n^{-1/4})$ suffice. In the residual RKHS setting,
\[
\frac{s}{2s+d}-\frac14=\frac{2s-d}{4(2s+d)}>0,
\]
so $r_n+\varepsilon_n+\xi_n=O(n^{-1/4})$ suffices for the main-text conclusion. More generally, when these three errors are $O(n^{-s/(2s+d)})$, take $\gamma=s/(2s+d)$ in the table. Under condition~\eqref{eq: weak rank tail}, the threshold simplifies to
\[
(4\gamma-1)\nu\geq2,
\qquad\text{or}\qquad
\nu\geq\frac{2(2s+d)}{2s-d}
\quad\text{at }\gamma=\frac{s}{2s+d}.
\]
These are sufficient upper-rate conditions. The parameter smoothness $\beta_H$ and residual-complexity smoothness $s$ describe different spectra.
\QED

\section{Approximation and Growing Model Classes}
\label{app: approximation extension}

The main theorem uses an approximating function with bounded norm and uniformly bounded adversary test norms on fixed parameter balls. We distinguish finite-sieve approximation under these norm bounds from growing norm bounds and fixed-class misspecification. The structural moments satisfy $m(\cdot;\alpha_0)=0$ throughout.

\subsection{Fixed-Norm Approximation and Feasible Sieves}
Let $\mathscr H$ be a Hilbert RKHS with continuous inclusion $\iota:\mathscr H\to L^2(P_X)^5$. Suppose $\|g_k\|_{\mathscr H}\leq R$ and $\|\iota g_k-\alpha_0\|_2\to0$. Hilbert spaces are reflexive, so every bounded sequence has a weakly convergent subsequence. That is, there are indices $k_j\uparrow\infty$ and $g\in\mathscr H$ such that
\[
g_{k_j}\rightharpoonup g\quad\text{in }\mathscr H,
\qquad\text{i.e.,}\qquad
\langle g_{k_j},h\rangle_{\mathscr H}\to\langle g,h\rangle_{\mathscr H}
\quad\text{for every }h\in\mathscr H.
\]
For each $v\in L^2(P_X)^5$, the map $h\mapsto\langle\iota h,v\rangle_2$ is a continuous linear functional on $\mathscr H$. Hence weak convergence in $\mathscr H$ and the assumed strong $L^2$ convergence give
$\langle\iota g,v\rangle_2=\lim_j\langle\iota g_{k_j},v\rangle_2=\langle\alpha_0,v\rangle_2$.
Thus $\iota g=\alpha_0$ in $L^2$, or equivalently $P_X$-almost surely. Weak lower semicontinuity of the Hilbert norm means
\[
\|g\|_{\mathscr H}\leq\liminf_{j\to\infty}\|g_{k_j}\|_{\mathscr H}\leq R.
\]
Consequently, the $L^2$ limit has an RKHS representative in the same radius-$R$ ball. A fixed bound on the RKHS norm therefore does not permit approximation of a truth outside the RKHS image in $L^2$.

For PRINT candidates with common output bounds and bounded truth, Lemma~\ref{lem: residual derivative link} converts vanishing moment error into vanishing derivative error. Under almost-sure rank, Lemma~\ref{lem: qualitative derivative stability} then gives the $L^2$ convergence used above. Finite sieves may nevertheless omit a truth in their common parameter RKHS at every $n$.

\begin{example}[Curvature-preserving customer-segment approximation]
\label{ex: segment approximation}
Suppose $X\in[0,1]^d$, the five coordinates of $\alpha_0$ are bounded and H\"older continuous of order $0<\beta\leq1$ with common constant $L$, and $-\beta_{p,2}\geq c$. Partition the cube into $K=k^d$ equal cells. Let $\mH_K$ consist of vector-valued functions constant on each cell, with a common output bound and $-b_2\geq c_{p,2}$, where $0<c_{p,2}\leq c$, and inherit the vector supremum norm. Coordinatewise Lebesgue cell averages $\alpha_K^\dagger$ of $\alpha_0$ preserve both restrictions. Each cell has diameter $\sqrt d/k$, so under any law of $X$,
\[
\|\alpha_K^\dagger-\alpha_0\|_2
\leq\sqrt5\,L\,d^{\beta/2}K^{-\beta/d},\qquad
\|\alpha_K^\dagger\|_{\mH_K}\leq\|\alpha_0\|_\infty.
\]
Boundedness of $J_0$ and the residual--derivative comparison then give $\|m(\cdot;\alpha_K^\dagger)\|_2\leq C K^{-\beta/d}$. Unless the truth is constant on the cells, it lies outside the sieve. This verifies the approximation and feasibility requirements of Assumption~\ref{ass: identification and spaces}(a) for finite customer segmentation. 
\end{example}

\subsection{Error Bounds with Growing Parameter and Adversary Norms}
The next result allows growing bounds $M_n$ and $L_n$ on the squared parameter norms and adversary test norms, respectively, with empirical complexity controlled on the prespecified candidate and adversary classes.

\begin{proposition}[Approximation bounds with growing parameter and adversary norms]
\label{prop: capacity oracle}

Assume the hypotheses of Corollary~\ref{cor: identification 2} (including its determinant proviso) and Lemma~\ref{lm: conditional moment restriction}, and the process regularity in Appendix~\ref{app: process bridge}. Retain Assumptions~\ref{ass: identification and spaces}(c) and~\ref{ass: parameters}(b), with a data-independent fold partition and complementary-fold training. Replace the remaining parts by the following primed versions. 
\begin{enumerate}[label={},leftmargin=0pt]
\item \textbf{Assumption~\ref{ass: identification and spaces}(a$'$).}
Retain boundedness of $(Y,P,G)$ and $\alpha_0$, but replace the condition on the approximating function by a nonempty prescribed set $\mathcal A_n\subseteq\mH_n$ with
$\sup_{\alpha\in\mathcal A_n}\|\alpha\|_{\mH_n}^2\leq M_n<\infty$, where $M_n$ may diverge. Require common output envelopes for $\mathcal A_n\cup\mH_{n,M_{\rm ref}}$ and $\mF_{n,B}$, for fixed $B>0$.

\item \textbf{Assumption~\ref{ass: identification and spaces}(b$'$).}
Retain star-shapedness and norm homogeneity, but replace the fixed residual-test bound $L_M$ by deterministic $L_n<\infty$ and allow any deterministic $\xi_n\geq0$: every $\alpha\in\mathcal A_n$ admits $g_\alpha\in\mF_n$ with
\begin{equation}
\|g_\alpha-m(\cdot;\alpha)\|_2\leq\xi_n,\qquad
\|g_\alpha\|_{\mF_n}\leq L_n.
\label{eq: approximate residual tests}
\end{equation}

\item \textbf{Assumption~\ref{ass: parameters}(a$'$).}
Replace the fixed parameter ball in the residual class by $\mathcal A_n\cup\mH_{n,M_{\rm ref}}$. Require the critical radii of this class and the eight adversary coordinate classes to be at most $Cb_n^\circ\to0$, with $C$ independent of $n,M_n,L_n$.

\item \textbf{Assumption~\ref{ass: parameters}(c$'$).}
Put $a_n^\circ=b_n^\circ+r_n+\sqrt{\log n/n}$ and choose $\lambda_n\asymp(a_n^\circ)^2$ with a sufficiently large lower constant. Allow arbitrary deterministic $\mu_n,\eta_n\geq0$. Retain the fixed-reference objective and approximate-optimality inequality, with $\widehat\alpha_n\in\mathcal A_n$ and the infimum restricted to $\mathcal A_n$.
\end{enumerate}
Define the penalized approximation error and total error bound by
\begin{align*}
\mathfrak A_n(\mu_n)
&=\inf_{\alpha\in\mathcal A_n}
\{\|m(\cdot;\alpha)\|_2^2+\mu_n\|\alpha\|_{\mH_n}^2\},
\\
E_n^2&=(a_n^\circ)^2(1+L_n^2)+\xi_n^2+\eta_n+\mathfrak A_n(\mu_n).
\end{align*}
There is a constant $C$ independent of $n,M_n,L_n$ such that, with probability tending to one,
\begin{equation}
\|m(\cdot;\widehat\alpha_n)\|_2+
\|T_0(\widehat\alpha_n-\alpha_0)\|_2\leq C E_n.
\label{eq: capacity oracle rate}
\end{equation}
Under uniform price-moment rank~\eqref{eq: uniform price rank}, coefficient error is $O_p(E_n)$. If additionally $-\beta_{p,2}(X)\geq c>0$ and every candidate has $-b_2(X)\geq c_{p,2}>0$, the policy in~\eqref{eqn: model-based optimal policy} has regret $O_p(E_n^2)$ by Theorem~\ref{thm: regret bound}. If $E_n\to0$, almost-sure rank and the common output envelopes imply coefficient consistency. Under the tail bound~\eqref{eq: weak rank tail}, coefficient error is $O_p\{E_n^{\nu/(\nu+2)}\}$ and, with the same curvature conditions, regret is $O_p\{E_n^{2\nu/(\nu+2)}\}$ when $E_n\to0$.

\end{proposition}

\noindent\textbf{Proof of Proposition~\ref{prop: capacity oracle}.}

The common output envelopes, critical-radius bounds, reference condition, and foldwise nuisance assumptions allow the proof of Lemma~\ref{lem: first stage stability} to be applied directly to $\mathcal A_n$ and the reference ball. Thus~\eqref{eq: feasible linear}--\eqref{eq: feasible upper} hold uniformly on $\mathcal A_n$ with $a_n^\circ$ replacing $a_n$ and constants independent of $M_n,L_n$. Work on this event intersected with the assumed optimization event, whose probability tends to one. The estimator belongs to $\mathcal A_n$ by construction, so no radius localization is needed.

Write $J_n$ for the adversary supremum defined in Appendix~\ref{app: convergence proof}. The two scalar maximizations there, with $\|m(\cdot;\alpha)\|_2$ replacing $\varepsilon_n$, and approximate optimality give
\begin{align*}
0\leq J_n(\alpha)&\leq C\{\|m(\cdot;\alpha)\|_2^2+(a_n^\circ)^2\},
\qquad\alpha\in\mathcal A_n,\\
J_n(\widehat\alpha_n)&\leq C'\{\mathfrak A_n(\mu_n)+(a_n^\circ)^2+\eta_n\}.
\end{align*}
Uniformity permits taking the infimum over $\mathcal A_n$ even when it is not attained.

For $u=\|m(\cdot;\alpha)\|_2$, choose $g=g_\alpha$ from~\eqref{eq: approximate residual tests} and $f=tg$, with fixed $0<t\leq1$ satisfying $2C_Qt\leq1/4$, where $C_Q$ is the constant in~\eqref{eq: feasible upper}. Then $f\in\mF_n$, $s_f\leq\max\{1,tL_n/\sqrt B\}$, $\|g\|_2\leq u+\xi_n$, and $\EE[m^\top g]\geq u^2-u\xi_n$. The feasible bounds yield
\begin{align*}
J_n(\alpha)\geq{}&t(u^2-u\xi_n)-Ca_n^\circ t(u+\xi_n)
-C(a_n^\circ)^2s_f\\
&-C_Qt^2(u+\xi_n)^2-C_Q(a_n^\circ)^2s_f^2
-\lambda_nt^2L_n^2\\
\geq{}&c\,u^2-C\{\xi_n^2+(a_n^\circ)^2(1+L_n^2)\},
\end{align*}
where the last line uses Young's inequality, $(u+\xi_n)^2\leq2u^2+2\xi_n^2$, and $\lambda_n\lesssim(a_n^\circ)^2$. Its constants are independent of $n,M_n,L_n$. Combining the upper and lower bounds at $\widehat\alpha_n$ gives $\|m(\cdot;\widehat\alpha_n)\|_2^2\leq CE_n^2$. The common output envelopes make the residual--derivative comparison uniform, proving~\eqref{eq: capacity oracle rate}. The pointwise inverses in Lemmas~\ref{lem: qualitative derivative stability}--\ref{lem: price moment stability}, the probability bound in Corollary~\ref{cor: weak rank}, and the quadratic regret bound give the remaining conclusions.

\QED

\begin{remark}[A sufficient condition for consistency]
For $\epsilon_n^2=\inf_{\alpha\in\mathcal A_n}\|m(\cdot;\alpha)\|_2^2$, an approximating sequence gives $\mathfrak A_n(\mu_n)\leq\epsilon_n^2+\mu_nM_n$. Hence a sufficient consistency condition is
\[
a_n^\circ(1+L_n)+\xi_n+\epsilon_n+
\sqrt{\mu_nM_n}+\sqrt{\eta_n}\longrightarrow0.
\]
This bound does not require the analyst to know $\epsilon_n$. A fixed constrained class with bounded exact residual tests gives a residual-error upper bound of order $a_n^\circ+\epsilon_n+\sqrt{\mu_nM_n}+\sqrt{\eta_n}$. A persistent $\epsilon_n$ prevents a vanishing-error guarantee, but is not a revenue-loss lower bound: different coefficients can induce the same constrained price. 
\end{remark}

\begin{remark}[Spectral bounds when the truth lies outside the RKHS]
Suppose $\mathcal A_n$ lies in a common RKHS with nonincreasing eigenvalues $\lambda_j$, $L^2$-orthonormal eigenfunctions $e_j$, and squared norm bound $M_n$. Write $\Pi_J$ for projection onto $\operatorname{span}\{e_1,\ldots,e_J\}$ and $q_{0,J}=\|(I-\Pi_J)\alpha_0\|_2$. The truth may lie outside the RKHS or the closure of its image. Assume, for the actual differences $u=\alpha-\alpha_0$, $\alpha\in\mathcal A_n$, the additional link
\[
\|T_0u\|_2^2\geq c_D\sum_{j\geq1}\rho_j\langle u,e_j\rangle_2^2,
\]
with common $c_D>0$ and nonincreasing positive $\rho_j$. Low-frequency inversion and the RKHS tail bound, followed by the triangle inequality and orthogonality, give
\begin{equation*}
\|u\|_2^2\leq\inf_{J\geq1}
\left\{\frac{\|T_0u\|_2^2}{c_D\rho_J}
 +\left(\sqrt{M_n\lambda_{J+1}}+q_{0,J}\right)^2\right\}.
\end{equation*}
For $u=\widehat\alpha_n-\alpha_0$, substitute $\|T_0u\|_2\leq CE_n$ on the event in Proposition~\ref{prop: capacity oracle}, and apply Theorem~\ref{thm: regret bound} for regret. In finite dimension $D$, restrict to $J\leq D$ and set $\lambda_{D+1}=0$, retaining $q_{0,D}$. No rate for $q_{0,J}$ follows from the kernel spectrum. If the truth is outside the image closure, its positive limiting value prevents coefficient consistency. If both truth and candidates have norm at most $\sqrt M$, then $q_{0,J}\leq\sqrt{M\lambda_{J+1}}$, recovering the $4M\lambda_{J+1}$ tail in Lemma~\ref{lem: full derivative spectral link}.
\end{remark}
\end{document}